\documentclass[letterpaper,10pt]{article}
\usepackage{fullpage}
\usepackage{amsmath,amssymb,amsthm}
\usepackage{graphicx,url,footmisc}
\usepackage{subfig}
\usepackage{enumerate}
\usepackage{authblk}
\usepackage{natbib}

\date{}
\bibpunct{(}{)}{;}{a}{,}{,}

\newtheorem{theorem}{Theorem}
\newtheorem{lemma}{Lemma}

\newtheorem{conjecture}{Conjecture}
\newtheorem{observation}{Observation}
\newtheorem{proposition}{Proposition}
\theoremstyle{definition}
\newtheorem{definition}{Definition}
\newtheorem{remark}{Remark}

\newcommand{\R}{\mathbb{R}}

\newcommand{\E}{\mathbb{E}}
\newcommand{\bA}{\mathbf{A}}
\newcommand{\bB}{\mathbf{B}}
\newcommand{\bD}{\mathbf{D}}
\newcommand{\bE}{\mathbf{E}}
\newcommand{\bJ}{\mathbf{J}}
\newcommand{\bH}{\mathbf{H}}
\newcommand{\bI}{\mathbf{I}}

\newcommand{\bM}{\mathbf{M}}
\newcommand{\bN}{\mathbf{N}}

\newcommand{\bP}{\mathbf{P}}
\newcommand{\bR}{\mathbf{R}}
\newcommand{\bS}{\mathbf{S}}
\newcommand{\bT}{\mathbf{T}}
\newcommand{\bQ}{\mathbf{Q}}
\newcommand{\bU}{\mathbf{U}}
\newcommand{\bV}{\mathbf{V}}
\newcommand{\bW}{\mathbf{W}}
\newcommand{\bX}{\mathbf{X}}
\newcommand{\bY}{\mathbf{Y}}
\newcommand{\bZ}{\mathbf{Z}}

\newcommand{\be}{\mathbf{e}}
\newcommand{\bv}{\mathbf{v}}
\newcommand{\bx}{\mathbf{x}}
\newcommand{\by}{\mathbf{y}}

\newcommand{\calL}{\mathcal{L}}
\newcommand{\calN}{\mathcal{N}}
\newcommand{\calO}{\mathcal{O}}

\newcommand{\Xhat}{\hat{\bX}}
\newcommand{\dhat}{\hat{d}}

\newcommand{\Yhat}{\hat{\bY}}
\newcommand{\Zhat}{\hat{\bZ}}

\newcommand{\zeromx}{\mathbf{0}}

\newcommand{\supp}{\operatorname{supp}}

\newcommand{\Ptilde}{\tilde{\bP}}

\newcommand{\Wtilde}{\tilde{\bW}}
\newcommand{\Wn}{\bW_n}
\newcommand{\Wntilde}{\Wtilde_n}
\newcommand{\Va}{\bV_1}
\newcommand{\Vb}{\bV_2}

\newcommand{\Xstar}{\bX^*}

\newcommand{\Zstar}{\bZ^*}

\newcommand{\tti}{2 \rightarrow \infty}
\newcommand{\Abar}{\bar{\bA}}
\newcommand{\Xbar}{\bar{\bX}}

\newcommand{\balpha}{\mathbf{\alpha}}
\newcommand{\bDelta}{\mathbf{\Delta}}
\newcommand{\bSigma}{\mathbf{\Sigma}}
\newcommand{\Sigmatilde}{\tilde{\bSigma}}

\newcommand{\alphavec}{\vec{\balpha}}

\newcommand{\UM}{\bU_{\bM}}
\newcommand{\UA}{\bU_{\bA}}

\newcommand{\SM}{\bS_{\bM}}
\newcommand{\SA}{\bS_{\bA}}
\newcommand{\UPt}{\bU_{\Ptilde}}
\newcommand{\SPt}{\bS_{\Ptilde}}
\newcommand{\UP}{\bU_{\bP}}
\newcommand{\SP}{\bS_{\bP}}

\newcommand{\ASE}{\operatorname{ASE}}
\newcommand{\Dir}{\operatorname{Dir}}
\newcommand{\OMNI}{\operatorname{OMNI}}

\newcommand{\RDPG}{\operatorname{RDPG}}
\newcommand{\JRDPG}{\operatorname{JRDPG}}

\newcommand{\Bern}{\operatorname{Bernoulli}}
\newcommand{\Var}{\operatorname{Var}}

\newcommand{\inlaw}{\xrightarrow{\calL}}
\newcommand{\inprob}{\xrightarrow{P}}
\newcommand{\iid}{\stackrel{\text{i.i.d.}}{\sim}}

\newcommand{\mxE}{\mathbf{E}}
\newcommand{\mxD}{\mathbf{D}}
\newcommand{\ErdosRenyi}{Erd\H{o}s-R\'{e}nyi }

\begin{document}

%\title{\bf A Central Limit Theorem for an Omnibus Embedding of Random Dot Product Graphs}
\title{A Central Limit Theorem for an Omnibus Embedding of Multiple Random Graphs and Implications for Multiscale Network Inference}
\author[1]{Keith Levin}
\author[2]{Avanti Athreya}
\author[2]{Minh Tang}
\author[3]{\\Vince Lyzinski}
\author[2]{Youngser Park}
\author[2]{Carey E. Priebe}

% \author[2]{Minh Tang}\and
% \author[3]{Vince Lyzinski}\and
% \author[2]{Youngser Park}\and
% \author[2]{Carey E. Priebe}
 \affil[1]{University of Michigan, Ann Arbor, MI}
 \affil[2]{Johns Hopkins University, Baltimore, MD}
 \affil[3]{University of Massachusetts, Amherst, MA}
%\affil{University of Michigan, Ann Arbor, MI\\ Johns Hopkins University, Baltimore, MD \\ and University of Massachusetts, Amherst, MA}
%\renewcommand\Authands{ and }
%\thanks{This research is partly sponsored by the Air Force Research Laboratory and DARPA, under agreement number FA8750-18-2-0035; as well as DARPA, under agreement numbers FA8750-12-2-0303, N66001-14-1-4028 and N66001-15-C-4041. The views and conclusions contained herein are those of the authors and should not be interpreted as necessarily representing the official policies or endorsements, either expressed or implied, of the Air Force Research Laboratory and DARPA, or the U.S. Government. The U.S. Government is authorized to reproduce and distribute reprints for Governmental purposes notwithstanding any copyright notation thereon. The authors also gratefully acknowledge the support of NSF grant DMS-1646108 and NIH grant BRAIN U01-NS108637.}
%}\hspace{.2cm}\\
%    and \\
%    Avanti Athreya \\
%    and \\
%    Minh Tang \\
%    and \\
%    Vince Lyzinski \\
%    and \\
%    Carey E. Priebe \\
%\begin{document}
  \maketitle
%} \fi
%\if0\blind
%{
%  \bigskip
%  \bigskip
%  \bigskip
%  \begin{center}
%    {\LARGE\bf A Central Limit Theorem for an Omnibus Embedding of Random Dot Product Graphs}
%\end{center}
%  \medskip
%} \fi

%\bigskip
\begin{abstract}
Performing statistical analyses on collections of graphs is of import to many disciplines, but principled, scalable methods for multi-sample graph inference are few. Here we describe an ``omnibus" embedding in which multiple graphs on the same vertex set are jointly embedded into a single space with a distinct representation for each graph. We prove a central limit theorem for this embedding and demonstrate how it streamlines graph comparison, obviating the need for pairwise subspace alignments. The omnibus embedding achieves near-optimal inference accuracy when graphs arise from a common distribution and yet retains discriminatory power as a test procedure for the comparison of different graphs. Moreover, this joint embedding and the accompanying central limit theorem are important for answering multiscale graph inference questions, such as the identification of specific subgraphs or vertices responsible for similarity or difference across networks. We illustrate this with a pair of analyses of connectome data derived from dMRI and fMRI scans of human subjects. In particular, we show that this embedding allows the identification of specific brain regions associated with population-level differences. Finally, we sketch how the omnibus embedding can be used to address pressing open problems, both theoretical and practical, in multisample graph inference.  
\end{abstract}
\noindent%
{\it Keywords:}  %3 to 6 keywords, that do not appear in the title
multiscale graph inference, graph embedding, multiple-graph hypothesis testing 
%\vfill

%\newpage
%\spacingset{1.45} % DON'T change the spacing!
\section{Introduction}
\label{sec:intro}

Statistical inference across multiple graphs is of vital interdisciplinary interest in domains as varied as machine learning, neuroscience, and epidemiology.
Inference on random graphs frequently depends on appropriate low-dimensional Euclidean representations of the vertices of these graphs, known as {\em graph embeddings}, typically given by spectral decompositions of adjacency or Laplacian matrices \citep{belkin03:_laplac, STFP-2011, chatterjee_usvt}. Nonetheless, while spectral methods for parametric inference in a single graph are well-studied, multi-sample graph inference is a nascent field. See, for example, the authors' work in \citep{tang14:_semipar,tang14:_nonpar} as among the only principled approaches to two-sample graph testing. What is more, for inference tasks involving multiple graphs---for instance, determining whether two or more graphs on the same vertex set are similar---discerning an optimal simultaneous embedding for all graphs is a challenge: how can such an embedding be structured to both provide estimation accuracy of common parameters when the graphs are similar, but retain discriminatory power when they are different?

A flexible, robust embedding procedure to achieve both goals would be of considerable utility in a range of real data applications. For instance, consider the problem of community detection in large networks. While algorithms for community detection abound, relatively few approaches exist to address community {\em classification}; that is, to leverage graph structure to successfully establish which subcommunities appear statistically similar or different. In fact, the authors' work in \cite{lyzinski15_HSBM} represents one of the earliest forays into statistically principled techniques for subgraph classification in hierarchical networks. Not surprisingly, the creation of a graph-statistical analogue of the classical analysis-of-variance $F$-test, in which a single test procedure would permit the comparison of graphs from multiple populations, is very much an open problem of immediate import. But even given a coherent framework for extracting graph-level differences across multiple populations of graphs, there remains the further complication of replicating this at multiple scales, by isolating---in the spirit of post-hoc tests such as Tukey's studentized range---precisely which subgraphs or vertices in a collection of vertex-matched graphs might be most similar or different. 

Our goal in this paper, then, is to provide a unified framework to answer the following questions:
\begin{enumerate}[(i)]
\item Given a collection of random graphs, can we develop a single statistical procedure that accurately estimates common underlying graph parameters, in the case when these parameters are equal across graphs, but also delivers meaningful power for testing when these graph parameters are distinct? 
\item Can we develop an inference procedure that identifies sources of graph similarity or difference at scales ranging from whole-graph to subgraph to vertex? For example, can we identify particular vertices that contribute significantly to statistical differences at the whole-graph level?
\item Can we develop an inference procedure that scales well to large graphs, addresses graphs that are weighted, directed, or whose edge information is corrupted, and which is amenable to downstream classical statistical methodology for Euclidean data?
\item Does such a statistical procedure compare favorably to existing state-of-the-art techniques for joint graph estimation and testing, and does it work well on real data?
\end{enumerate}
Here, we address each of these open problems with a single embedding procedure. Specifically, we describe an {\em omnibus} embedding, in which the adjacency matrices of multiple graphs on the same vertex-matched set are jointly embedded into a single space {\em with a distinct representation for each graph and, indeed, each vertex of each graph}. We then prove a central limit theorem for this embedding, a limit theorem similar in spirit to, but requiring a significantly more delicate probabilistic analysis than, the one proved in \cite{athreya2013limit}. We show, in both simulated and real data, that the asymptotic normality of these embedded vertices has demonstrable utility. First, the omnibus embedding performs nearly optimally for the recovery of graph parameters when the graphs are from the same distribution, but compares favorably with state-of-the-art hypothesis testing procedures to discern whether graphs are different. Second, the simultaneous embedding into a shared space
allows for the comparison of graphs without the need to
perform pairwise alignments of the embeddings of different graphs.
Third, the asymptotic normality of the omnibus embedding permits the application of a wide array of subsequent Euclidean inference techniques, most notably a multivariate analysis of variance (MANOVA) to isolate statistically significant {\em vertices} across several graphs. Thus, the omnibus embedding provides a statistically sound analogue of a post-hoc Tukey test for multisample graph inference. We demonstrate this with an analysis of real data, comparing a collection of magnetic resonance imaging (MRI) scans of human brains to identify dissimilar graphs and then to further pinpoint specific intra-graph features that account for global graph differences. 

The main theoretical results of this paper are a consistency theorem for the omnibus embedding, akin to \cite{lyzinski13:_perfec}, and a central limit theorem, akin to \cite{athreya2013limit}, for the distribution of any finite collection of rows of this omnibus embedding. 
We emphasize that distributional results for spectral decompositions of random graphs are few. The classic results of
\cite{furedi1981eigenvalues} describe the eigenvalues of the \ErdosRenyi
random graph and the work of \cite{tao2012random} concerns distributions
of eigenvectors of more general random matrices under moment restrictions,
but \cite{athreya2013limit} and \cite{tang_lse} are among the only references for central limit theorems for spectral decompositions of adjacency and Laplacian matrices for a class of independent-edge random graphs
broader than the \ErdosRenyi model.

Our consistency result shows that the omnibus embedding provides consistent estimates of certain unobserved vectors, called latent positions, that are associated to vertices of the graphs.
At present, the best available spectral estimates
of such latent positions involve averaging across graphs followed by
an embedding, resulting in a single set of estimated latent positions,
rather than in a {\em distinct set for each graph}.
We find in simulations, that our omnibus-derived estimates perform competitively
with these existing spectral estimates of the latent positions,
while still retaining graph-specific information.
In addition, we show that the omnibus embedding allows for a test statistic that improves on the state-of-the-art two-sample test procedure
presented in \cite{tang14:_semipar} for determining whether two
random dot product graphs \citep{young2007random}
have the same latent positions.

% Multiple-graph inference is a nascent field, and the existing literature on even two-sample graph inference is scarce. We point to \cite{tang14:_semipar} and \cite{tang14:_nonpar}, which provide theoretical and empirical results for testing whether two random dot product graphs are statistically similar or not, as among very few principled methodologies currently available for testing network similarity.
% The first of these papers on two-sample graph testing, 
Specifically, \cite{tang14:_semipar} introduces a test statistic generated by performing a Euclidean, lower-dimensional embedding of the graph adjacency matrix \citep[see][]{STFP-2011} of each of the two networks, followed by a Procrustes alignment \citep{gower_procrustes,dryden_mardia_shape} of the two embeddings. Addressing the nonparametric analogue of this question---whether two graphs have the same latent position {\em distribution}---is the focus of \cite{tang14:_nonpar}, which uses the embeddings of each graph to estimate associated density functions.
The Procrustes alignment required by \cite{tang14:_nonpar} both complicates the test statistic and, empirically, weakens the power of the test (see Section \ref{sec:expts}).  Furthermore, it is unclear how to effectively adapt pairwise Procrustes alignments to tests involving more than two graphs. The omnibus embedding allows us to avoid these issues altogether by providing a multiple-graph representation that is well-suited to both latent position estimation and comparative graph inference methods.
%, from two-sample hypothesis testing to joint graph clustering \citep[see, for example, an earlier connectomic analysis in][]{chen_worm}.

Our paper is organized as follows. In Sec.~\ref{sec:summary_real_data}, we give an overview of our main results and present two real-data examples in which the omnibus embedding uncovers important multiscale information in collections of brain networks. In Sec.~\ref{sec:formal_definitions}, we present background information and formal definitions. In Sec.~\ref{sec:main_results}, we provide detailed statements of our principal theoretical results. In Sec.~\ref{sec:expts}, we present simulation data that illustrates the power of the omnibus embedding as a tool for both estimation and testing. In our Supplementary Material, we provide detailed proofs, including a sharpening of a vertex-exchangeability argument for bounding residual terms in the difference between omnibus estimates and true graph parameter values.  We conclude with a discussion of extensions and open problems on multi-sample graph inference.

\section{Summary of main results and applications to real data}\label{sec:summary_real_data}

Recall that our goal is to develop a single spectral embedding technique for multiple-graph samples that (a) estimates common graph parameters, (b) retains discriminatory power for multisample graph hypothesis testing and (c) allows for a principled approach to identifying specific vertices that drive graph similarities or differences. In this section, we give an informal description of the omnibus embedding and a pared-down statement of our central limit theorem for this embedding, keeping notation to a minimum. We then demonstrate immediate payoffs in exploratory data analysis, leaving a more detailed technical descriptions of the method and our results for later sections.

To provide a theoretically-principled paradigm for graph inference for stochastically-varying networks, we focus on a particular class of random graphs. 
%We delineate our graph models here, beginning with the random dot product graph (RDPG), the stochastic blockmodel(SBM), and the latent structure model (LSM).
% The stochastic blockmodel graph \cite{holland} is currently one of the most popular generative models for random graphs; graphs generated from this model have an intrinsic community structure, because each vertex is assigned to a so-called {\em block} and the probability that two vertices are incident depends only on their block assignments.
%\begin{definition}%[$K$-block stochastic blockmodel]	
% We use graph adjacency matrices as compact representations of connections between graphs; $\bH$ represents a matrix and $H_i$ its $i$th row. Euclidean distance is denoted $| \cdot |$; ${\top}$ denotes tranpose and $\perp$ denotes the orthogonal complement. Probability and expectation are indicated by $P$ and $E$.
	We define a {\em graph} $G$ to be an ordered pair of $(V,\mathcal{E})$ where $V$ is the {\em vertex} or {\em node} set, and $\mathcal{E}$, the set of {\em edges}, is a subset of the Cartesian product of $V \times V$. In a graph whose vertex set has cardinality $n$, we will usually represent $V$ as $V=\{1, 2, \dots, n\}$, and we say there {\em is an edge between} $i$ and $j$ if $(i,j)\in \mathcal{E}$.  The {\em adjacency} matrix $\bA$ provides a  representation of such a graph:
	$$\bA_{ij}=1 \textrm{   if   }(i,j) \in \mathcal{E}, \textrm{  and  }\bA_{ij}=0 \textrm{ otherwise. }$$ 
	Where there is no danger of confusion, we will often refer to a graph $G$ and its adjacency matrix $\bA$
	interchangeably.
	
	Any model of a stochastic network must describe the probabilistic mechanism of connections between vertices. We focus on a class of {\em latent position} random graphs \cite{Hoff2002,diaconis2007graph,asta_cls},in which every vertex has associated to it a (typically unobserved) {\em latent position}, itself an object belonging to some (often Euclidean) space $\mathcal{X}$.  Probabilities of an edge between two vertices $i$ and $j$, $p_{ij}$, are a function $\kappa(\cdot,\cdot): \mathcal{X} \times \mathcal{X} \rightarrow [0,1]$ (known as the {\em link function}) of their associated latent positions $(x_i, x_j)$. Thus $p_{ij}=\kappa(x_i, x_j)$, and given these probabilities, the entries $\bA_{ij}$ of the adjacency matrix $\bA$ are independent Bernoulli random variables with success probabilities $p_{ij}$. We consolidate these probabilities into a matrix $\bP=(p_{ij})$, and write $\bA \sim \bP$ to denote this relationship. 
	
	The latent position graph model has tremendous utility in modeling natural phenomena. For example, individuals in a disease network may have a hidden vector of attributes (prior illness, high-risk occupation) that observers of disease dynamics do not see, but which nevertheless strongly influence the chance that such an individual may become ill or infect others. Because the link function is relatively unrestricted, latent position models can replicate a wide array of graph phenomena \citep{olhede_wolfe_histogram}. 
		
		In a $d$-dimensional {\em random dot product} graph \citep{young2007random}, the latent space is an appropriately-constrained subspace of $\mathbb{R}^d$, and the link function is simply the dot product of the two latent $d$-dimensional vectors. The invariance of the inner product to orthogonal transformations is a nonidentifiability in the model, so we frequently specify accuracy up to a rotation matrix $\bW$.
%		Matrix spectral decompositions of the adjacency matrix can be deployed for accurate latent position estimation, community detection, and hypothesis testing for random dot product graphs; for a comprehensive summary of these techniques, see \cite{athreya_survey}.
%		Because of the invariance of the inner product to orthogonal transformations latent positions can be estimated only up to an orthogonal transformation. 
%		The popular {\em stochastic blockmodel} (SBM) can be regarded as a random dot product graph. In an SBM, there are a finite number of possible latent positions for each vertex, one for each block, and the latent position exactly determines the block assignment for that vertex.		
		Random dot product graphs are often divided into two types: those in which the latent positions are fixed, and those in which the latent positions are themselves random. We will address both cases here: in our theoretical results, the latent positions $X_i \in \mathbb{R}^d$ for vertex $i$ are drawn independently from a common distribution $F$ on $\mathbb{R}^d$; and our practical applications, we consider how to use an omnibus embedding to address the question of equality of potentially non-random latent positions.  For the case in which the latent positions are drawn at random from some distribution $F$, an important graph inference task is the inference of  properties of $F$ from an observation of the graph alone. In the graph inference setting, there is both randomness in the latent positions, and {\em given these latent positions}, a subsequent conditional randomness in the existence of edges between vertices. 
%		Because of this, the task of inferring properties of the underlying distribution $F$ from a mere observation of the adjacency matrix $\bA$ is more complicated than the classical problem of inferring properties of $F$ directly from the $X_i$'s, the latter of which represent an i.i.d. sample from $F$. This is because {\em these latent positions $X_i$ are not observed in the first place}.  
A key to inference in such models is the initial step of consistently estimating the unobserved $X_i$'s from a spectral decomposition of $\bA$, and then using these estimates, denoted $\hat{X}_i$, to infer properties of $F$.

% The present work focuses on {\em random dot product graphs} \citep[RDPG][]{young2007random}, independent-edge random graphs in which each vertex $i$ has an associated unobserved vector $\bX_i \in \R^d$, called the {\em latent position} or {\em latent vector}.
% The probability $p_{ij}$ of an edge joining vertices $i$ and $j$ is simply the dot product of their associated latent positions,
% independent of all other edges in the graph.
For an RDPG with $n$ vertices, the $n\times d$ matrix of latent positions $\bX$ is formed by taking vector $\bX_i$ associated to vertex $i$ to be the $i$-th row of $\bX$. Then $\bP=[p_{ij}]$, the matrix of probabilities of edges between vertices, is easily expressed as $\bP=\bX\bX^T$. The aforementioned nonidentifiability is now transparent: if $\bW$ is orthogonal, then $\bX\bW \bW^T \bX^T=\bP$ as well, so the rotated latent positions $\bX\bW$ generate the same the matrix of probabilities.
Given such a model, a natural inference task is that of estimating the latent
position matrix $\bX$ up to some orthogonal transformation.

Because of the assumption that the matrix $\bP$ is of comparatively low rank,
random dot product graphs can be analyzed with a number of tools from classical linear algebra,
such as singular-value decompositions of their adjacency matrices. Nevertheless, this tractability does not compromise the utility of the model. Random dot product graphs are flexible enough to approximate a wide class of independent-edge random graphs \citep{tang2012universally},
including the stochastic block model \citep{holland,karrer2011stochastic}.

Under mild assumptions, the adjacency matrix $\bA$ of a random dot product graph is a rough approximation of the matrix $\bP=[p_{ij}]$ of edge probabilities in the sense that the spectral norm of $\bA-\bP$ can be controlled;
see for example \cite{oliveira2009concentration} and \cite{lu13:_spect}.
In \cite{STFP-2011}, \cite{STFP-2011} and \cite{lyzinski13:_perfec}, it is established that, under eigengap assumptions on $\bP$, a partial spectral decomposition of the adjacency matrix $\bA$, known as the {\em adjacency spectral embedding} (ASE), allows for consistent estimation of the true, unobserved latent positions $\bX$. That is, if we define $\hat{\bX}=\UA\SA^{1/2}$, where $\SA$ is the diagonal matrix of the top $d$ eigenvalues of $A$, sorted by magnitude, and if $\UA$ are the associated unit eigenvectors, then the rows of this truncated eigendecomposition of $\bA$ are consistent estimates $\{\Xhat_i\}$ of the latent positions $\{\bX_i\}$. Of course, these latent positions are often the parameters we wish to estimate. In \cite{lyzinski15_HSBM}, it is shown that embedding the adjacency matrix and then performing a novel angle-based clustering of the rows is key to decomposing large, hierarchical networks into structurally similar subcommunities. In \cite{athreya2013limit}, it is shown that the suitably-scaled eigenvectors of the adjacency matrix converge in distribution to a Gaussian mixture. In this paper, we prove a similar result for an omnibus matrix generated from multiple independent graphs.

The ASE provides a consistent estimate for the true latent positions in a random dot product graph up to orthogonal transformations. Hence a Procrustes distance between the adjacency spectral embedding of two graphs on the same vertex set serves as a test statistic for determining whether two random dot product graphs have the same latent positions
\citep{tang14:_semipar}.
Specifically,
let $\bA^{(1)}$ and $\bA^{(2)}$ be the adjacency matrices of two random dot product graphs on the same vertex set (with known vertex correspondence), and let $\Xhat$ and $\Yhat$ be their respective adjacency spectral embeddings.
If the two graphs have the same generating $\bP$ matrices, 
it is reasonable to surmise that the Procrustes distance
\begin{equation} \label{eq:semipar:proc}  \min_{\bW \in \calO^{d \times d}}\|\Xhat -\Yhat \bW \|_F,\end{equation}
where $\| \cdot \|_F$ denotes the Frobenius norm of a matrix, will be relatively small.
% Observe that the minimum over orthogonal matrices $\bW$---that is, the Procrustes distance between the embeddings $\Xhat$ and $\Yhat$---is taken
% because of the non-identifiability arising from the fact that
% if $\bW$ is a unitary $d \times d$ matrix,
% then the matrix $\bX\bW$ generates the same matrix of probabilities $\bP$ as does $\bX$.
In \cite{tang14:_semipar}, the authors show that a scaled version of the Procrustes distance in \eqref{eq:semipar:proc} provides a valid and consistent
test for the equality of latent positions for a pair of random dot product
graphs. Unfortunately, the fact that a Procrustes minimization must be performed
both complicates the test statistic and compromises its power.

Here, we instead consider an embedding of an {\em omnibus matrix}, defined as follows. Given two independent $d$-dimensional RDPG adjacency matrices $\bA^{(1)}$ and $\bA^{(2)}$, on the same vertex set with known vertex correspondence, the {\em omnibus matrix} $\bM$ is given by
% (see Section~\ref{sec:expts}) that it is possible
% instead to perform a $d$-dimensional singular value decomposition of the $2n \times 2n$ omnibus matrix of the form
\begin{equation}\label{eq:omnibus_M}
\bM=\begin{bmatrix}
\bA^{(1)} 			& \frac{\bA^{(1)}+\bA^{(2)}}{2}\\
\frac{\bA^{(1)} + \bA^{(2)}}{2}	& \bA^{(2)}
\end{bmatrix},
\end{equation}
Note that this matrix easily extends to a sequence of graphs $\bA^{(1)}, \cdots, \bA^{(m)}$, where the block diagonal entries are the matrices $\bA^{(i)}$ and the $(l,k)$-th off-diagonal block is the matrix $\frac{\bA^{(k)}+ \bA^{(l)}}{2}$.

Analogously to our notation for the adjacency spectral embedding $\hat{\bX}$, let $\SM$ represent the $d \times d$ matrix of top $d$ eigenvalues of $\bM$, ordered again by magnitude, and let $\UM$ be the $mn \times d$-dimensional matrix of associated eigenvectors. Define the {\em omnibus embedding}, denoted $\textrm{OMNI}(\bM)$, by $\UM \SM^{1/2}$. 
We stress that $\textrm{OMNI}(\bM)$ produces $m$ separate points in Euclidean {\em for each graph vertex}---effectively, one such point for each copy of the multiple graphs in our sample. This property renders the omnibus embedding useful for all manner of post-hoc inference. 

If we consider the rows of the omnibus embedding as potential estimates for the latent positions, two immediate questions are arise. Are these estimates consistent, and can we describe a scaled limiting distribution for them as graph size increases? We answer both of these in the affirmative.
\\

\noindent {\bf Key result 1: The rows of the omnibus embedding provide consistent estimates for graph latent positions.} If the latent positions of the graphs $\bA^{(1)}, \cdots, \bA^{m}$ are equal, then under mild assumptions, the rows of $\textrm{OMNI}(\bM)$ provide consistent estimates of their corresponding latent positions. Specifically, if $h=n(s-1)+ i$, where $1 \leq i \leq n$ and $1 \leq s \leq m$, then there exists an orthogonal matrix $\bW$ such that
$$\max_{1 \leq i \leq n}\|(\UM\SM^{1/2})_{h}-\bW\bX_i\|<\frac{C \log mn}{\sqrt{n}}$$
with high probability.
This consistency result is especially useful because it bounds the error between true and estimated latent positions for {\em all latent positions simultaneously}. It further guarantees that the omnibus embedding competes well against the current best-performing estimator of the common latent positions $\bX$, which is the adjacency spectral embedding of the sample mean matrix $\bar{\bA}=\frac{\sum_{i=1}^m \bA^{(i)}}{m}$ \citep{runze_law_large_graphs}.  
Thus, the omnibus embedding is not just consistent when the latent positions are equal; it is close to near-optimal, and we exhibit this clearly in simulations (see Sec.~\ref{sec:expts}).

Our second key result concerns the limiting distribution, as the graph size $n$ increases, of the rows of the omnibus embedding in the case when the graphs are independent and have the same latent positions.
\\

\noindent {\bf Key result 2: For large graphs, the scaled rows of the omnibus embedding are asymptotically normal}. Suppose that the latent positions for each graph are drawn i.i.d from a suitable distribution $F$, and that conditional on these latent positions, the adjacency matrices $\bA^{(1)}, \cdots, \bA^{(m)}$ are independent realizations of random dot product graphs with the given latent positions. Let $\bZ$ represent that $mn\times d$ dimensional matrix of latent positions for the graphs. Let $h=n(s-1) + i$, where $1 \leq i \leq n$ and $1 \leq s \leq m-1$.  Then under mild assumptions, there exists a sequence of orthogonal matrices $\bW_n$ such that as $n \rightarrow \infty$, 
$$\sqrt{n}[(\UM\SM^{1/2})\bW_n-\bZ]_h$$ converges to a mean-zero $d$-dimensional Gaussian mixture.
Hence the omnibus embedding allows for accurate estimation when the latent positions are equal; provides multiple points for each vertex; and under mild assumptions on the structure of the latent positions, these embeddings are approximately normal for large graph sizes. Even more remarkably, 
for testing whether two graphs have the same latent positions, we can build the omnibus embedding
and consider only the Frobenius norm of the difference between the matrices defined by, respectively, the first $n$ and the second $n$ rows of this decomposition.  This matrix difference, {\em without any further Procrustes alignment}, also serves as a test statistic for the equality of latent positions, which brings us to our third point. 
\\

\noindent {\bf Key simulation evidence: The omnibus embedding has meaningful power for two- and multi-sample graph hypothesis testing}. The omnibus embedding, without subsequent Procrustes alignments, yields an improvement in power over state-of-the-art methods in two-graph testing, as borne out by a comparison on simulated data. Combined with our earlier bounds on the $2 \to \infty$-norm difference between true and estimated latent positions, this demonstrates that the omnibus embedding provides estimation accuracy when the graphs are drawn from the same latent positions and improved discriminatory power when they are different.
What is more, the omnibus embedding produces multiple points for each vertex and our asymptotic normality guarantees that these points are approximately normal. As a consequence, the omnibus embedding not only permits the discovery of graph-wide differences, but also the {\em isolation of vertices} that contribute to these differences, via, for instance, a multivariate analysis of variance (MANOVA) applied to these embedded vectors. This leads us to our final point.
\\

\noindent {\bf Key real data analysis: the omnibus embedding isolates graph-wide differences and gives principled evidence for vertex significance in those differences, and works well on complicated real data}. On weighted, directed, noisily-observed real graphs, slight modifications to the omnibus embedding procedure yield genuine exploratory insights into graph structure and vertex importance, with actionable import in application domains. 

To demonstrate this, we present in the next subsections detailed analyses
of two neuroscientific data sets.
The first is a connectomic data set, in which we analyze a collection of paired diffusion MRI (dMRI) brain scans across 57 patients \citep{kiar_two_truths}, a comparison of 114 graphs with 172 common vertices. 
Second, we consider the COBRE data set \citep{AineETAL2017},
a collection of functional MRI scans of 
$54$ schizophrenic patients and $69$ healthy controls,
with each scan yielding a graph on $264$ common vertices.
In both data sets,
we show how the omnibus embeddings can identify whole-graph differences
as well as particular vertices involved in this difference.

\subsection{Discerning vertex-level difference in paired brain scans of human subjects}
\label{subsec:BNU1}

As a case study in the utility of our techniques, we consider data from human brain scans collected on 57 subjects at Beijing Normal University. The data, labeled ``BNU1", is available at \\
\texttt{http://fcon\_1000.projects.nitrc.org/indi/CoRR/html/bnu\_1.html}.

This diffusion MRI data comprises two scans on each of 57 different patients, for a total of 114 scans. The dMRI data was converted into weighted graphs via the Neuro-Data-MRI-to-Graphs (NDMG) pipeline of \cite{kiar_two_truths}, with vertices
representing sub-regions defined via spatial proximity and
edges by tensor-based fiber streamlines connecting these regions. As such, by condensing vertices in a given brain region still further, neuroscientists can represent this data at different scales. We focus on data at a resolution in which the $m=114$ graphs have $n=172$ common vertices. (We point out that in \cite{cep_two_truths}, this same data, at slightly different scale, serves as a useful illustration of different structural properties uncovered by different spectral embeddings). Our inference goals are to determine which of these graphs appear statistically similar and to elucidate which vertices might be key contributors to such difference.

We binarize the graphs via a simple thresholding operation,
replacing all nonzero edge weights with 1, and leaving unchanged edges of weight zero. For reference, Figure~\ref{fig:sample_adj_mat} shows a visualization of one pair of adjacency matrices. Alternative approaches to weighted graphs, such as a replacement of weights by a rank-ordering, are also possible and have proven useful \citep[see][]{cep_two_truths}, but we do not pursue this here.

\begin{figure}[tbh!]
    \centering
    \subfloat[]{\includegraphics[width=0.4\columnwidth]{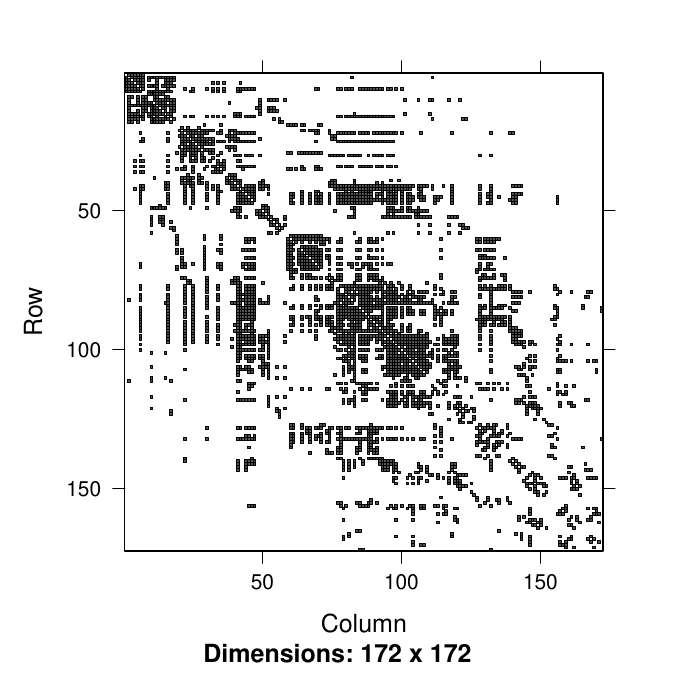} }
   \subfloat[]{\includegraphics[width=0.4\columnwidth]{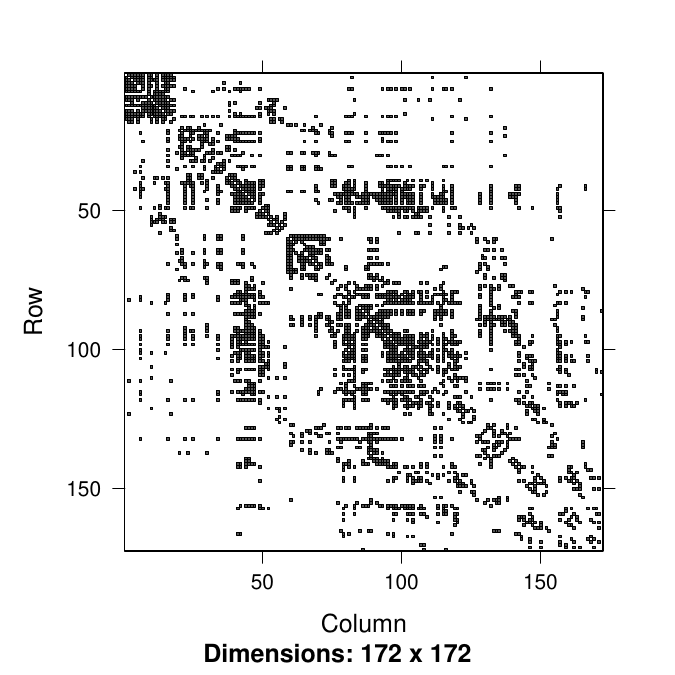} }
    \caption{Binarized adjacency matrices for two BNU-1 brain graphs}
    \label{fig:sample_adj_mat}
\end{figure}

With these $m=114$ binarized, undirected adjacency matrices, we generate the $m$-fold omnibus matrix $M$, the $m$-fold analogue of the matrix in Eq.~\eqref{eq:omnibus_M}, which is of size $(114 \times 172)\times(114 \times 172)$.
To select a dimension for the omnibus embedding,
we apply the profile-likelihood method of \cite{zhu06:_autom} to $M$.
This procedure performs model selection (i.e., estimates the rank of $\E M$)
by locating an elbow in the screeplot of the eigenvalues of $M$
(shown in Figure~\ref{fig:embedded_dim_BNU}(a))
This yields an estimated embedding dimension of $\hat{d}=10$.
As a check, we perform the same estimation procedure
on all of the 114 graphs as well.
Reassuringly, we recover an estimated embedding dimension close to $10$
for almost all of them, and proceed with $\hat{d}=10$,
as summarized by the boxplot in Fig.~\ref{fig:embedded_dim_BNU}(b).
\begin{figure}[tbh!]
    \centering
    \subfloat[]{\includegraphics[width=0.4\columnwidth]{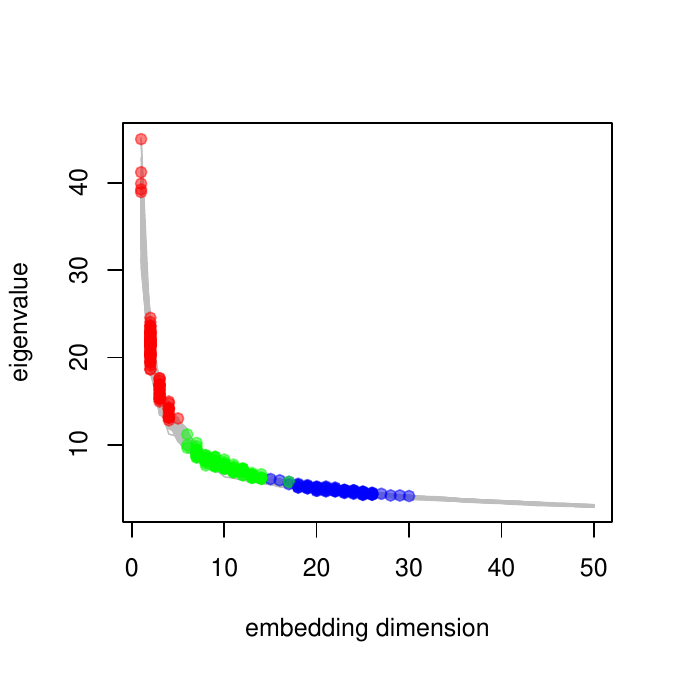} }
   \subfloat[]{\includegraphics[width=0.4\columnwidth]{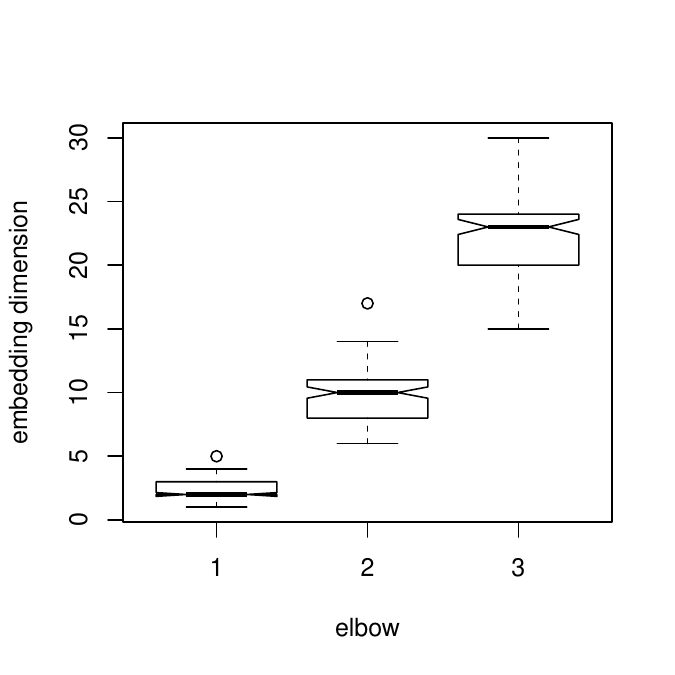} }
    \caption{(a) Eigenvalues and embedding dimension for the omnibus matrix as BNU data. Note the second elbow at $\dhat=10$. (b) Box plot of the first three elbows in the scree plot identified by the profile-likelihood method of \cite{zhu06:_autom} applied to the 114 brain graphs individually. Note that the second elbow is concentrated around $\dhat = 10$.}
    \label{fig:embedded_dim_BNU}
\end{figure}
We now construct a {\em centered} omnibus matrix, in which we subtract the sample mean $\bar{\bA}=m^{-1}\sum_{l=1}^m \bA^{(l)}$ from the omnibus matrix.
That is, we construct an omnibus matrix
from the centered graphs $\bB^{(l)}=\bA^{(l)}-\bar{\bA}$
instead of the observed graphs $\bA{(l)}$.
Having constructed this centered omnibus matrix,
we embed it into $\dhat=10$ dimensions,
producing a matrix $\hat{Z} \in \R^{mn \times \dhat}$,
with $\dhat=10$ columns and $mn = 114\times172$ rows.
Observe that $\hat{Z}$ can be subdivided into 114 blocks each of size 172,
one for each graph.
For convenience, we denote these submatrices, each of size $172\times 10$, by
$\hat{{\bf X}}^{(1)}, \cdots, \hat{{\bf X}}^{(114)}$.
Under our model assumptions, each of these submatrices is an estimate of the latent position matrix of the corresponding brain graph.

Because the omnibus embedding introduces an alignment between graphs by placing an average on the off-diagonal blocks of the omnibus matrix, we find that merely considering a Frobenius norm difference between blocks of the omnibus embedding, i.e.,
$$\|\hat{{\bf X}}^{(l)}-\hat{{\bf X}}^{(k)}||_F,$$
{\em without} any further Procrustes alignments, provides meaningful power in distinguishing between graphs with different latent positions (again, see Sec.~\ref{sec:expts} for simulation evidence, and Sec.~\ref{sec:conc} for theoretical discussion of why the omnibus embedding obviates the need for further subspace alignments). As a consequence, we can create a $114\times 114$ dissimilarity matrix $\bD=(\bD_{kl}) \in \R^{114 \times 114}$, defined as
$$\bD_{kl}:=\|\hat{{\bf X}}^{(l)}-\hat{{\bf X}}^{(k)}||_F$$
which records the Frobenius norm differences of the omnibus embeddings of the $k$-th and $l$-th graph in our collection. We illustrate this dissimilarity matrix in the first panel of Fig.~\ref{fig:omni_dissim_matrix}, and we show how this matrix can be hierarchically clustered. 
\begin{figure}[h!]
    \centering
    \subfloat[]{\includegraphics[width=0.5\columnwidth]{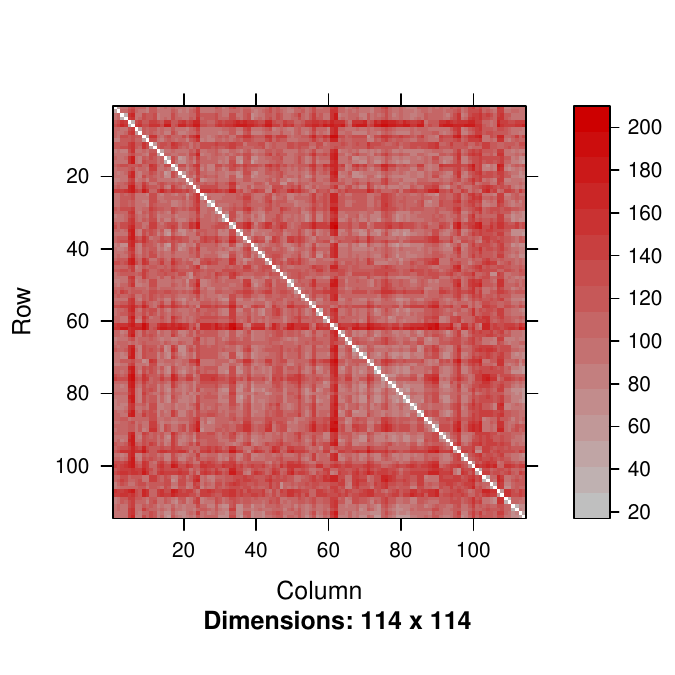}}
   \subfloat[]{\includegraphics[width=0.5\columnwidth]{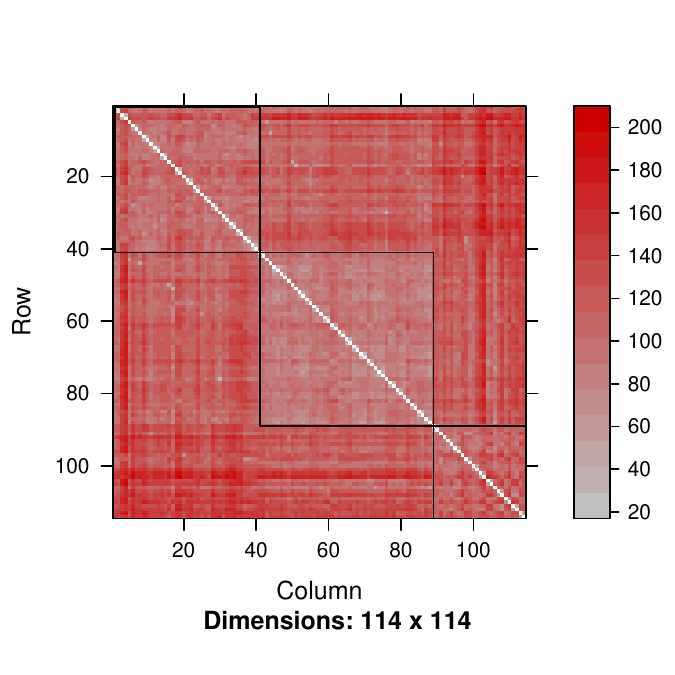}} 
   \caption{(a) Omnibus dissimilarity matrix $\bD$ across 114 graphs. (b) Results of hierarchical clustering of this dissimilarity matrix. }
   \label{fig:omni_dissim_matrix}
\end{figure}

Using classical multidimensional scaling \citep{cox_MDS}, we embed this dissimilarity matrix into $2$-dimensional Euclidean space.
This yields a collection of 114 points in $\R^2$, each one of which represents
one graph.
We then cluster this collection of points using Gaussian mixture modeling,
in which we select $c=3$ clusters according to the Bayesian Information Criterion (BIC). The three resulting clusters are depicted in Fig.~\ref{fig:GMM_clustering_BNU}(b). We remark that out of 57 subjects, only 10 subject scans are divided across clusters, suggesting that the clusters capture meaningful similarity across graphs.
\begin{figure}[tbh!]
    \centering
    \subfloat[]{\includegraphics[width=0.5\columnwidth]{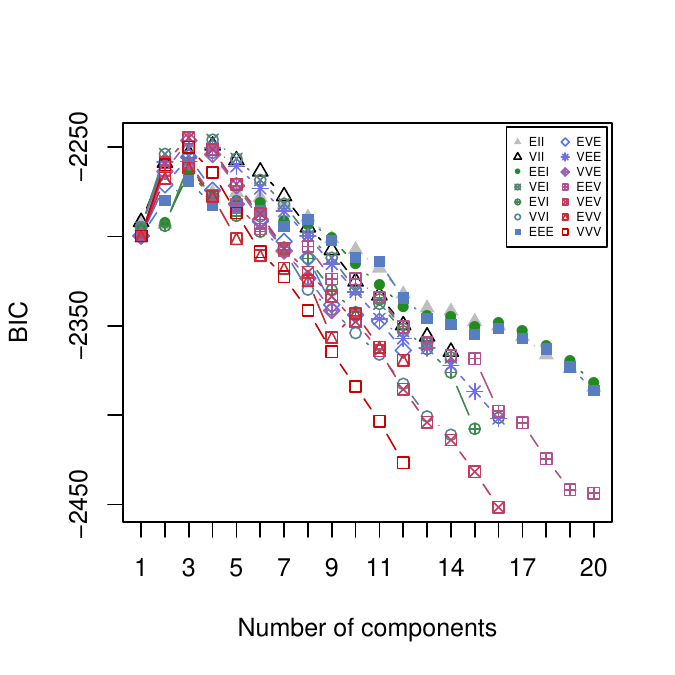}}
   \subfloat[]{\includegraphics[width=0.5\columnwidth]{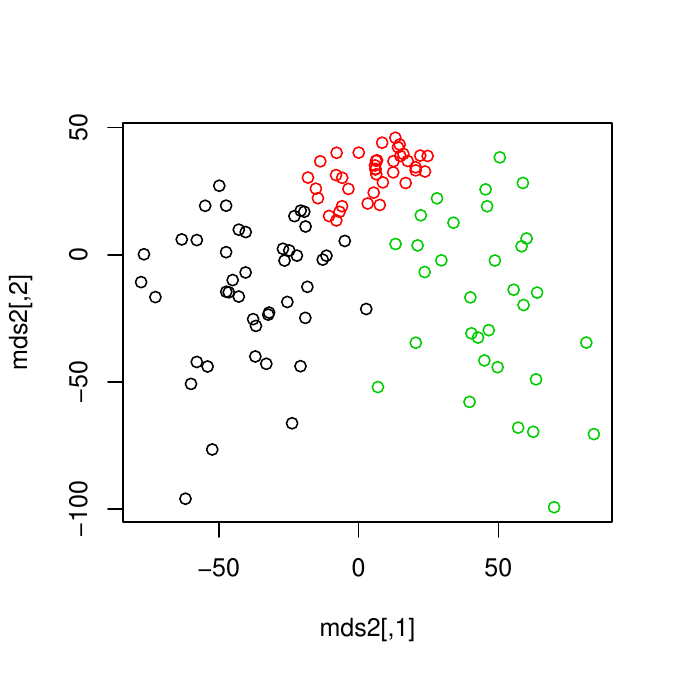}} 
    \caption{(a) Clustering using Gaussian mixture modeling and selection of $c=3$ clusters as applied to the $2$-dimensional CMDS embedding of $\bD$. (b) Visualization of the embeddings and their resulting clusters. Each point represents a single graph from the BNU1 data set.}
   \label{fig:GMM_clustering_BNU}
\end{figure}

If it were truly the case that all of these graphs
had the same latent positions,
our main central limit theorem would ensure that for large graph sizes
these embedded points would be asymptotically normal.
Thus, if we consider these three clusters as identifying three distinct types
of graphs, we can now compare embedded latent positions of individual vertices
across graphs to determine which vertices play the most similar or
different roles in their respective graphs.
The (theoretical) asymptotic normality leads us to consider a multivariate analysis of variance (MANOVA).
For each vertex, we consider the embedded points corresponding to that vertex
that arise from the graphs in Cluster 1, Cluster 2, and Cluster 3,
respectively.
{\em Because we have multiple embedded points for each vertex and these
embedded points are asymptotically normal,} MANOVA is,
as an exploratory tool, principled.
MANOVA produces $p$-value for each vertex, associated to the test of equality
of the true mean vectors for the normal distributions governing the embedded
points in each of the $3$ classes
Since there are 172 vertices, we obtain 172 corresponding $p$-values, and we correct for multiple comparisons using the Bonferroni correction.  The $p$-values are ordered by significance in Fig.~\ref{fig:MANOVA_pvalue_BNU}.
\begin{figure}[!h]
\centering
\includegraphics[width=0.5\textwidth]{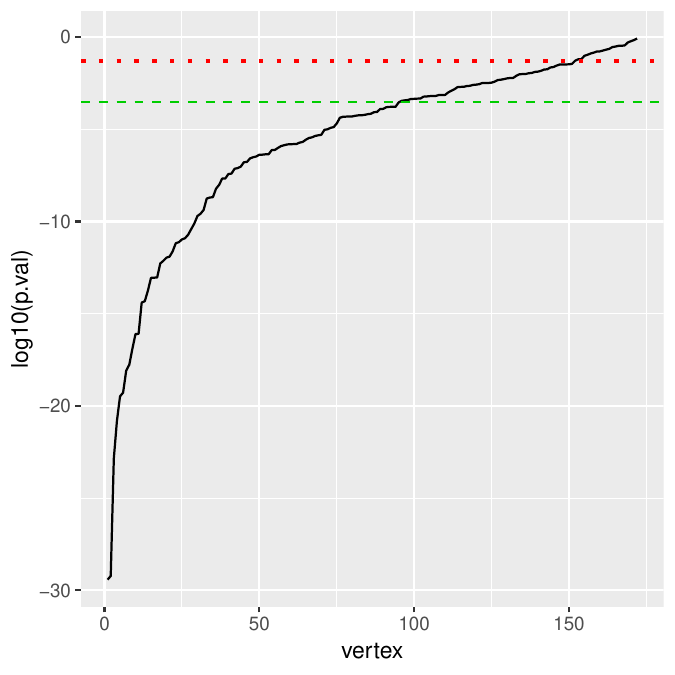}
\caption{MANOVA $p$-values, with vertices sorted by significance and adjusted for multiple comparisons. The dotted lines indicate the $p=0.05$ threshold (green) and the threshold after Bonferroni correction.}
\label{fig:MANOVA_pvalue_BNU}
\end{figure}
We focus on one of the two most significant vertices, Vertex 98.
Recall that by nature of the omnibus embedding, we have multiple points
embedded in $\R^{\dhat}$, each of which correspond to the $98$-th vertex
in one of the 114 graphs. Partitioning these $114$ points according to the clusters associated to their respective graphs, we can perform nonparametric tests of difference across
these collections of points to further illuminate how this
vertex differs in its behavior across the three graph clusters.
Fig.~\ref{fig:MANOVA_most_sig_vertex_cluster_diff_BNU_dim1_and_2}
illustrates  how strikingly different are the first two principal dimensions
of the embedded points for this vertex across the three different clusters.
For contrast, we examine one of the least significant vertices, Vertex 124, and reproduce the analogous plots to those in Fig.\ref{fig:MANOVA_most_sig_vertex_cluster_diff_BNU_dim1_and_2} for its first principal dimension.
The results are displayed in Fig.~\ref{fig:MANOVA_least_sig_vertex_cluster_diff_BNU}.

\begin{figure}[!h]
\centering
\subfloat[]{\includegraphics[width=0.5\textwidth]{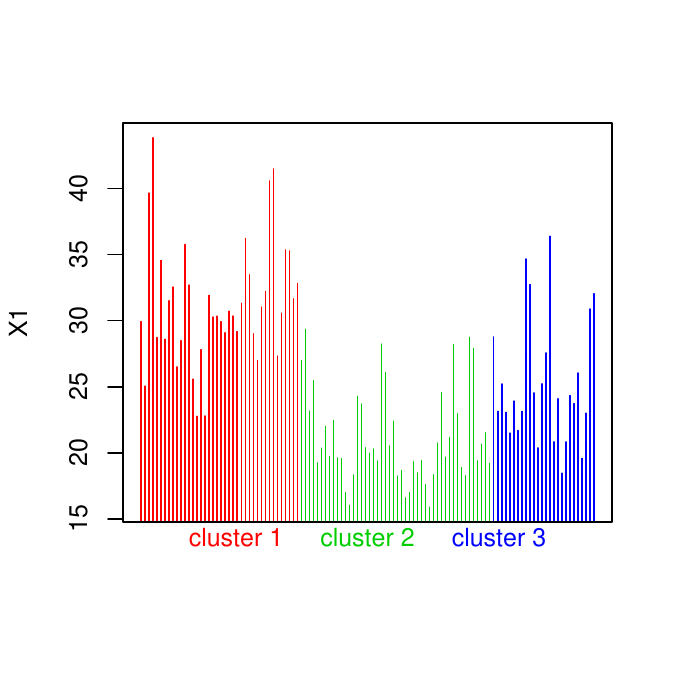}}
\subfloat[]{\includegraphics[width=0.5\textwidth]{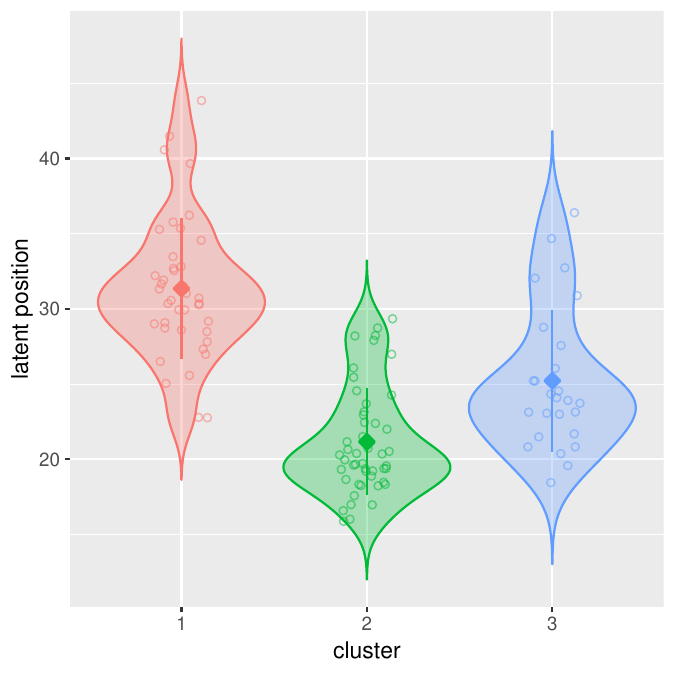}}\\
\subfloat[]{\includegraphics[width=0.5\textwidth]{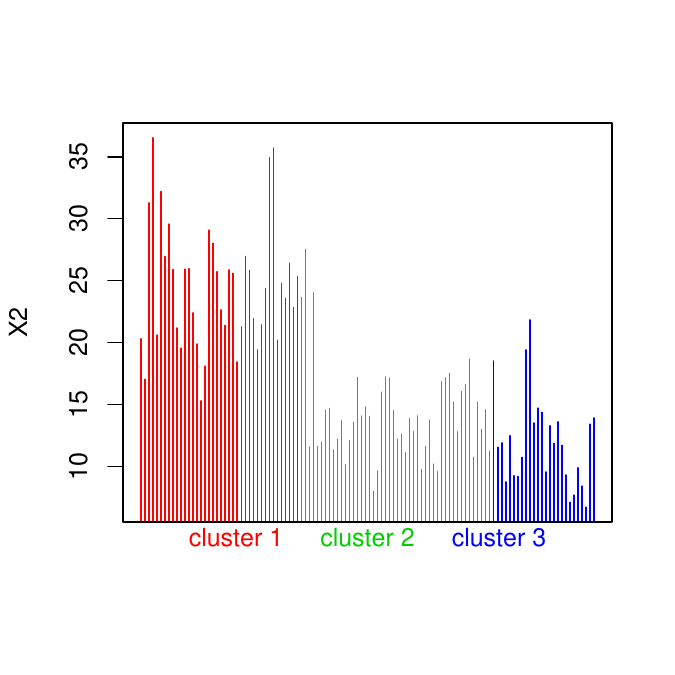}}
\subfloat[]{\includegraphics[width=0.5\textwidth]{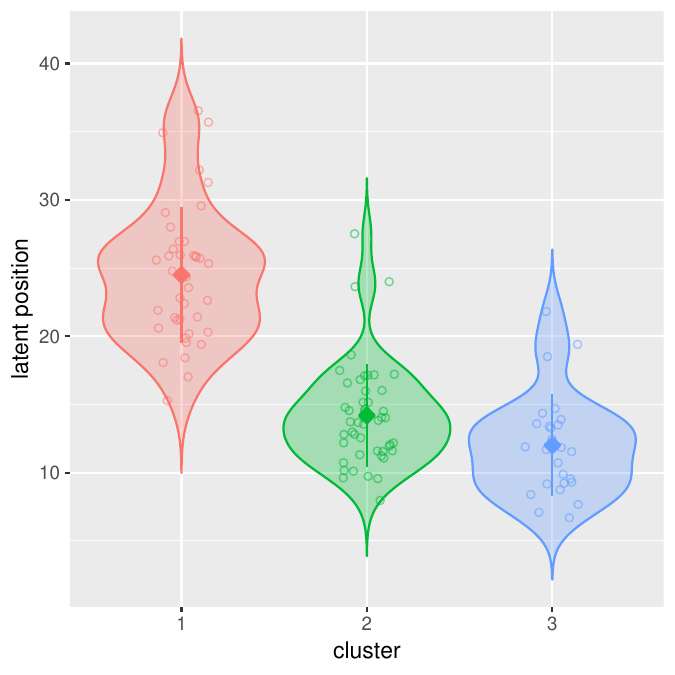}}
\caption{(a) For the most significant vertex (Vertex 98), a histogram of the first principal dimension of embedded latent position, grouped by cluster. (b) For this same vertex, estimated mean and confidence intervals for the first dimension of the embedded position, again grouped by cluster. (c) and (d) show the analogous plots for the second principal dimension. Observe that the first principal dimension distinguishes between the first and second cluster, and the second principal dimension between the first and the third cluster.}
\label{fig:MANOVA_most_sig_vertex_cluster_diff_BNU_dim1_and_2}
\end{figure}

\begin{figure}[h!]
\centering
\subfloat[]{\includegraphics[width=0.5\textwidth]{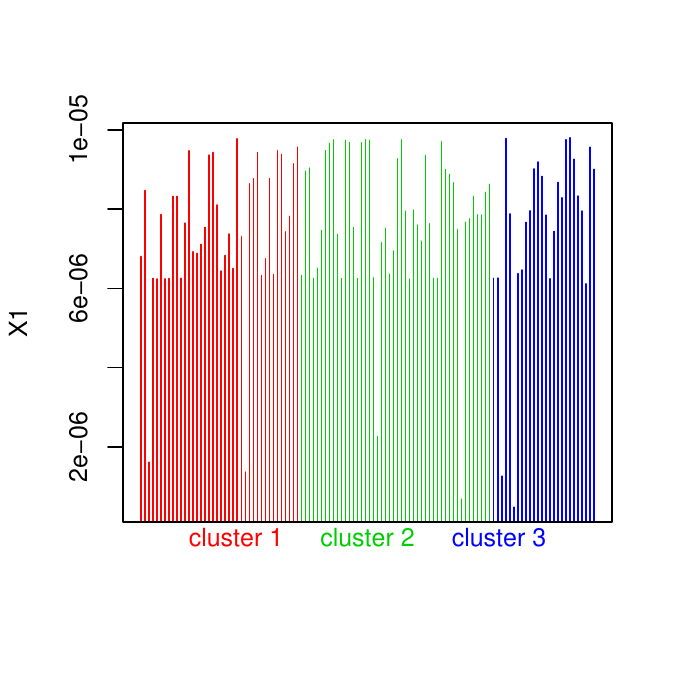}}
\subfloat[]{\includegraphics[width=0.5\textwidth]{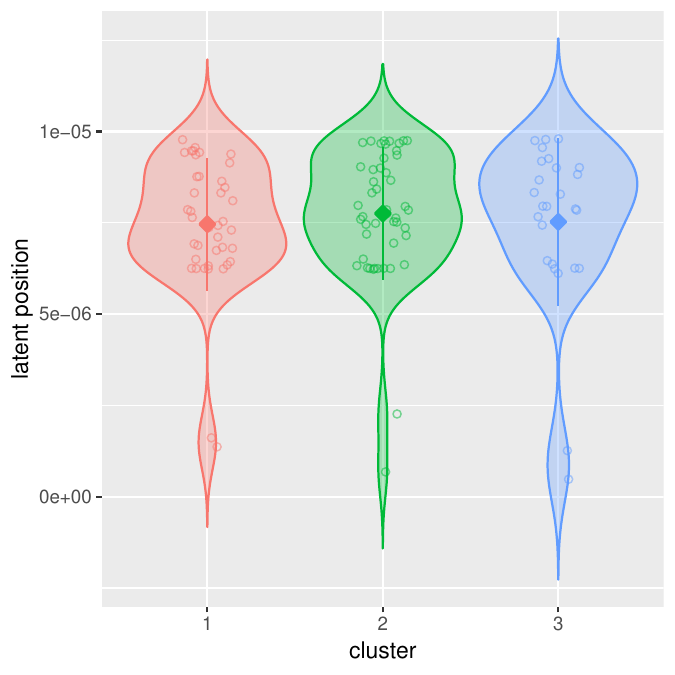}
}
\caption{(a) For the least significant vertex, a histogram of the first dimension of estimated latent positions, grouped by cluster. (b) For this same vertex, estimated means and confidence intervals for the first dimension of the estimated latent position, again grouped by cluster.}
\label{fig:MANOVA_least_sig_vertex_cluster_diff_BNU}
\end{figure}

The contrast between Figs.~\ref{fig:MANOVA_most_sig_vertex_cluster_diff_BNU_dim1_and_2} and \ref{fig:MANOVA_least_sig_vertex_cluster_diff_BNU} is striking.
Because the omnibus embedding gives us multiple points for each vertex,
we are able to isolate vertices that are responsible for between-graph
differences
and then interface with neuroscientists to discern what physical distinctions
might be present at this vertex or brain location across the clusters of graphs.

We recognize, of course, some immediate concerns with our procedure.
First, the fact that our clusters are determined post-hoc implies that
the embedded vectors are not independent samples from different populations.
Second, the asymptotic normality of the embedded positions is a large-sample
result, and applies to an arbitrary but finitely fixed collection of rows. Despite these limitations, we stress that our theoretical results supply a principled foundation on which to build a more refined analysis, and to date this is among the only approaches for the identification and
comparison of individual vertices and their role in
driving differences between (populations of) graphs.

\subsection{Identifying brain regions associated with schizophrenia}
\label{subsec:COBRE}

We next consider the COBRE data set \citep{AineETAL2017},
a collection of  scans of both schizophrenic and healthy patients.
Each  scan yields a graph on $n = 264$ vertices, corresponding to 264
brain regions of interest \citep{PowerETAL2011},
with edge weights given by correlations between BOLD signals
measured in those regions.
The data set contains  scans for
$54$ schizophrenic patients and $69$ healthy controls, for a total of
$m = 123$ brain graphs.

We follow the general framework of our BNU1 analysis above.
Under the null hypothesis that all $m$ graphs share the same underlying
latent positions, the omnibus embedding yields for each vertex
a collection of $m$ points in $\R^d$ that are
normally distributed about the true latent position of that vertex.
By applying an omnibus embedding to the $m=123$ subjects in the COBRE dataset,
we can therefore test,
for each vertex $i \in [264]$, whether or not the healthy and
schizophrenic populations display a difference in that vertex,
by comparing the latent positions of vertex $i$ associated with the
schizophrenic patients against those associated with the healthy patients.
That is, let $m_s = 54$ denote the number of schizophrenic patients
and $m_h = 69$ denote the number of healthy controls,
with respective embeddings given by
$$\{ X^{(j)}_i : j =1,2,\dots,m_s \}, \textrm{ and } \{ Y^{(j)}_i : j=1,2,\dots,m_h \}$$
We can test whether the samples
$$\{ X^{(1)}_i, X^{(2)}_i, \dots, X^{(m_h)} \} \subseteq \R^d
\textrm{ and } \{ Y^{(1)}_i, Y^{(2)}_i, \dots, Y^{(m_s)} \} \subseteq \R^d$$
appear to come from the same distribution.
By Theorem~\ref{thm:main},
if all $m$ subjects' graphs are drawn from the same
underlying RDPG, then it is natural to test the hypothesis that both
%$\{ X^{(j)}_i : j=1,2,\dots,m_s \}$ and $\{ Y^{(j)}_i : j=1,2,\dots,m_h \}$
 the $X_i^(j)$ and then $Y_i^(j), 1 \leq j \leq m_k$ are drawn from the same normal distribution.
We use Hotelling's $t^2$ test \citep{Hotelling1931,Anderson2003} (and we remark that experiments applying a permutation test for this same
purpose yield broadly similar results).
We note that while in the BNU1 data example in Section~\ref{subsec:BNU1},
we required a clustering to discover collections of similarly-behaving
networks, the COBRE data set already has two populations of interest
in the form of the healthy and schizophrenic patients.

We begin by building the omnibus matrix of $m=123$  brain graphs,
each on $n=264$ vertices.
In contrast to the BNU1 data presented above,
here we work with the weighted graph obtained from  scans,
rather than binarizing them.
We apply a three-dimensional omnibus embedding to these $m$ graphs,
yielding $123$ points in $\R^3$ for each of the $n=264$ brain regions
for a total of $32472$ points.
For each vertex $i \in \{1,2,\dots,264\}$,
there are $m=123$ points in $\R^3$ each corresponding
to vertex $i$ in one of the brain graphs.
$54$ of these $123$ points correspond to the estimated
latent position of the $i$-th vertex in the schizophrenic patients,
while the remaining $69$ points correspond to the estimated latent position
of the $i$-the vertex in the healthy patients.
For each vertex $i$, we apply Hotelling's $t^2$ test to assess whether or not
the healthy and schizophrenic estimated latent positions appear to come
from different populations.
Thus, for each of the $264$ regions of interest, we obtain a $p$-value
that captures the extent to which the estimated latent positions of the
healthy and schizophrenic patients appear to differ in their distributions.

Figure~\ref{fig:pval_by_parcel} summarizes the result of the procedure
just described. Using the Power parcellation \citep{PowerETAL2011},
we group the 264 brain regions into larger {\em parcels},
which capture what are believed by
neuroscientists to correspond to functional subnetworks of the brain.
For example, a parcel called the
{\em default mode network} is associated with wakeful, undirected thought
(i.e., mind wandering), and is implicated in
schizophrenia \citep{WhitfieldGabrieliETAL2009,FoxETAL2015}.
We collect, for each of the 14 Power parcels, the $p$-values associated with
all of the brain regions (i.e., vertices) in that parcel, and display in
Figure~\ref{fig:pval_by_parcel} a histogram of those $p$-values.
Under this setup,
parcels in which the populations are largely the same will have histograms
that appear more or less flat,
while parcels in which schizophrenic patients
display different behavior from their healthy counterparts will result
in left-skewed histograms.
Observing Figure~\ref{fig:pval_by_parcel}, we see strong visual evidence
that the default mode, the sensory/somatomotor hand
and the uncertain parcels are affected by schizophrenia.
%TODO: do a KS test on these against uniform?

%\begin{table}
%  \label{tab:powerparcel}
%  \caption{Brief description of each of the 14 networks in the Power parcellation. See \cite{PowerETAL2011} for more information.}
%  \begin{center}
%  \begin{tabular} {  r | l  }
%  {\bf Parcel } & {\bf Description } \\
%  Auditory & auditory processing \\
%  Cerebellar & structures in the cerebellum \\
%  Cingulo-opercular Task Control & description \\
%  Default mode & description \\
%  Dorsal attention & description \\
%  Fronto-parietal Task Control & description \\
%  Memory retrieval & description \\
%  Salience & description \\
%  Sensory/somatomotor hand & description \\
%  Sensory/somatomotor mouth & description \\
%  Subcortical & description \\
%  Uncertain & description \\
%  Ventral attention & description \\
%  Visual & description
%  \end{tabular}
%  \end{center}
%\end{table}

\begin{figure}[t!]
 \centering
 \includegraphics[width=\columnwidth]{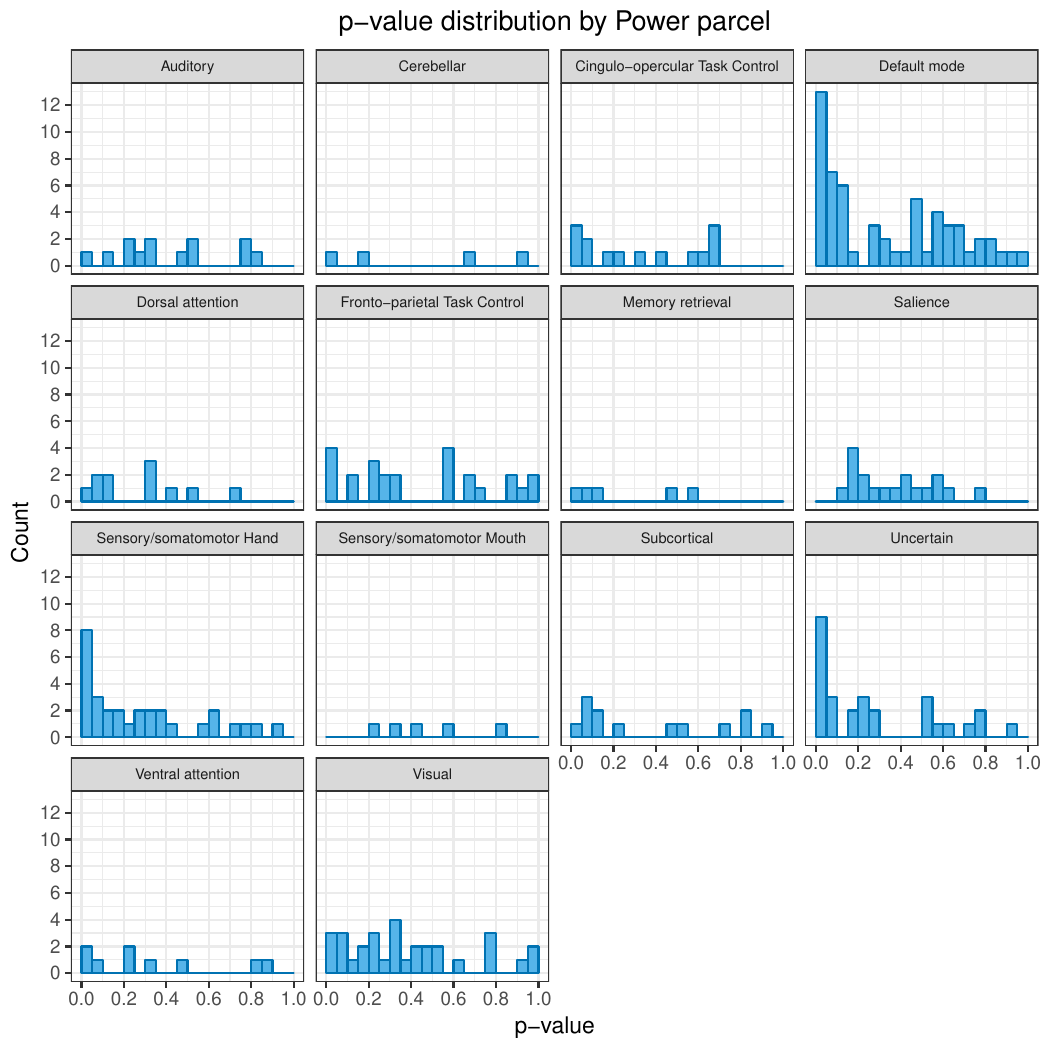}
 \caption{Histograms of the distribution of $p$-values within each parcel.
       Each histogram corresponds to one of the fourteen parcels in the
	Power parcellation \citep{PowerETAL2011},
	and shows the distribution of the $p$-values obtained from
	applying the Hotelling $t^2$ test to the omnibus embeddings
	of the brain regions in that parcel.
	We see that certain parcels (most notably the default mode,
	sensory/somatomotor hand, and uncertain parcels)
	clearly display non-uniform $p$-value distributions, suggesting that
	these parcels differ in schizophrenic patients compared to
	their healthy counterparts.
       }
 \label{fig:pval_by_parcel}
\end{figure}

Here again we see the utility of the omnibus embedding.
Thanks to the alignment of the embeddings across all $123$ graphs in the
sample, we obtain, after comparatively little processing, a concise summary of
which vertices differ in their behavior across the two populations of
interest. Further, this information can be summarized into an simple display of information---in this case, summarizing which Power parcels are likely involved in schizophrenia---that is interpretable by neuroscientists and other domain specialists.

\section{Background, notation, and definitions}\label{sec:formal_definitions}
We now turn toward a more thorough exploration of the theoretical
results alluded to above. We begin by establishing notation
and a few definitions that will prove useful in the sequel.
\subsection{Notation and Definitions}
For a positive integer $n$, we let $[n] = \{1,2,\dots,n\}$,
and denote the identity, zero and all-ones matrices by, respectively,
$\bI$, $\zeromx$ and $\bJ$.
For an $n \times n$ matrix $\bH$, we let $\lambda_i(\bH)$ denote
the $i$-th largest eigenvalue of $\bH$ and 
we let $\sigma_i(\bH)$ denote the $i$-th singular value of $\bH$.
We use $\otimes$ to denote the Kronecker product.
For a vector $\bv$, we let $\| \bv \|$ denote the Euclidean norm of $\bv$.
For a matrix $\bH \in \R^{n_1 \times n_2}$,
we denote by $\bH_{\cdot j}$
the column vector formed by the $j$-th column of $\bH$,
and let $\bH_{i \cdot}$ denote the row vector
formed by the $i$-th row of $\bH$.
For ease of notation, we let $\bH_i \in \R^{n_2}$
denote the \emph{column} vector formed by transposing the $i$-th row
of $\bH$. That is, $\bH_i = (\bH_{i \cdot})^T$.
We let $\| \bH \|$ denote the spectral norm of $\bH$,
$\| \bH \|_F$ denote the Frobenius norm of $\bH$
and $\|\bH\|_{\tti}$ denote the maximum of the Euclidean norms
of the rows of $\bH$, i.e., $\|\bH\|_{\tti}=\max_{i} \| \bH_i \|$.
%We let $\Pi_n$ denote the set of all $n$-by-$n$ permutation matrices.
Where there is no danger of confusion, we will often refer to a graph $G$ and its adjacency matrix $\bA$
interchangeably.
Throughout, we will use $C > 0$ to denote a constant, not depending on $n$,
whose value may vary from one line to another.
For an event $E$, we denote its complement by $E^c$.
Given a sequence of events $\{ E_n \}$,
we say that $E_n$ occurs with high probability,
and write $E_n \text{ w.h.p. }$,
if $\Pr[ E_n^c ] \le Cn^{-2}$ for $n$ sufficiently large.
We note that $E_n$ w.h.p. implies, by the Borel-Cantelli Lemma,
that with probability $1$ there exists an $n_0$ such that
$E_n$ holds for all $n \ge n_0$.

Our focus here is on $d$-dimensional random dot product graphs,
for which the edge connection probabilities arise as inner products between
vectors, called latent positions, that are associated to the vertices.
Therefore, we define an {\em an inner product distribution} as a probability distribution over a suitable subset of $\R^d$, as follows:
\begin{definition}
\label{def:innerprod}
\emph{($d$-dimensional Inner Product Distribution)}
Let $F$ be a probability distribution on $\R^d$.
We say that $F$ is a
\emph{$d$-dimensional inner product distribution}
on $\R^d$ if for all $\bx,\by \in \supp F$, we have $\bx^T \by \in [0,1]$.
\end{definition}
\begin{definition}\label{def:RDPG}
\emph{(Random Dot Product Graph)}
Let $F$ be a $d$-dimensional inner product distribution
with $\bX_1,\bX_2,\dots,\bX_n \iid F$, collected in the rows of the matrix
$\bX=[\bX_1, \bX_2, \dots, \bX_n]^T \in \R^{n \times d}$.
Suppose $\bA$ is a random adjacency matrix given by
\begin{equation} \label{eq:rdpg}
\Pr[\bA|\bX]=
\prod_{i<j}(\bX_i^T\bX_j)^{\bA_{ij}}(1-\bX_i^T\bX_j)^{1-\bA_{ij}}
\end{equation}
We then write $(\bA,\bX) \sim \RDPG(F,n)$ and say that $\bA$ is the adjacency
matrix of a {\em random dot product graph} with {\em latent positions} given by the rows of $\bX$.
\end{definition}
We note that we restrict our attention here to hollow, undirected graphs.

Given $\bX$, the probability $p_{ij}$ of observing an edge between
vertex $i$ and vertex $j$ is simply $\bX_i^T\bX_j$,
the dot product of the associated latent positions $\bX_i$ and $\bX_j$.
We define the matrix of such probabilities by $\bP=[p_{ij}]=\bX\bX^T$,
and write $\bA \sim \Bern(\bP)$ to denote that the existence of an
edge between any two vertices $1 \le i < j \le n$ is a Bernoulli
random variable with probability $p_{ij}$, with these edges independent.
That is, if $\bP = \bX\bX^T$, then $\bA \sim \Bern(\bP)$ implies that
conditioned on $\bX$, $\bA$ is distributed as in Eq.~\eqref{eq:rdpg}.

\begin{remark} \label{rem:nonid}
% Given a graph distributed as an RDPG,
% the natural task is to recover the latent positions $\bX$ that gave
% rise to the observed graph.
% However, the RDPG model has an inherent nonidentifiability
% in this respect:
Note that if $\bX \in \R^{n \times d}$ is a matrix of latent positions
and $\bW \in \R^{d \times d}$ is orthogonal,
$\bX$ and $\bX\bW$ give rise to the same distribution over graphs in
Equation~\eqref{eq:rdpg}.
Thus, the RDPG model has a nonidentifiability up to orthogonal transformation.
\end{remark}

The focus of this paper is on multi-graph inference. As such, we consider a collection of $m$ random dot product graphs,
all with the same latent positions, which motivates the following definition:
\begin{definition}
\emph{(Joint Random Dot Product Graph)}
\label{def:JRDPG}
Let $F$ be a $d$-dimensional inner product distribution on $\R^d$.
We say that random graphs $\bA^{(1)},\bA^{(2)},\dots,\bA^{(m)}$
are distributed as a \emph{joint random dot product graph (JRDPG)}
and write $(\bA^{(1)},\bA^{(2)},\dots,\bA^{(m)},\bX) \sim \JRDPG(F,n,m)$
if $\bX = [\bX_1, \bX_2,\dots,\bX_n]^T \in \R^{n \times d}$ has its (transposed)
rows distributed i.i.d. as $\bX_i \sim F$, and we have 
marginal distributions $(\bA^{(k)},\bX) \sim \RDPG(F,n)$
for each $k=1,2,\dots,m$.
That is, the $\bA^{(k)}$ are conditionally independent
given $\bX$, with edges independently distributed as
$\bA^{(k)}_{i,j} \sim \Bern( (\bX\bX^T)_{ij} )$ for all $1 \le i < j \le n$
and all $k \in [m]$.
%If $A^{(1)},A^{(2)},\dots,A^{(m)}$ are distributed as a JRDPG, we write
%$(A^{(1)},A^{(2)},\dots,A^{(m)},\bX) \sim \JRDPG(F,n,m)$
\end{definition}

Throughout, we let $\delta > 0$ denote the eigengap of
\begin{equation} \label{eq:def:Delta}
\bDelta = \E \bX_1 \bX_1^T \in \R^{d \times d},
\end{equation}
the second moment matrix of $\bX_1 \sim F$.
That is, $\delta = \lambda_d( \bDelta ) > 0 = \lambda_{d+1}( \bDelta )$.
We note that $\bDelta$ can be chosen diagonal without loss of generality
after a suitable change of basis \citep{athreya2013limit}.
We assume further that $\bDelta$ is such that its diagonal entries are
in nonincreasing order, so that
$\bDelta_{1,1} \ge \bDelta_{2,2} \ge \dots \ge \bDelta_{d,d} = \delta.$
We assume that the matrix $\bDelta$ is constant in $n$,
so that $d$ and $\delta$ are constants, while the number of graphs
$m$ is allowed to grow with $n$.
We leave for future work the exploration of the case where the
model parameters are allowed to vary with the number of vertices $n$.

Since we rely on spectral decompositions, we begin with a straightforward one:
the spectral decomposition of the positive semidefinite matrix $\bP=\bX\bX^T$.
\begin{definition}\emph{(Spectral Decomposition of $\bP$)}
Since $\bP$ is symmetric and positive semidefinite, let
$\bP= \UP \SP \UP^T$ denote its spectral decomposition,
with $\UP \in \R^{n \times d}$ having orthonormal columns
and $\SP \in \R^{d \times d}$ diagonal
with nonincreasing entries
$(\SP)_{1,1}\ge (\SP)_{2,2} \ge \cdots \ge (\SP)_{d,d} > 0$.
\end{definition}

We note that while $\bP = \bX \bX^T$ is not observed,
existing spectral norm bounds
\citep[e.g.,][]{oliveira2009concentration,lu13:_spect}
establish that if $\bA \sim \Bern(\bP)$,
the spectral norm of $\bA-\bP$ is comparatively small.
As a result, we regard $\bA$ as a noisy version of $\bP$,
and we begin our inference procedures with a spectral decomposition of $\bA$.
\begin{definition}\emph{\citep[Adjacency Spectral Embedding;][]{STFP-2011}}
Let $\bA \in \R^{n \times n}$ be the
adjacency matrix of an undirected $d$-dimensional random dot product graph.
The $d$-dimensional \emph{adjacency spectral embedding} (ASE) of $\bA$
is a spectral decomposition of $\bA$ based on its top $d$ eigenvalues,
obtained by
$\ASE(\bA,d) = \UA \SA^{1/2}$, where $\SA \in \R^{d \times d}$ is a diagonal
matrix whose entries are the top eigenvalues of $\bA$ (in nonincreasing order)
and $\UA \in \R^{n \times d}$ is the matrix whose columns are the
orthonormal eigenvectors corresponding to the eigenvalues in $\SA$.
\end{definition}

\begin{remark}\label{remark:nonneg_eig}
We observe that without any additional assumptions, the top $d$ eigenvalues
of $\bA$ are not guaranteed to be nonnegative.
However, under our eigengap assumptions on $\bDelta$,
the i.i.d.-ness of the latent positions ensures that for large $n$,
the eigenvalues of $\bA$ will be nonnegative with high probability
(see Observation~\ref{obs:deltaLB} in the Supplementary Material).
\end{remark}

Given a set of $m$ adjacency matrices distributed as
$$(\bA^{(1)},\bA^{(2)},\dots,\bA^{(m)},\bX) \sim \JRDPG(F,n,m)$$
for distribution $F$ on $\R^d$,
a natural inference task is to recover the $n$ latent positions
$\bX_1,\bX_2,\dots,\bX_n \in \R^d$ shared by the vertices of the $m$ graphs.
To estimate the underlying latent positions from these $m$ graphs, \cite{runze_law_large_graphs} provides justification for the estimate
$\Xbar = \ASE( \Abar, d )$, where $\Abar$ is the sample mean of the
adjacency matrices $\bA^{(1)},\bA^{(2)},\dots,\bA^{(m)}$.
However, $\Xbar$ is ill-suited to any task that requires
comparing latent positions across the $m$ graphs,
since the $\Xbar$ estimate collapses the $m$ graphs into a single
set of $n$ latent positions.
This motivates the \emph{omnibus embedding},
which still yields a single spectral decomposition, but with a separate $d$-dimensional representation for each of the $m$ graphs.
This makes the omnibus embedding useful for {\em simultaneous}
inference across all $m$ observed graphs.
\begin{definition}\emph{(Omnibus embedding)}
Let $\bA^{(1)},\bA^{(2)},\dots,\bA^{(m)} \in \R^{n \times n}$
be (possibly weighted) adjacency matrices
of a collection of $m$ undirected graphs.
We define the $mn$-by-$mn$ omnibus matrix
of $\bA^{(1)}, \bA^{(2)}, \dots, \bA^{(m)}$ by
\begin{equation} \label{eq:omnidef}
\bM =
\begin{bmatrix}
\bA^{(1)} & \frac{1}{2}(\bA^{(1)} + \bA^{(2)}) & \dots & \frac{1}{2}(\bA^{(1)} + \bA^{(m)}) \\
\frac{1}{2}(\bA^{(2)} + \bA^{(1)}) & \bA^{(2)} & \dots & \frac{1}{2}(\bA^{(2)} + \bA^{(m)}) \\
\vdots & \vdots & \ddots & \vdots \\
\frac{1}{2}(\bA^{(m)} + \bA^{(1)}) & \frac{1}{2}(\bA^{(m)} + \bA^{(2)})
& \dots & \bA^{(m)} \end{bmatrix},
\end{equation}
and the $d$-dimensional \emph{omnibus embedding} of
$\bA^{(1)},\bA^{(2)},\dots,\bA^{(m)}$
is the adjacency spectral embedding of $\bM$:
$$ \OMNI(\bA^{(1)},\bA^{(2)},\dots,\bA^{(m)},d)= \ASE( \bM, d ). $$
\end{definition}
If $(\bA^{(1)},\bA^{(2)},\dots,\bA^{(m)},\bX) \sim \JRDPG(F,n,m)$,
then the omnibus embedding provides a natural approach to 
estimating $\bX$ \emph{without} collapsing the $m$ graphs into a single
representation as with $\Xbar = \ASE(\Abar,d)$.
Under the JRDPG, the omnibus matrix has expected value
$$ \E \bM = \Ptilde = \bJ_m \otimes \bP = \UPt \SPt \UPt^T $$
for $\UPt \in \R^{mn \times d}$ having $d$ orthonormal columns
and $\SPt \in \R^{d \times d}$ diagonal.
Since $\bM$ is a reasonable estimate for $\Ptilde = \E \bM$
\citep[see, for example,][]{oliveira2009concentration},
the matrix $\Zhat = \OMNI(\bA^{(1)},\bA^{(2)},\dots,\bA^{(m)},d)$
is a natural estimate of the $mn$ latent positions
collected in the matrix
$\bZ = [\bX^T \bX^T \dots \bX^T]^T \in \R^{mn \times d}$.
Here again, as in Remark~\ref{rem:nonid}, $\Zhat$ only recovers
the true latent positions $\bZ$ up to an orthogonal rotation.
The matrix
\begin{equation} \label{eq:Zstruct}
\Zstar = \begin{bmatrix} \Xstar \\ \Xstar \\ \vdots \\ \Xstar \end{bmatrix}
        = \UPt \SPt^{1/2} \in \R^{mn \times d},
\end{equation}
provides a reasonable canonical choice of latent positions,
so that $\bZ = \Zstar \bW$ for some suitably-chosen orthogonal matrix
$\bW \in \R^{d \times d}$, and our main theorem shows that we can recover
$\bZ$ (up to orthogonal rotation) by recovering $\Zstar$.

\section{Main Results}\label{sec:main_results}

In this section, we give theoretical results on the
consistency and asymptotic distribution of the estimated latent positions based on the omnibus matrix $\bM$.
In the next section,
we demonstrate from simulations that the omnibus embedding can be successfully leveraged for subsequent inference, specifically two-sample testing.

Lemma~\ref{lem:omni2toinf} shows that the omnibus embedding
provides uniformly consistent estimates of the true latent positions,
up to an orthogonal transformation,
roughly analogous to Lemma 5 in \cite{lyzinski13:_perfec}.
Lemma~\ref{lem:omni2toinf}
shows consistency of the omnibus embedding under the $\tti$ norm,
implying that all $mn$ of the estimated latent positions 
are near their corresponding true positions.
We recall that the orthogonal transformation $\Wtilde$
in the statement of the lemma is necessary since,
as discussed in Remark~\ref{rem:nonid},
$\bP = \bX \bX^T = (\bX \bW)(\bX \bW)^T$
for any orthogonal $\bW \in \R^{d \times d}$.
\begin{lemma}\label{lem:omni2toinf}
	With $\Ptilde$, $\bM$, $\UM$, and $\UPt$ defined as above, there exists
	an orthogonal matrix $\Wtilde \in \R^{d \times d}$ such that
	with high probability,
	\begin{equation}\label{eq:omni2toinf_actualbound}
	\|\UM \SM^{1/2}-\UPt \SPt^{1/2} \Wtilde \|_{\tti}
	\le \frac{Cm^{1/2} \log mn }{\sqrt{n}} .
	\end{equation}
\end{lemma}
\begin{proof}
This result is proved in the supplemental material.
\end{proof}

As noted earlier, our central limit theorem for the omnibus embedding is analogous to a similar result proved in \cite{athreya2013limit}, but with the crucial difference that we no longer require that the second moment matrix have distinct eigenvalues. As in \cite{athreya2013limit}, our proof here depends on writing the difference between a row of the omnibus embedding and its  corresponding latent position as a pair of summands: the first,  to which a classical Central Limit Theorem can be applied, and the second, essentially a combination of residual terms, which converges to zero. The weakening of the assumption of distinct eigenvalues necessitates significant changes in how to bound the residual terms. In fact, \cite{athreya2013limit} adapts a result of \cite{bickel_sarkar_2013}---the latter of which depends on the assumption of distinct eigenvalues---to control these terms. Here, we resort to somewhat different methodology: we prove instead that analogous bounds to those in \cite{lyzinski15_HSBM,tang_lse} hold for the estimated latent positions
based on the omnibus matrix $\bM$, and this enables us to establish that here, too, the rows of the omnibus embedding are also approximately normally distributed.  Further, en route to this limiting result, we compute the explicit variance of the omnibus matrix, and show that as $m$, the number of graphs embedded, increases, this contributes to a reduction in the variance of the estimated latent positions.

\begin{theorem} \label{thm:main}
Let $(\bA^{(1)},\bA^{(2)},\dots,\bA^{(m)},\bX) \sim \JRDPG(F,n,m)$ for some
$d$-dimensional inner product distribution $F$ and let $\bM$ denote
the omnibus matrix as in \eqref{eq:omnidef}. Let
$\bZ = \Zstar \bW$ with $\Zstar$ as defined in Equation~\eqref{eq:Zstruct},
with estimate
$\Zhat = \OMNI(\bA^{(1)},\bA^{(2)},\dots,\bA^{(m)},d)$.
Let $h = m(s-1) + i$ for $i \in [n],s \in [m]$, so that $\Zhat_h$
denotes the estimated latent position
of the $i$-th vertex in the $s$-th graph $\bA^{(s)}$.
That is, $\Zhat_h$ is the
column vector formed by transposing the $h$-th row of the matrix
$\Zhat = \UM \SM^{1/2} = \OMNI(\bA^{(1)},\bA^{(2)},\dots,\bA^{(m)},d)$.
Let $\Phi(\bx,\bSigma)$ denote the cdf of a (multivariate)
Gaussian with mean zero and covariance matrix $\bSigma$,
evaluated at $\bx \in \R^d$.
There exists a sequence of orthogonal $d$-by-$d$ matrices
$( \Wntilde )_{n=1}^\infty$ such that for all $\bx \in \R^d$,
$$ \lim_{n \rightarrow \infty}
        \Pr\left[ n^{1/2} \left( \Zhat \Wntilde - \bZ \right)_h
                \le \bx \right]
= \int_{\supp F} \Phi\left(\bx, \bSigma(\by) \right) dF(\by), $$
where
$\bSigma(\by) = (m+3)\bDelta^{-1} \Sigmatilde(\by) \bDelta^{-1}/(4m), $
$\bDelta$ is as defined in \eqref{eq:def:Delta} and
$$\Sigmatilde(\by)
= \E\left[ (\by^T \bX_1 - ( \by^T \bX_1)^2 ) \bX_1 \bX_1^T \right].$$
\end{theorem}
\begin{proof}
This result is proved in the supplemental material.
\end{proof}

\section{Experimental results}
\label{sec:expts}

In this section, we present experiments on synthetic data
exploring the efficacy of the omnibus embedding described above.
We consider both estimation of latent positions and two-sample graph testing.
\subsection{Recovery of Latent Positions}
Perhaps the most ubiquitous estimation problem for RDPG data is that of
estimating the latent positions (i.e., the rows of the matrix $\bX$); consequently, we begin by exploring how well the omnibus embedding recovers the latent
positions of a given random dot product graph.
If one wishes merely to estimate the latent positions $\bX$
of a set of $m$ graphs
$(\bA^{(1)},\bA^{(2)},\dots,\bA^{(m)},\bX) \sim \JRDPG(F,n,m)$,
the estimate $\Xbar = \ASE( \sum_{i=1}^m \bA^{(i)}/m, d )$,
the embedding of the sample mean of the adjacency matrices
performs well asymptotically \citep{runze_law_large_graphs}.
Indeed, all else equal,
the embedding $\Xbar$ is preferable to the omnibus embedding
if only because it requires an eigendecomposition
of an $n$-by-$n$ matrix rather
than the much larger $mn$-by-$mn$ omnibus matrix.

\begin{figure}[t!]
  \centering
    \includegraphics[width=0.6\columnwidth]{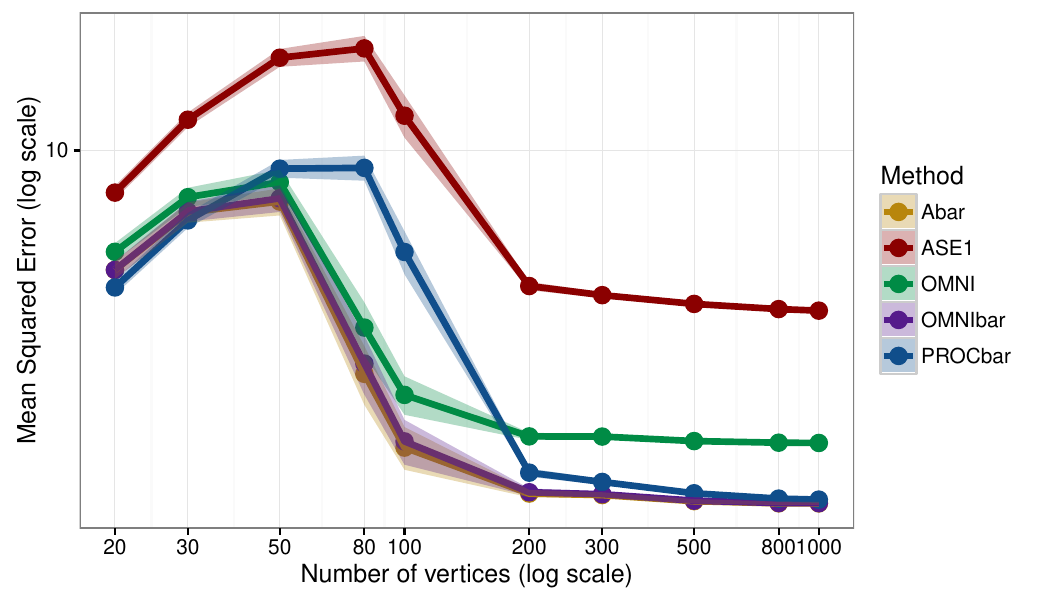}
  \caption{Mean squared error (MSE) in recovery of latent positions (up to rotation) in a 2-graph joint RDPG model as a function of the number of vertices. The figure shows the performance of ASE applied to a single graph (red), ASE embedding of the mean graph (gold), the Procrustes-based pairwise embedding (blue), the omnibus embedding (green) and the mean omnibus embedding (purple). Each point is the mean of 50 trials, with error bars indicating two times the standard error. We see that the mean omnibus embedding (OMNIbar) achieves performance competitive with that of the optimal embedding $\ASE(\Abar,d)$, while the Procrustes alignment estimation is notably inferior to the other two-graph techniques for graphs of size between 80 and 200 vertices (and we note that the gap appears to persist at larger graph sizes, though it shrinks).}
  \label{fig:compareMSE}
\end{figure}

Of course, the omnibus embedding can still be used to
to estimate the latent positions, potentially at the cost of
increased variance.
Figure \ref{fig:compareMSE} compares the mean-squared error of various
techniques for estimating the latent positions for a random dot product graph.
The figure plots the (empirical) mean squared error in recovering the
latent positions of a $3$-dimensional
JRDPG as a function of the number of vertices $n$.
Each point in the plot is the empirical mean of 50 independent trials.
In each trial, the vertex latent positions are drawn i.i.d. from a
Dirichlet with parameter $[1,\,1,\,1]^T \in \R^{3}$.
Having generated a random set of latent positions, we generate two graphs,
$\bA^{(1)},\bA^{(2)} \in \R^{n \times n}$ independently,
based on this set of latent positions.
Thus, we have $(\bA^{(1)},\bA^{(2)},\bX) \sim \JRDPG(F,n,2)$,
where $F = \Dir([1,\,1,\,1]^T)$ is a Dirichlet with parameter
$[1,\,1,\,1]^T \in \R^3$, and $n$ varies.
The lines correspond to
\begin{enumerate}
\item {\bf ASE1 (red)}: we embed only one of the two observed graphs,
        and use only the ASE of that graph to estimate the latent positions
        in $\bX$. That is, we consider $\ASE(\bA^{(1)})$ as our estimate
        of $\bX$, ignoring entirely the information present in
        $\bA^{(2)}$. This condition serves as a baseline for how much
        additional information is provided by the second graph $\bA^{(2)}$.
\item {\bf Abar (gold)}: we embed the average of the two graphs,
        $\Abar = (\bA^{(1)} + \bA^{(2)})/2$ as $\Xhat = \ASE( \Abar, 3 )$.
        As discussed in, for example, \cite{runze_law_large_graphs},
        this is the lowest-variance estimate of the latent positions $\bX$.
\item {\bf OMNI (green)}: We apply the omnibus embedding to obtain
        $\Zhat = \ASE(\bM,3)$,
        where $\bM$ is as in Equation~\eqref{eq:omnidef}.
        We then use only the first $n$ rows of
        $\Zhat \in \R^{2n \times d}$ as our estimate of $\bX$.
        Thus, this embedding takes advantage of the information available
        in both graphs $\bA^{(1)}$ and $\bA^{(2)}$, but does not
        use both graphs equally, since the first rows of $\Zhat$ are based
        primarily on the information contained in $\bA^{(1)}$.
\item {\bf OMNIbar (purple)}: We again apply the omnibus embedding to obtain
        estimated latent positions
        $\Zhat = \ASE(\bM,3)$, but this time we use all available
        information by averaging the first $n$ rows and the second $n$ rows
        of $\Zhat$.
\item {\bf PROCbar (blue)}: We separately embed the graphs
        $\bA^{(1)}$ and $\bA^{(2)}$, obtaining two separate estimates of the
        latent positions in $\R^3$.
        We then align these two sets of estimated latent positions
        via Procrustes alignment, and average the aligned embeddings to obtain
        our final estimate of the latent positions.
\end{enumerate}
First, let us note that ASE applied to a single graph (red)
lags all other methods.
This is expected, since all other methods assessed in
Figure~\ref{fig:compareMSE} use information from both observed graphs
$\bA^{(1)}$ and $\bA^{(2)}$ rather than only $\bA^{(1)}$.
We see that all other methods perform essentially equally well
on graphs of 50 vertices or fewer.
Given the dearth of signal in these smaller graphs,
we do not expect any method to
recover the latent positions accurately.

Crucially, however, we see that the OMNIbar estimate (purple) performs nearly
identically to the Abar estimate (gold), the natural choice among spectral methods for the estimation latent positions
\citep[for more on the efficiency of Abar, see][]{runze_law_large_graphs}.
The Procrustes estimate (in blue)
provides a two-graph analogue of ASE (red):
it combines two ASE estimates via Procrustes alignment,
but does not enforce an {\em a priori} alignment of the estimated latent positions
in the manner of the omnibus embedding does (we discuss this enforced alignment in \ref{sec:conc} as well.)
As predicted by the results in \cite{lyzinski13:_perfec} and \cite{tang14:_semipar},
the Procrustes estimate is competitive with the Abar (gold)
estimate for suitably large graphs.
The OMNI estimate (in green) serves, in a sense, as an in-between
method, in that it uses information available from both graphs,
but in contrast to Procrustes (blue), OMNIbar (purple)
and Abar (gold), it does not make complete use of the information
available in the second graph.
For this reason, it is noteworthy that the OMNI estimate
outperforms the Procrustes estimate for graphs of 80-100 vertices.
That is, for certain graph sizes,
the omnibus estimate appears to more optimally leverage the information in both graphs
than the Procrustes estimate does,
despite the fact that the information in the second graph has comparatively little
influence on the OMNI embedding.

\subsection{Two-graph Hypothesis Testing}

We now turn to the matter of using the omnibus embedding for testing
the semiparametric hypothesis that two observed graphs are drawn from the
same underlying latent positions.
Suppose we have a set of points
$\bX_1,\bX_2,\dots,\bX_n,\bY_1,\bY_2,\dots,\bY_n \in \R^d$.
Let the graph $G_1$ with adjacency matrix $\bA^{(1)}$ have edges distributed
independently as
$ \bA^{(1)}_{ij} \sim \Bern( \bX_i^T \bX_j )$.
Similarly, let $G_2$ have adjacency matrix $\bA^{(2)}$ with edges
distributed independently as
$ \bA^{(2)}_{ij} \sim \Bern( \bY_i^T \bY_j )$.
As discussed previously,
while $\Abar = (\bA^{(1)}+\bA^{(2)})/2$ may be optimal for estimation of latent
positions, it is not clear how to use the embedding $\ASE(\Abar,d)$ to
test the following hypothesis:
\begin{equation} \label{eq:H0}
   H_0 : \bX_i = \bY_i \enspace \forall i \,\in [n].
\end{equation}
On the other hand,
the omnibus embedding provides a natural test of
the null hypothesis \eqref{eq:H0}
by comparing the first $n$ and last $n$ embeddings of the omnibus matrix
$$ \bM = \begin{bmatrix} \bA^{(1)} & (\bA^{(1)} + \bA^{(2)})/2 \\
                        (\bA^{(1)} + \bA^{(2)})/2 & \bA^{(2)}
        \end{bmatrix}. $$
Intuitively, when $H_0$ holds,
the distributional result in Theorem~\ref{thm:main} holds,
and the $i$-th and $(n+i)$-th rows of $\OMNI(\bA^{(1)},\bA^{(2)},d)$
are equidistributed (though they are not independent).
On the other hand, when $H_0$ fails to hold, there exists at least one
$i \in [n]$ for which the $i$-th and $(n+i)$-th rows of $\bM$ are \emph{not}
identically distributed, and thus the corresponding embeddings are
also distributionally distinct.
This suggests a test that compares the first $n$ rows of
$\OMNI(\bA^{(1)},\bA^{(2)},d)$
against the last $n$ rows (see below for details).
Here, we empirically explore the power this test against its
Procrustes-based alternative from \cite{tang14:_semipar}.

Our setup is as follows.
We draw $\bX_1,\bX_2,\dots,\bX_n \in \R^3$ i.i.d. according to a
Dirichlet distribution $F$ with parameter
$\alphavec = [1, 1, 1]^T$.
Assembling these $n$ points into a matrix
$ \bX = [\bX_1 \bX_2 \dots \bX_n]^T \in \R^{n \times 3}, $
we can generate a graph $G_1$ with adjacency matrix $\bA^{(1)}$
with entries
$ \bA^{(1)}_{ij} \sim \Bern( (\bX \bX^T)_{ij} )$.
Thus, $(\bA^{(1)},\bX) \sim \RDPG(F,n)$.
We generate a second graph $G_2$ by first
drawing random points $\bZ_1,\bZ_2,\dots,\bZ_n \iid F$.
Selecting a set of indices $I \subset [n]$ of size $k < n$ uniformly at
random from among all such $\binom{n}{k}$ sets,
we let $G_2$ have latent positions
$$ \bY_i = \begin{cases} \bZ_i & \mbox{ if } i \in I \\
                        \bX_i & \mbox{ otherwise. } \end{cases} $$
Assembling these points into a matrix
$\bY = [\bY_1, \bY_2, \dots, \bY_n]^T \in \R^{n \times 3}, $
we generate graph $G_2$ with adjacency matrix $\bA^{(2)}$
with edges generated independently according to
$ \bA^{(2)}_{ij} \sim \Bern( (\bY \bY^T)_{ij} ).$
The task is then to test the hypothesis
\begin{equation} \label{eq:ptsnull}
 H_0 : \bX = \bY.
\end{equation}
To test this hypothesis, we consider two different tests, one based on
a Procrustes alignment of the adjacency spectral embeddings of $G_1$ and $G_2$
\citep{tang14:_semipar}
and the other based on the omnibus embedding.
Both approaches are based on estimates of the latent positions
of the two graphs.
In both cases we use a test statistic of the form
$ T = \sum_{i=1}^n \| \Xhat_i - \Yhat_i \|_F^2, $
and accept or reject based on a Monte Carlo estimate of the
critical value of $T$ under the null hypothesis,
in which $\bX_i = \bY_i$ for all $i \in [n]$.
In each trial, we use $500$ Monte Carlo iterates to estimate the
distribution of $T$.

We note that in the experiments presented here,
we assume that the latent positions
$\bX_1,\bX_2,\dots,\bX_n$ of graph $G_1$ are known for sampling purposes,
so that the matrix $\bP = \E \bA^{(1)}$ is known exactly, rather than
estimated from the observed adjacency matrix $\bA^{(1)}$.
This allows us to sample from the true null distribution.
As proved in \cite{lyzinski13:_perfec},
the estimated latent positions $\Xhat_1 = \ASE(\bA^{(1)})$
and $\Xhat_2 = \ASE( \bA^{(2)} )$ recover the true latent positions
$\bX_1$ and $\bX_2$ (up to rotation) to arbitrary accuracy
in $(2,\infty)$-norm for suitably large $n$~\citep{lyzinski13:_perfec}.
Without using this known matrix $\bP$, we would require that our matrices
have tens of thousands of vertices before the variance associated with
estimating the latent positions would no longer overwhelm the signal present
in the few altered latent positions.

Three major factors influence the complexity of testing
the null hypothesis in Equation \eqref{eq:ptsnull}:
the number of vertices $n$,
the number of changed latent positions $k = |I|$,
and the distances $\|\bX_i - \bY_i\|_F$ between the latent positions.
The three plots in Figure \ref{fig:trueP:power} illustrate
the first two of these three factors.
These three plots show the power of two different approaches to testing
the null hypothesis \eqref{eq:ptsnull} for different sized graphs
and for different values of $k$, the number of altered latent positions.
In all three conditions, both methods improve as the number of vertices
increases, as expected, especially since we do not require
estimation of the underlying expected matrix $\bP$ for Monte Carlo
estimation of the null distribution of the test statistic.
We see that when only one vertex is changed, neither method has power much above $0.25$.
However, in the case of $k = 5$ and $k = 10$, is it clear that the
omnibus-based test achieves higher power than the Procrustes-based
test, especially in the range of 30 to 250 vertices.

\begin{figure}[t!]
  \centering
   % \subfloat[]{ \includegraphics[width=0.325\columnwidth,natwidth=504,natheight=288]{expts/truePonediff/onediff_power_by_vxs.pdf} }
   % \subfloat[]{ \includegraphics[width=0.325\columnwidth,natwidth=504,natheight=288]{expts/truePfivediff/fivediff_power_by_vxs.pdf} }
   % \subfloat[]{ \includegraphics[width=0.325\columnwidth,natwidth=504,natheight=288]{expts/truePtendiff/tendiff_power_by_vxs.pdf} }
   \subfloat[]{ \includegraphics[width=0.5\columnwidth]{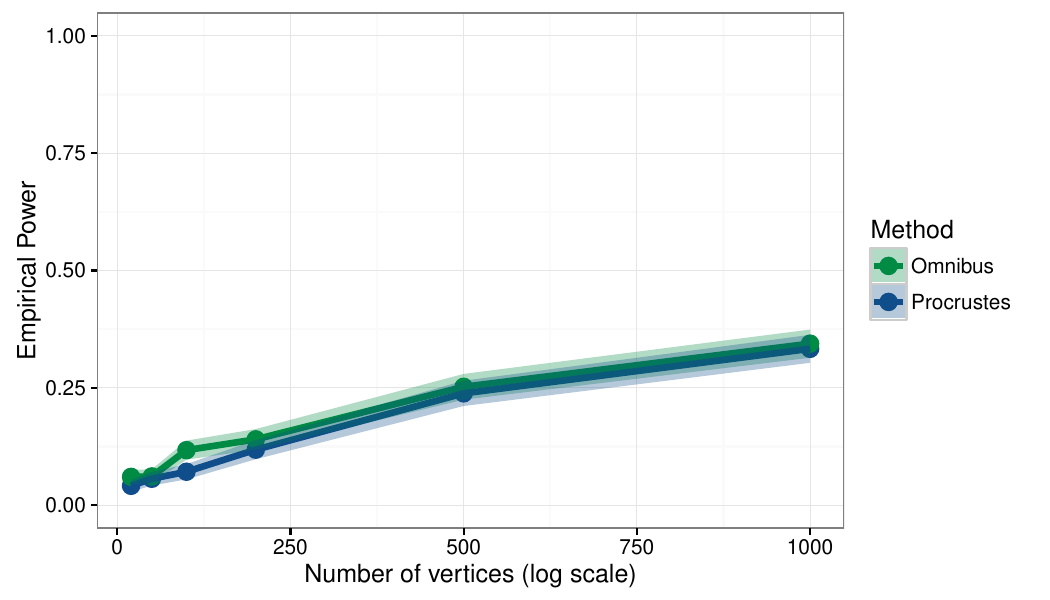} }
   \subfloat[]{ \includegraphics[width=0.5\columnwidth]{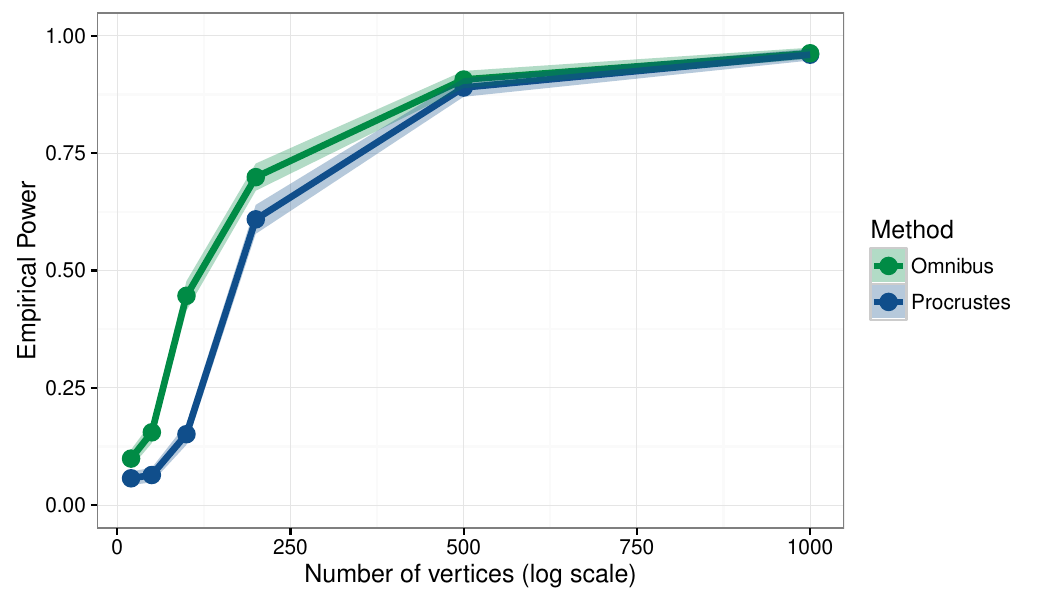} }\\
   \subfloat[]{ \includegraphics[width=0.5\columnwidth]{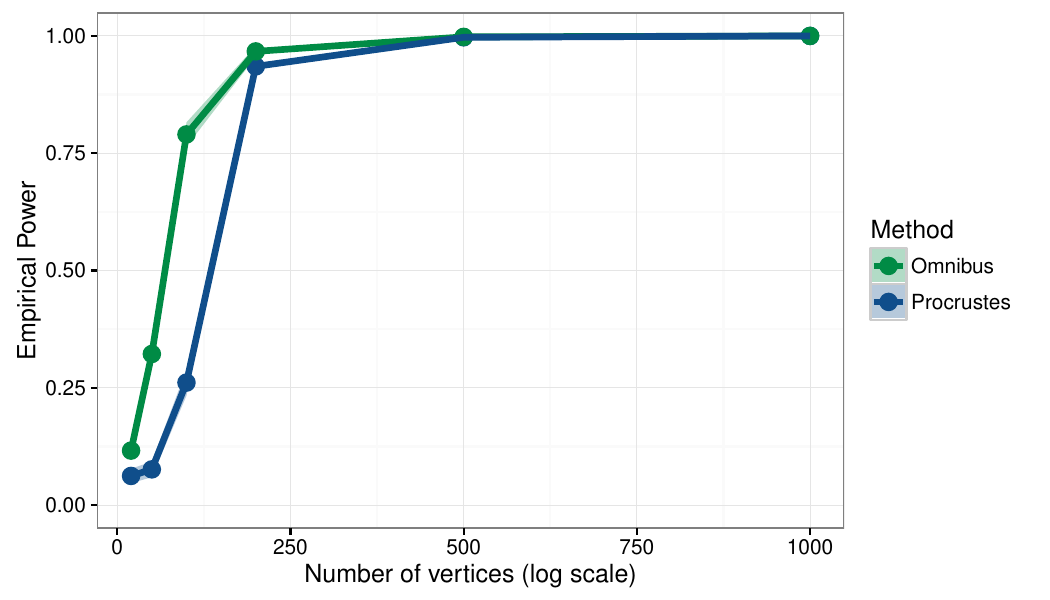} }
  \caption{Power of the ASE-based (blue) and omnibus-based (green)
        tests to detect when the two graphs being testing differ in
        (a) one, (b) five, and (c) ten of their latent positions.
        Each point is the proportion of 1000 trials for which the given
        technique correctly rejected the null hypothesis,
        and error bars denote two standard errors of this empirical mean
        in either direction. }
  \label{fig:trueP:power}
\end{figure}

Figure~\ref{fig:driftpow} shows the effect of the difference between the latent position matrices under null and alternative.
We consider a $3$-dimensional RDPG on $n$ vertices,
in which one latent position, $i \in [n]$,
is fixed to be equal to $\bx_i = (0.8, 0.1, 0.1)^T$
and the remaining latent positions are drawn i.i.d. from a Dirichlet
with parameter $\alphavec = (1,1,1)^T$.
We collect these latent positions in the rows
of the matrix $\bX \in \R^{n \times 3}$.
To produce the latent positions $\bY \in \R^{n \times 3}$ of the second graph,
we use the same latent positions in $\bX$, but we alter the $i$-th position
to be $\bY_i = (1-\lambda)\bx_i + \lambda (0.1,0.1,0.8)^T$ for
$\lambda \in [0,1]$ a ``drift'' parameter, controlling how much the
latent position changes between the two graphs.
Intuitively, correctly rejecting $H_0 : \bX = \bY$ is easier
for larger values of $\lambda$; the greater the gap between latent position matrices under null and alternative, the more easily our test procedure should discriminate between them.
Figure~\ref{fig:driftpow} shows how the size of the drift parameter
influences  the power.
We see that for $n=30$ vertices (top left), neither the omnibus
nor Procrustes test has power appreciably better than
approximately $0.05$, largely in agreement with the
what we observed in Figure~\ref{fig:trueP:power}.
Similarly, when $n=200$ vertices (bottom right),
both methods perform approximately equally
(though omnibus does appear to consistently outperform Procrustes testing).
The case of $n=50$ and $n=100$
vertices (upper right and bottom left, respectively), though, offers a fascinating instance in which the omnibus test
consistently outperforms the Procrustes test.
Particularly interesting to note is the $n=50$ case (top right),
in which we see that performance of the Procrustes test is more or less
flat as a function of drift parameter $\lambda$, while the omnibus
embedding clearly improves as $\lambda$ increases, with performance
climbing well above that of Procrustes for $\lambda > 0.8$.

\begin{figure}[t!]
  \centering
   % \subfloat[]{ \includegraphics[width=0.5\columnwidth,natwidth=504,natheight=288]{expts/trueP_1drift/pow_by_lambda_nvx30.pdf} }
   % \subfloat[]{ \includegraphics[width=0.5\columnwidth,natwidth=504,natheight=288]{expts/trueP_1drift/pow_by_lambda_nvx50.pdf} } \\
   % \subfloat[]{ \includegraphics[width=0.5\columnwidth,natwidth=504,natheight=288]{expts/trueP_1drift/pow_by_lambda_nvx100.pdf} }
   % \subfloat[]{ \includegraphics[width=0.5\columnwidth,natwidth=504,natheight=288]{expts/trueP_1drift/pow_by_lambda_nvx200.pdf} }
   \subfloat[]{ \includegraphics[width=0.5\columnwidth]{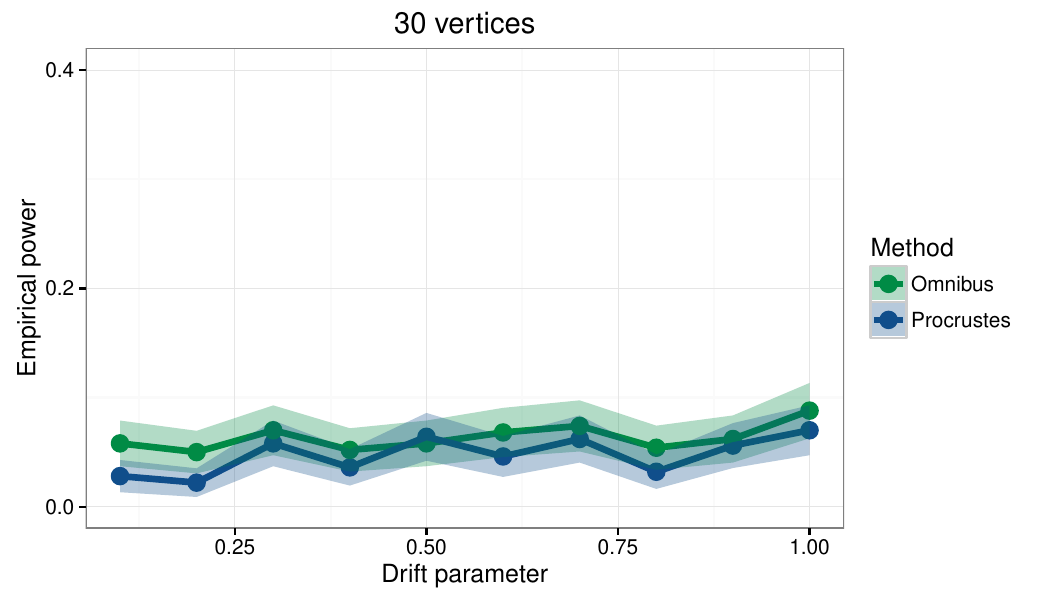} }
   \subfloat[]{ \includegraphics[width=0.5\columnwidth]{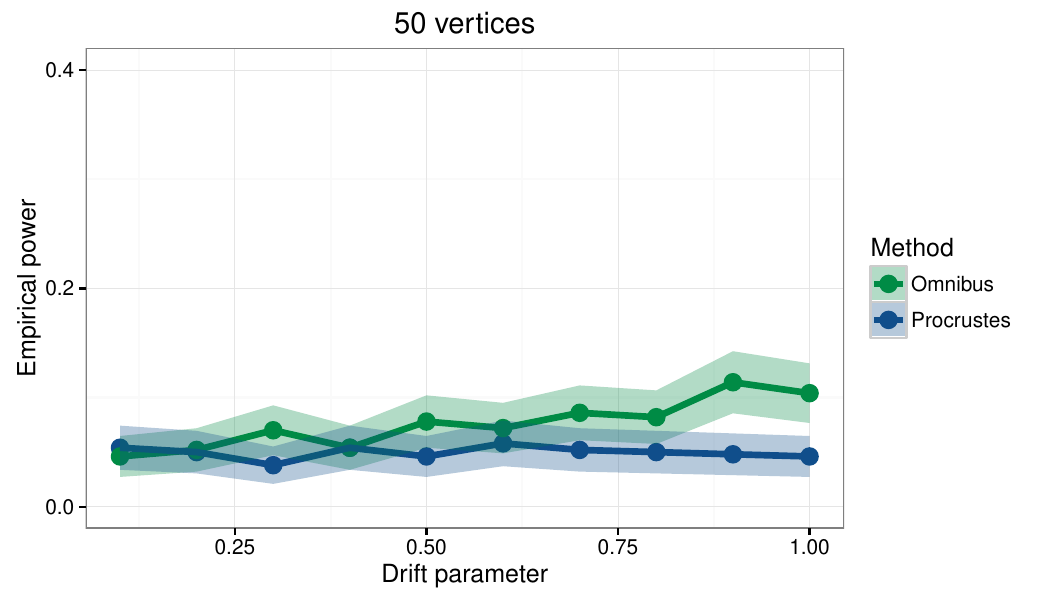} } \\
   \subfloat[]{ \includegraphics[width=0.5\columnwidth]{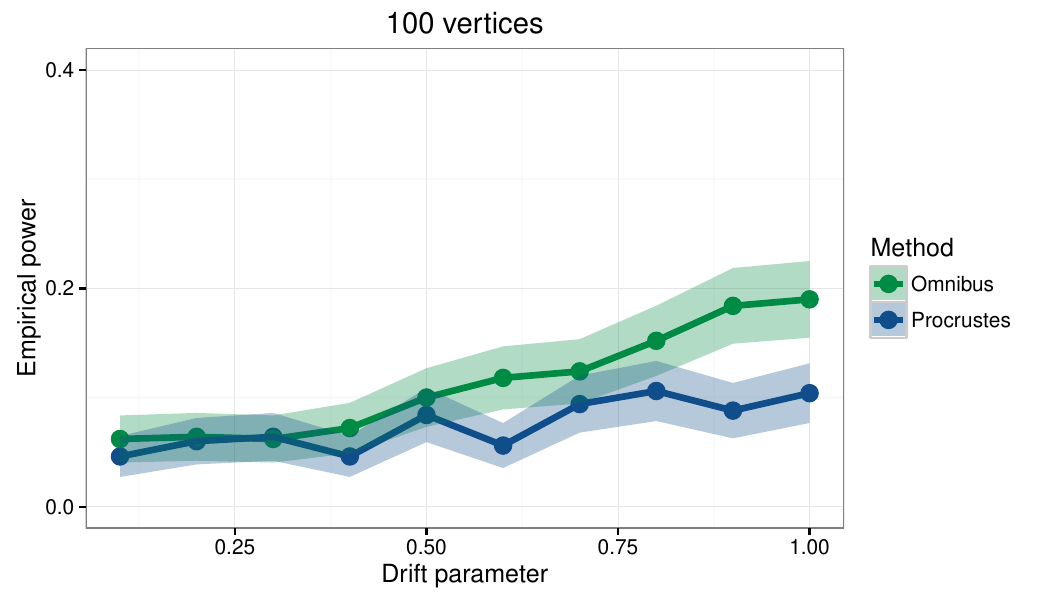} }
   \subfloat[]{ \includegraphics[width=0.5\columnwidth]{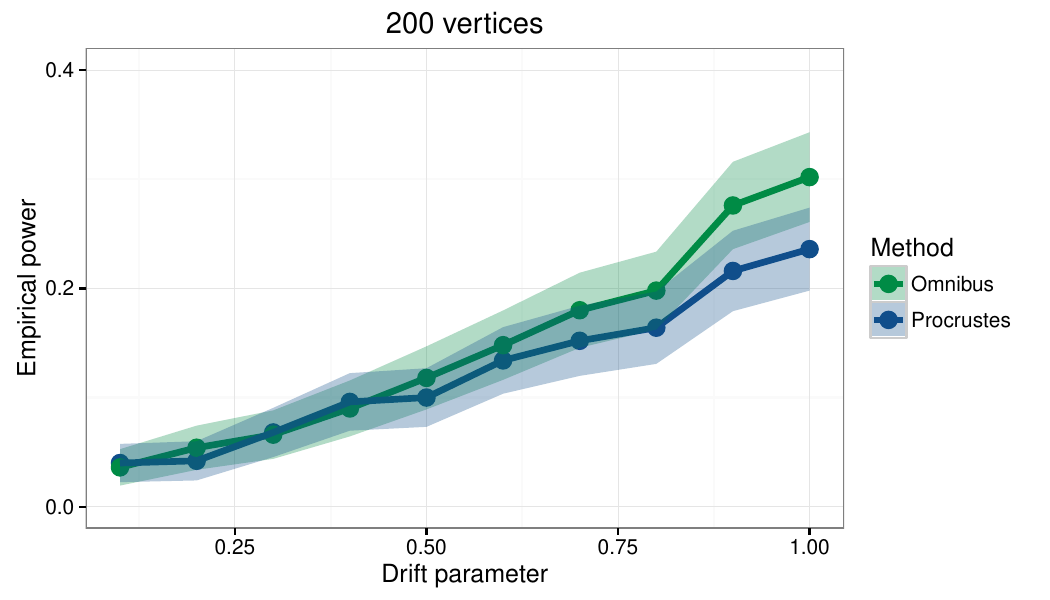} }
  \caption{Power of the ASE-based (blue) and omnibus-based (green)
        tests to detect when the two graphs being testing differ in
        their latent positions.
        Subplots show power as a function of the drift parameter
        $\lambda$ for (a) $n=30$, (b) $n=50$, (c) $n=100$ and (d) $n=200$
        vertices.
        Each point is the proportion of 500 trials for which the given
        technique correctly rejected the null hypothesis,
        and error bars denote two standard errors of this empirical mean.}
  \label{fig:driftpow}
\end{figure}

\section{Discussion and Conclusion}\label{sec:conc}

The omnibus embedding is a simple, scalable procedure for the simultaneous embedding of multiple graphs on the same vertex set, the output of which are multiple points in Euclidean space for each graph vertex. For a wide class of latent position random graphs, this embedding generates accurate estimates of latent positions and supplies empirical power for distinguishing when graphs are statistically different. Our consistency results in the $2 \to \infty$ norm for the omnibus-derived estimates are competitive with state-of-the-art spectral approaches to latent position estimation, and our distributional results for the asymptotic normality of the rows of the omnibus embedding render principled the application of classical Euclidean inference techniques, such as analyses of variance, for the comparison of multiple population of graphs and the identification of drivers of graph similarity or difference at multiple scales, from whole graphs to subcommunities to vertices. We illustrate the utility of the omnibus embedding in data analyses of two different collections of noisy, weighted brain scans, and we uncover new insights into brain regions and vertices that are responsible for graph-level differences in two distinct data sets. 

Further, we quantify  the impact of multiple graphs on the variance of the rows of the embedding, specifically in relation to the variance given in \cite{athreya2013limit}.
This result shows that as the number of graphs, $m$, grows, a significant reduction in the variance is achievable.
Experimental data suggest that the omnibus embedding is
competitive with state-of-the-art, multiple-graph spectral estimation of latent positions, and we surmise that the variance of the rows in the omnibus embedding is close to optimal for latent position estimators derived from the adjacency spectral embedding.
That is, the variance of the omnibus embedding is asymptotically equal to the variance obtained by first averaging the $m$ graphs to get $\Abar$
(which corresponds, in essence, to the maximum likelihood estimate for $\bP$),
and then performing an adjacency spectral embedding of $\Abar$.
Let $\Zhat_i$ correspond to the $i$-th row of
$\Zhat = \OMNI(\bA^{(1)},\bA^{(2)},\dots,\bA^{(m)},d)$, and let
$\Xbar_i$ denote the $i$-th row of the adjacency spectral embedding of $\Abar$.
Let $\bar{\Zhat}_i$ denote the average value of the $m$ rows of $\Zhat$
corresponding to the $i$-th vertex,
namely, the average of the $m$ vectors corresponding to the $i$-th vertex in the omnibus embedding.
We conjecture that averaging the rows of the omnibus embedding accounts for all of the reduction in variance when one compares a single row of the omnibus embedding and a single row of $\Xbar$.
Figure~\ref{fig:compareMSE} provides weak evidence in favor of this
conjecture, since it illustrates that the MSE of both the omnibus-
and Procrustes-based estimates of the latent position estimates
are very close to that of the estimate based on the mean adjacency matrix.
\begin{conjecture}\emph{(Decomposition of Variance)}
With notation as above, for large $n$,
$$ \Var\left( \sqrt{n}\Xbar_i \right)
\approx \Var \left( \sqrt{n}\bar{\Zhat}_i \right)
<\Var\left( \sqrt{n}\Zhat_i \right). $$
\end{conjecture}

We have also demonstrated that the omnibus embedding can be profitably deployed
for two-sample semiparametric hypothesis testing of graph-valued data.
Our omnibus embedding provides a natural mechanism for the
simultaneous embedding of multiple graphs into a single vector space.
This eliminates the need for multiple Procrustes alignments,
which were required in previously-explored
approaches to multiple-graph testing \citep{tang14:_semipar}.
In the two-graph hypothesis testing framework of \cite{tang14:_semipar},
each graph is embedded separately.
Under the assumption of equality of latent positions
(i.e., under $H_0$ in Equation ~\eqref{eq:H0}),
we note that embedding the first graph estimates the true latent positions
$\bX$ up to a unitary transformation in $\R^{d \times d}$.
Call this estimate $\Xhat_1$.
Similarly, $\Xhat_2$, the estimates based on the second graph,
estimates $\bX$ only up to some {\em potentially different} unitary rotation,
i.e., $\Xhat_2 \approx \bX \bW^*$ for some unitary $\bW^*$.
Procrustes alignment is thus required to discover the rotation aligning $\Xhat_1$ with $\Xhat_2$.
In \cite{tang14:_semipar}, it was shown that this Procrustes alignment,
given by
\begin{equation} \label{eq:procmin}
 \min_{\bW \in \calO_d} \| \Xhat_1 - \Xhat_2 \bW \|_F,
\end{equation}
converges under the null hypothesis.
The effect of this Procrustes alignment on subsequent inference
is ill-understood. At the very least, it has the potential to
introduce variance, and our simulations in Section \ref{sec:expts}
suggest that it negatively impacts performance in both estimation
and testing settings.
Furthermore, when the matrix $\bP = \bX \bX^T$
does not have distinct eigenvalues
(i.e., is not uniquely diagonalizable), this Procrustes step is unavoidable,
since the difference $\|\Xhat_1 - \Xhat_2\|_F$ need not converge at all.

In contrast, our omnibus embedding builds an alignment of the graphs
into its very structure. To see this, consider, for simplicity, the $m=2$ case.
Let $\bX \in \R^{n \times d}$ be the matrix whose rows are the latent positions
of both graphs $G_1$ and $G_2$, and let $\bM \in \R^{2n \times 2n}$ be their
omnibus matrix.
Then %conditioning on $X$,
\begin{equation*}
\E \bM = \Ptilde = \begin{bmatrix} \bP & \bP \\ \bP & \bP \end{bmatrix}
        = \begin{bmatrix} \bX \\ \bX \end{bmatrix}
                \begin{bmatrix} \bX \\ \bX \end{bmatrix}^T.
\end{equation*}
Suppose now that we wish to factorize $\Ptilde$ as
$$ \Ptilde = \begin{bmatrix} \bX  \\ \bX \bW^* \end{bmatrix}
        \begin{bmatrix} \bX  \\ \bX \bW^* \end{bmatrix}^T
        = \begin{bmatrix} \bP & \bX (\bW^*)^T \bX^T \\
                        \bX \bW^* \bX^T & \bP \end{bmatrix}. $$
That is, we want to consider graphs $G_1$ and $G_2$ as being generated from
the same latent positions,
but in one case, say, under a  \emph{different} rotation.
This possibility necessitates the Procrustes alignment in the case of
separately-embedded graphs.
In the case of the omnibus matrix,
the structure of the $\Ptilde$ matrix implies that $\bW^*=\bI_d$.
Thus, in contrast to the Procrustes alignment,
the omnibus matrix incorporates an alignment {\em a priori}.
Simulations show that the omnibus embedding
outperforms the Procrustes-based test for equality of latent positions,
especially in the case of moderately-sized graphs.

To further illustrate the utility of this omnibus embedding, consider the case of testing whether three different random dot product graphs have the same generating latent positions. The omnibus embedding gives us a {\em single} canonical representation of all three graphs: Let $\Xhat^O_1$, $\Xhat^O_2$, and $\Xhat^O_3$ be the estimates for the three latent position matrices generated from the omnibus embedding.
To test whether any two of these random graphs have the same generating latent positions, we merely have to compare the Frobenius norms of their differences, as opposed to computing three separate Procrustes alignments.
In the latter case, in effect, we do not have a canonical choice of coordinates in which to compare our graphs simultaneously.

In our analysis of BNU1 data, we ``center" our omnibus matrix, by first considering $\bB^{(i)}=\bA^{(i)}-\bar{\bA}$ and then performing an omnibus embedding on the $\bB^{(i)}$ matrices. While our theorems are written for the uncentered case, the analysis of the centered version proceeds along similiar lines. We find  that in many practical settings, centering meaningfully improves our ability to detect differences across graphs, and we offer the following conjectures as to why.  First, we surmise that centering allows us to better assess covariance structure between estimated latent positions, and thereby improve clustering in a dissimilarity matrix. Second, centering can mitigate the effect of degree heterogeneity across graphs. Third, centering can dampen the potentially noisy impact of common subgraphs, if they exist, to more clearly address graph difference.

Investigating the impact of centering, both for theory and practice, is ongoing, and it is a prominent open problem in the analysis of the omnibus embedding. Of course, other open problems abound, such as an analysis of the omnibus embedding when the $m$ graphs are correlated, are weighted, or are corrupted by occlusion or noise; a closer examination of the impact of the Procrustes alignment on power; the development of an analogue to a Tukey test for determining which graphs differ when we test equality of multiple graphs; the comparative efficiency of the omnibus embedding relative to other spectral estimates; and finally, results for the omnibus embedding under the alternative, when the graph distributions are unequal. The elegance of the omnibus embedding, especially its anchoring in a long and robust history of spectral inference procedures, makes it an ideal point of departure for multiple graph inference, and the richness of the open problems it inspires suggests that the omnibus embedding will remain a key part of the graph statistician's arsenal.
\section*{Acknowledgments}
This research is partly sponsored by the Air Force Research Laboratory and DARPA, under agreement number FA8750-18-2-0035; as well as DARPA, under agreement numbers FA8750-12-2-0303, N66001-14-1-4028 and N66001-15-C-4041. The views and conclusions contained herein are those of the authors and should not be interpreted as necessarily representing the official policies or endorsements, either expressed or implied, of the Air Force Research Laboratory and DARPA, or the U.S. Government. The U.S. Government is authorized to reproduce and distribute reprints for Governmental purposes notwithstanding any copyright notation thereon. The authors also gratefully acknowledge the support of NSF grant DMS-1646108 and NIH grant BRAIN U01-NS108637.

\bibliographystyle{plainnat}
\bibliography{omni_jasa2}

\begin{thebibliography}{40}
\providecommand{\natexlab}[1]{#1}
\providecommand{\url}[1]{\texttt{#1}}
\expandafter\ifx\csname urlstyle\endcsname\relax
  \providecommand{\doi}[1]{doi: #1}\else
  \providecommand{\doi}{doi: \begingroup \urlstyle{rm}\Url}\fi

\bibitem[Aine et~al.(2017)Aine, Bockholt, Bustillo, Cañive, Caprihan,
  Gasparovic, Hanlon, Houck, Jung, Lauriello, Liu, Mayer, Perrone-Bizzozero,
  Posse, Stephen, Turner, Clark, and Calhoun]{AineETAL2017}
C.~J. Aine, H.~J. Bockholt, J.~R. Bustillo, J.~M. Cañive, A.~Caprihan,
  C.~Gasparovic, F.~M. Hanlon, J.~M. Houck, R.~E. Jung, J.~Lauriello, J.~Liu,
  A.~R. Mayer, N.~I. Perrone-Bizzozero, S.~Posse, J.~M. Stephen, J.~A. Turner,
  V.~P. Clark, and Vince~D. Calhoun.
\newblock Multimodal {Neuroimaging} in {Schizophrenia}: {Description} and
  {Dissemination}.
\newblock \emph{Neuroinformatics}, 15\penalty0 (4):\penalty0 343--364, 2017.

\bibitem[Anderson(2003)]{Anderson2003}
T.~W. Anderson.
\newblock \emph{An Introduction to Multivariate Statistical Analysis}.
\newblock Wiley, 3rd edition, 2003.

\bibitem[Athreya et~al.(2016)Athreya, Lyzinski, Marchette, Priebe, Sussman, and
  Tang]{athreya2013limit}
A.~Athreya, V.~Lyzinski, D.~J. Marchette, C.~E. Priebe, D.~L. Sussman, and
  M.~Tang.
\newblock A limit theorem for scaled eigenvectors of random dot product graphs.
\newblock \emph{Sankhya A}, 78:\penalty0 1--18, 2016.

\bibitem[Belkin and Niyogi(2003)]{belkin03:_laplac}
M.~Belkin and P.~Niyogi.
\newblock Laplacian eigenmaps for dimensionality reduction and data
  representation.
\newblock \emph{Neural Computation}, 15:\penalty0 1373--1396, 2003.

\bibitem[Bhatia(1997)]{Bhatia1997}
R.~Bhatia.
\newblock \emph{Matrix Analysis}.
\newblock Springer, 1997.

\bibitem[Bickel and Sarkar(2015)]{bickel_sarkar_2013}
P.~Bickel and P.~Sarkar.
\newblock Role of normalization for spectral clustering in stochastic
  blockmodels.
\newblock \emph{Annals of Statistics}, 43:\penalty0 962--990, 2015.

\bibitem[Chatterjee(2015)]{chatterjee_usvt}
S.~Chatterjee.
\newblock Matrix estimation by universal singular value thresholding.
\newblock \emph{Annals of Statistics}, 43:\penalty0 177--214, 2015.

\bibitem[Cox and Cox(2000)]{cox_MDS}
M.A.A Cox and Trevor Cox.
\newblock \emph{Multidimensional Scaling}.
\newblock CRC Press, 2000.

\bibitem[Davis and Kahan(1970)]{DavKah1970}
C.~Davis and W.~M. Kahan.
\newblock The rotation of eigenvectors by a perturbation.
\newblock \emph{{SIAM} J. Numerical Analysis}, 7\penalty0 (1), March 1970.

\bibitem[Diaconis and Janson(2008)]{diaconis2007graph}
P.~Diaconis and S.~Janson.
\newblock Graph limits and exchangeable random graphs.
\newblock \emph{Rend. Mat. Appl.}, 28:\penalty0 33--61, 2008.

\bibitem[Dryden and Mardia(1998)]{dryden_mardia_shape}
I.~L. Dryden and K.~V. Mardia.
\newblock \emph{Statistical Shape Analysis}.
\newblock Wiley, 1998.

\bibitem[Fox et~al.(2015)Fox, Spreng, Ellamil, Andrews-Hanna, and
  Christoff]{FoxETAL2015}
K.~C. Fox, R.~N. Spreng, M.~Ellamil, J.~R. Andrews-Hanna, and K.~Christoff.
\newblock The wandering brain: meta-analysis of functional neuroimaging studies
  of mind-wandering and related spontaneous thought processes.
\newblock \emph{Neuroimage}, 111:\penalty0 611--621, 2015.

\bibitem[F{\"u}redi and Koml{\'o}s(1981)]{furedi1981eigenvalues}
Z.~F{\"u}redi and J.~Koml{\'o}s.
\newblock The eigenvalues of random symmetric matrices.
\newblock \emph{Combinatorica}, 1\penalty0 (3):\penalty0 233--241, 1981.

\bibitem[Gower(1975)]{gower_procrustes}
J.~C. Gower.
\newblock Generalized procrustes analysis.
\newblock \emph{Psychometrika}, 40:\penalty0 33--51, 1975.

\bibitem[Hoff et~al.(2002)Hoff, Raftery, and Handcock]{Hoff2002}
P.~D. Hoff, A.~E. Raftery, and M.~S. Handcock.
\newblock {Latent space approaches to social network analysis}.
\newblock \emph{Journal of the American Statistical Association}, 97\penalty0
  (460):\penalty0 1090--1098, 2002.

\bibitem[Holland et~al.(1983)Holland, Laskey, and Leinhardt]{holland}
P.~W Holland, K.~B. Laskey, and S.~Leinhardt.
\newblock Stochastic blockmodels: first steps.
\newblock \emph{Social Networks}, 5:\penalty0 109--137, 1983.

\bibitem[Horn and Johnson(1985)]{horn85:_matrix_analy}
R.~Horn and C.~Johnson.
\newblock \emph{Matrix Analysis}.
\newblock Cambridge University Press, 1985.

\bibitem[Hotelling(1931)]{Hotelling1931}
H.~Hotelling.
\newblock The generalization of {Student's} ratio.
\newblock \emph{Annals of Mathematical Statistics}, 2\penalty0 (3):\penalty0
  360--378, 1931.

\bibitem[Karrer and Newman(2011)]{karrer2011stochastic}
B.~Karrer and M.~E.~J. Newman.
\newblock Stochastic blockmodels and community structure in networks.
\newblock \emph{Physical Review E}, 83:\penalty0 016107, 2011.

\bibitem[Kiar(2018)]{kiar_two_truths}
Kiar.
\newblock A high-throughput pipeline identifies robust connectomes but
  troublesome variability.
\newblock bioRxiv preprint at
  \url{https://www.biorxiv.org/content/early/2018/04/24/188706}, 2018.

\bibitem[Lu and Peng(2013)]{lu13:_spect}
L.~Lu and X.~Peng.
\newblock Spectra of edge-independent random graphs.
\newblock \emph{Electronic Journal of Combinatorics}, 20, 2013.

\bibitem[Lyzinski et~al.(2014)Lyzinski, Sussman, Tang, Athreya, and
  Priebe]{lyzinski13:_perfec}
V.~Lyzinski, D.~L. Sussman, M.~Tang, A.~Athreya, and C.~E. Priebe.
\newblock Perfect clustering for stochastic blockmodel graphs via adjacency
  spectral embedding.
\newblock \emph{Electronic Journal of Statistics}, 8:\penalty0 2905--2922,
  2014.

\bibitem[Lyzinski et~al.(2017)Lyzinski, Tang, Athreya, Park, and
  Priebe]{lyzinski15_HSBM}
V.~Lyzinski, M.~Tang, A.~Athreya, Y.~Park, and C.~E. Priebe.
\newblock Community detection and classification in hierarchical stochastic
  blockmodels.
\newblock \emph{IEEE Transactions in Network Science and Engineering},
  4\penalty0 (1):\penalty0 13--26, 2017.

\bibitem[Olhede and Wolfe(2014)]{olhede_wolfe_histogram}
S.~C. Olhede and P.~J. Wolfe.
\newblock Network histograms and universality of block model approximation.
\newblock \emph{Proceedings of the National Academy of Sciences}, 111:\penalty0
  14722--14727, 2014.

\bibitem[Oliveira(2009)]{oliveira2009concentration}
R.~I. Oliveira.
\newblock Concentration of the adjacency matrix and of the {L}aplacian in
  random graphs with independent edges.
\newblock \url{http://arxiv.org/abs/0911.0600}, 2009.

\bibitem[Power et~al.(2011)Power, Cohen, Nelson, Wig, Barnes, Church, Vogel,
  Laumann, Miezin, Schlaggar, et~al.]{PowerETAL2011}
J.~D. Power, A.~L. Cohen, S.~M. Nelson, G.~S. Wig, K.~A. Barnes, J.~A. Church,
  A.~C. Vogel, T.~O. Laumann, F.~M. Miezin, B.~L. Schlaggar, et~al.
\newblock Functional network organization of the human brain.
\newblock \emph{Neuron}, 72\penalty0 (4):\penalty0 665--678, 2011.

\bibitem[Priebe et~al.(To appear)Priebe, Park, Vogelstein, Conroy, Lyzinski,
  Tang, Athreya, Cape, and Bridgeford]{cep_two_truths}
C.~E. Priebe, Youngser Park, Joshua~T. Vogelstein, John~M. Conroy, Vince
  Lyzinski, Minh Tang, Avanti Athreya, Joshua Cape, and Eric Bridgeford.
\newblock On a ‘two truths’ phenomenon in spectral graph clustering.
\newblock \emph{Proceedings of the National Academy of Sciences}, To appear.
\newblock Arxiv preprint at \url{https://arxiv.org/pdf/1808.07801.pdf}.

\bibitem[Smith et~al.(2017)Smith, Asta, and Calder]{asta_cls}
A.~L. Smith, D.~Asta, and C.~A. Calder.
\newblock The geometry of continuous latent space models for network data.
\newblock Arxiv preprint at \url{http://arxiv.org/abs/1712.08641}, 2017.

\bibitem[Sussman et~al.(2012)Sussman, Tang, Fishkind, and Priebe]{STFP-2011}
D.~L. Sussman, M.~Tang, D.~E. Fishkind, and C.~E. Priebe.
\newblock A consistent adjacency spectral embedding for stochastic blockmodel
  graphs.
\newblock \emph{J. Amer. Statist. Assoc.}, 107\penalty0 (499):\penalty0
  1119--1128, 2012.

\bibitem[Tang and Priebe(2018)]{tang_lse}
M.~Tang and C.~E. Priebe.
\newblock Limit theorems for eigenvectors of the normalized laplacian for
  random graphs.
\newblock \emph{Annals of Statistics}, 46\penalty0 (5):\penalty0 2360--2415,
  2018.

\bibitem[Tang et~al.(2013)Tang, Sussman, and Priebe]{tang2012universally}
M.~Tang, D.~L. Sussman, and C.~E. Priebe.
\newblock Universally consistent vertex classification for latent position
  graphs.
\newblock \emph{Ann. Statist.}, 41:\penalty0 1406 -- 1430, 2013.

\bibitem[Tang et~al.(2017{\natexlab{a}})Tang, Athreya, Sussman, Lyzinski, Park,
  and Priebe]{tang14:_semipar}
M.~Tang, A.~Athreya, D.~L. Sussman, V.~Lyzinski, Y.~Park, and C.~E. Priebe.
\newblock A semiparametric two-sample hypothesis testing problem for random dot
  product graphs.
\newblock \emph{Journal of Computational and Graphical Statistics}, 26\penalty0
  (2):\penalty0 344--354, 2017{\natexlab{a}}.

\bibitem[Tang et~al.(2017{\natexlab{b}})Tang, Athreya, Sussman, Lyzinski, and
  Priebe]{tang14:_nonpar}
M.~Tang, A.~Athreya, D.~L. Sussman, V.~Lyzinski, and C.~E. Priebe.
\newblock A nonparametric two-sample hypothesis testing problem for random dot
  product graphs.
\newblock \emph{Bernoulli}, 23:\penalty0 1599--1630, 2017{\natexlab{b}}.

\bibitem[Tang et~al.(2016)Tang, Ketcha, Vogelstein, Priebe, and
  Sussman]{runze_law_large_graphs}
R.~Tang, M.~Ketcha, J.~T. Vogelstein, C.~E. Priebe, and D.~L. Sussman.
\newblock Laws of large graphs.
\newblock arXiv preprint at \url{http://arxiv.org/abs/1609.01672}, 2016.

\bibitem[Tao and Vu(2012)]{tao2012random}
T.~Tao and V.~Vu.
\newblock Random matrices: Universal properties of eigenvectors.
\newblock \emph{Random Matrices: Theory and Applications}, 1\penalty0 (01),
  2012.

\bibitem[Tropp(2015)]{Tropp2015}
J.~A. Tropp.
\newblock An introduction to matrix concentration inequalities.
\newblock \emph{Foundations and Trends in Machine Learning}, 8\penalty0
  (1--2):\penalty0 1 -- 230, 2015.

\bibitem[Whitfield-Gabrieli et~al.(2009)Whitfield-Gabrieli, Thermenos,
  Milanovic, Tsuang, Faraone, McCarley, Shenton, Green, Nieto-Castanon,
  LaViolette, Wojcik, Gabrieli, and Seidman]{WhitfieldGabrieliETAL2009}
S.~Whitfield-Gabrieli, H.~W. Thermenos, S.~Milanovic, M.~T. Tsuang, S.~V.
  Faraone, R.~W. McCarley, M.~E. Shenton, A.~I. Green, A.~Nieto-Castanon,
  P.~LaViolette, J.~Wojcik, J.~D. Gabrieli, and L.~J. Seidman.
\newblock Hyperactivity and hyperconnectivity of the default network in
  schizophrenia and in first-degree relatives of persons with schizophrenia.
\newblock \emph{Proceedings of the National Academy of Sciences}, 106\penalty0
  (4):\penalty0 1279--1284, 2009.

\bibitem[Young and Scheinerman(2007)]{young2007random}
S.~Young and E.~Scheinerman.
\newblock Random dot product graph models for social networks.
\newblock In \emph{Proceedings of the 5th international conference on
  algorithms and models for the web-graph}, pages 138--149, 2007.

\bibitem[Yu et~al.(2015)Yu, Wang, and Samworth]{DK_usefulvariant}
Y.~Yu, T.~Wang, and R.~J. Samworth.
\newblock A useful variant of the {Davis}-{Kahan} theorem for statisticians.
\newblock \emph{Biometrika}, 102:\penalty0 315--323, 2015.

\bibitem[Zhu and Ghodsi(2006)]{zhu06:_autom}
M.~Zhu and A.~Ghodsi.
\newblock Automatic dimensionality selection from the scree plot via the use of
  profile likelihood.
\newblock \emph{Computational Statistics and Data Analysis}, 51:\penalty0
  918--930, 2006.

\end{thebibliography}

\pagebreak

%\bigskip
\begin{center}
{\large\bf SUPPLEMENTARY MATERIAL}
\end{center}

We collect here the technical proofs supporting our main result,
Theorem~\ref{thm:main}.  We consider the (transposed) $h$-th row of the matrix
$$\sqrt{n} \left( \UM \SM^{1/2} \bV^T \Wn - \UPt \SPt^{1/2} \right),$$
where $\bV,\Wn \in \R^{d \times d}$ are orthogonal transformations.
We follow the reasoning of Theorem 18 in~\cite{lyzinski15_HSBM},
decomposing this matrix as
$\sqrt{n} \left( \UM \SM^{1/2} \bV^T \Wn - \UPt \SPt^{1/2} \right)
= \sqrt{n}(\bN + \bH)$,
where $\bN,\bH \in \R^{mn \times d}$.
To prove our central limit theorem, we show that the
(transposed) $h$-th row of $\sqrt{n} \bH$ converges in probability to $\zeromx$
and that the (transposed) $h$-th row of $\sqrt{n} \bN$ converges
in distribution to a mixture of normals.  We note that Lemma \ref{lem:Mclose}, Observation \ref{obs:deltaLB}, Lemma \ref{lem:UclosetoW}, Propositions \ref{prop:unitaryhoeff} and  \ref{prop:hoeffUtransU}, Lemma \ref{lem:approxcommute}, and Lemma \ref{lem:stringent_control_residuals} provide the groundwork for establishing our consistency result, Lemma \ref{lem:omni2toinf}, and,
thereafter, for showing that the $h$-th row of $\sqrt{n} \bH$
converges in probability to $\zeromx$.
Next, Lemma \ref{lem:inlaw} establishes that the $h$-th row of
$\sqrt{n} \bN$ converges
in distribution to a mixture of normals; the proof of Theorem \ref{thm:main} then follows from Slutsky's Theorem.

We begin with a standard matrix concentration inequality,
reproduced from \cite{Tropp2015}.

\begin{theorem}\emph{\citep[Matrix Bernstein;][Theorem 1.6.2]{Tropp2015}}
\label{thm:mxbern}
Consider independent random
Hermitian matrices $\bH^{(1)},\bH^{(2)},\dots,\bH^{(k)} \in \R^{n \times n}$
with
$\E \bH^{(i)} = \zeromx$ and $\|\bH^{(i)}\| \le L$ with probability $1$
for all $i$ for some fixed $L > 0$.
Define $\bH = \sum_{i=1}^k \bH^{(i)}$, and let $v(\bH) = \| \E \bH^2 \|$.
Then for all $t \ge 0$,
$$ \Pr\left[ \| \bH \| \ge t \right]
\le 2n\exp\left\{ \frac{ - t^2/2 }{ v(\bH) + Lt/3 } \right\}. $$
\end{theorem}

We will apply this matrix Bernstein inequality to the omnibus matrix
to obtain a bound on $\|\bM - \Ptilde \|$, from which it will follow by Weyl's
inequality~\citep{horn85:_matrix_analy} that the eigenvalues of $\bM$ are close
to those of $\E \bM$.

\begin{lemma} \label{lem:Mclose}
Let $\bM \in \R^{mn \times mn}$
be the omnibus matrix of $\bA^{(1)},\bA^{(2)},\dots,\bA^{(m)}$, where
$$(\bA^{(1)},\bA^{(2)},\dots,\bA^{(m)},\bX) \sim \JRDPG(F,n,m). $$
Then
$ \|\bM - \E \bM\| \le C m n^{1/2} \log^{1/2} mn \text{ w.h.p. } $
\end{lemma}
\begin{proof}
Condition on some $\bP = \bX \bX^T$, so that
\begin{equation} \label{eq:def:Ptilde}
\E \bM = \Ptilde =
        \begin{bmatrix}
        \bP      & \bP      & \dots  & \bP \\
        \bP      & \bP      & \dots  & \vdots \\
        \vdots   & \vdots   & \ddots & \vdots \\
        \bP      & \dots    & \dots  & \bP \end{bmatrix}
        \in \R^{mn \times mn}.
\end{equation}
We will apply Theorem \ref{thm:mxbern} to $\bM - \E \bM$.
For all $q \in [m]$ and $i,j \in [n]$,
let $\be^{ij} = \be_i\be_j^T + \be_j\be_i^T$ and define block matrix
$\mxE_{q,i,j} \in \R^{mn \times mn}$ with blocks of size $n$-by-$n$ by
$$ \mxE_{q,i,j} =
\begin{bmatrix}
\zeromx & \dots  & \zeromx & \be^{ij}  & \zeromx & \dots & \zeromx \\
\zeromx & \dots  & \zeromx & \be^{ij}  & \zeromx & \dots & \zeromx \\
\vdots    & \ddots &           & \vdots    &           &       & \vdots \\
\zeromx & \dots  & \zeromx & \be^{ij}  & \zeromx & \dots & \zeromx \\
\be^{ij}  & \dots  & \be^{ij}  & 2\be^{ij} & \be^{ij}  & \dots & \be^{ij} \\
\zeromx & \dots  & \zeromx & \be^{ij}  & \zeromx & \dots & 0_n \\
\vdots    & \ddots &           & \vdots    &           &       & \vdots \\
\zeromx & \dots  & \zeromx & \be^{ij}  & \zeromx & \dots & \zeromx
\end{bmatrix}, $$
where the $\be^{ij}$ terms appear in the $q$-th row and $q$-th column.
Using this definition, we have
$$ \bM - \E \bM
= \sum_{q=1}^m \sum_{1 \le i < j \le n}
        \frac{ \bA^{(q)}_{ij} - \bP_{ij} }{2} \mxE_{q,i,j}, $$
which is a sum of $m\binom{n}{2}$ independent zero-mean matrices, with
$ \left\| (\bA^{(q)}_{ij} - \bP_{ij} ) \mxE_{q,i,j}/2 \right\|
  %\le \| \frac{1}{2} \mxE_{q,i,j} \right\|
  %\le \frac{1}{2} \| \mxE_{q,i,j} \|_F =
\le \sqrt{m+1} $
for all $q \in [m]$ and $i,j \in [n]$.

To apply Theorem \ref{thm:mxbern}, it remains to
consider the variance term $v(\bM- \E \bM)$.
We note first that, letting
$\mxD_{ij} = \be_i\be_i^T + \be_j\be_j^T \in \R^{n \times n}$,
we have
$$ \mxE_{q,i,j} \mxE_{q,i,j}
= \begin{bmatrix}
\mxD_{ij} &\dots  &\mxD_{ij} & \mxD_{ij}      &\mxD_{ij} &\dots  &\mxD_{ij} \\
%\mxD_{ij} &\dots  &\mxD_{ij} & \mxD_{ij}      &\mxD_{ij} &\dots  &\mxD_{ij} \\
\vdots    &\ddots &\vdots    & \vdots         &\vdots    &\ddots &\vdots \\
\mxD_{ij} &\dots  &\mxD_{ij} & \mxD_{ij}      &\mxD_{ij} &\dots  &\mxD_{ij} \\
\mxD_{ij} &\dots  &\mxD_{ij} & (m+3)\mxD_{ij} &\mxD_{ij} &\dots  &\mxD_{ij} \\
\mxD_{ij} &\dots  &\mxD_{ij} & \mxD_{ij}      &\mxD_{ij} &\dots  &\mxD_{ij} \\
%\vdots    &\ddots &\vdots    & \vdots         &\vdots    &\ddots &\mxD_{ij} \\
\vdots    &\ddots &\vdots    & \vdots         &\vdots    &\ddots &\vdots \\
\mxD_{ij} &\dots  &\mxD_{ij} & \mxD_{ij}      &\mxD_{ij} &\dots  &\mxD_{ij}
\end{bmatrix}, $$
where the $(m+3)\mxD_{ij}$ term appears in the $q$-th entry on the diagonal.
Using the fact that the maximum row sum is an upper bound on
the spectral norm \citep{horn85:_matrix_analy},
\begin{equation} \label{eq:specbound}
v(\bM-\E \bM) =
\left\|
        \sum_{q=1}^m \sum_{1 \le i < j \le n}
                \frac{ (\bA^{(q)}_{ij} - \bP_{ij})^2 }{ 4 }
                \mxE_{q,i,j} \mxE_{q,i,j}
        \right\|
\le \frac{ (m+1)^2(n-1) }{ 4 }.
\end{equation}
Applying this upper bound on $v(\bM-\E \bM)$ in Theorem~\ref{thm:mxbern},
with $t = 12(m+1)\sqrt{ (n-1) \log mn }$, we obtain
$$ \Pr\left[ \| \bM - \E \bM \| \ge 12(m+1)\sqrt{ (n-1) \log mn } \right]
        \le 2m^{-3}n^{-2}. $$
%which completes the proof.
%for all $m \ge 1, n \ge 1$. This is summably small in $n$,
%whence the Borel-Cantelli lemma yields the desired result for
%this value of $P$.
Integrating over all $\bX$ yields the result.
\end{proof}

\begin{observation} \label{obs:tildeeigs}
Let $\lambda_1 \ge \lambda_2 \ge \dots \ge \lambda_d > 0$ denote
the top $d$ eigenvalues of $\bP = \bX \bX^T$,
and let $\Ptilde$ be as in Equation~\eqref{eq:def:Ptilde}.
Then
$\sigma(\Ptilde) = \{m\lambda_1,m\lambda_2,\dots,m\lambda_d,0,\dots,0\}$.
\end{observation}
\begin{proof}
This is immediate from the structure of $\Ptilde$,
as defined in Equation~\eqref{eq:def:Ptilde}.
\end{proof}

\begin{observation} \label{obs:deltaLB}
Let $F$ be an inner product distribution on $\R^d$
with random vectors $\bX_1,\bX_2,\dots,\bX_n,\bY \iid F$.
With probability at least $1 - d^2/n^2$,
it holds for all $i \in [d]$ that
$ | \lambda_i(\bP) - n\lambda_i(\E \bY \bY^T) | \le 2d\sqrt{ n \log n }. $
Further, we have for all $i \in [d]$,
$ \lambda_i(\Ptilde) \le C nm \delta$ with high probability.
\end{observation}
\begin{proof}
A slightly looser version of this bound appeared in~\cite{athreya2013limit}.
We include a proof of this improved result
for the sake of completeness.

Note that for $1 \le i \le d$, we have
$ \lambda_i(\bP) = \lambda_i(\bX \bX^T) = \lambda_i( \bX^T \bX).$
Hoeffding's inequality applied to
$ (\bX^T \bX - n\E \bY \bY^T)_{ij} = \sum_{t=1}^n (\bX_{ti} \bX_{tj} - \E \bY_i \bY_j) $
yields, for all $i,j \in [d]$,
$$ \Pr\left[ |(\bX^T \bX) - n\E \bY \bY^T|_{ij} \ge 2\sqrt{ n \log n } \right]
        \le \frac{2}{n^2}. $$
A union bound over all $i,j \in [d]$ implies that
$ \|\bX^T \bX - n\E \bY \bY^T \|_F^2 \le 4d^2 n \log n $
with probability at least $1 - 2d^2/n^2$.
Upper bounding the spectral norm by the Frobenius norm,
we have
$ \| \bX^T \bX - n\E \bY \bY^T \| \le 2d \sqrt{ n \log n } $
with probability at least $1 - 2d^2/n^2$,
and Weyl's inequality~\citep{horn85:_matrix_analy} thus implies
$ | \lambda_i( \bP ) - n\lambda_i( \E \bY \bY^T ) | \le 2d \sqrt{ n \log n }$
for all $1 \le i \le d$,
from which Observation~\ref{obs:tildeeigs} and the reverse triangle inequality
yield
$$ \lambda_i(\Ptilde)
  = m\lambda_i(\bP) \ge m \lambda_d(\bP)
  \ge m| n\lambda_d( \E \bY \bY^T ) - 2d \sqrt{n \log n} |
  \ge C mn $$
for suitably large $n$.
\end{proof}

The next several lemmas follow the reasoning in~\cite{lyzinski15_HSBM},
in particular Proposition 16, Lemma 17 and Theorem 18,
and thus they are stated here without proof.
\begin{lemma}\emph{\citep[Adapted from][Prop. 16]{lyzinski15_HSBM}}
\label{lem:UclosetoW}
Let $\Ptilde = \UPt \SPt \UPt^T$ be the eigendecomposition of $\Ptilde$,
where $\UPt \in \R^{mn \times d}$ has orthonormal columns
and $\SPt \in \R^{d \times d}$ is diagonal and invertible.
Let $\SM \in \R^{d \times d}$ be the
diagonal matrix of the top $d$ eigenvalues of $\bM$
and $\UM \in \R^{mn \times d}$ be the matrix with
orthonormal columns containing the top $d$ corresponding eigenvectors,
so that $\UM \SM \UM^T$ is our estimate of $\Ptilde$, as described above.
Let $\Va \bSigma \Vb^T$ be the SVD of $\UPt^T \UM$.
Then
$$ \| \UPt^T \UM - \Va \Vb^T \|_F
	\le \frac{ C \log mn }{ n } \text{ w.h.p. }$$\end{lemma}

It will be helpful to have the following two propositions,
both of which follow from standard applications of
Hoeffding's inequality.
\begin{proposition} \label{prop:unitaryhoeff}
With notation as above,
$$ \| \UPt^T (\bM - \Ptilde) \|_F
   \le C \sqrt{ mn(m + \log mn) } \text{ w.h.p. } $$
\end{proposition}

\begin{proposition} \label{prop:hoeffUtransU}
With notation as above,
$$ \| \UPt^T (\bM-\Ptilde) \UPt \|_F 
\le C \sqrt{ m \log mn } \text{ w.h.p. } $$
\end{proposition}

In what follows, we let $\bV = \Va \Vb^T$, where $\Va$ and $\Vb$ are as
defined in Lemma~\ref{lem:UclosetoW}, i.e.,
$\Va \bSigma \Vb^T$ is the SVD of $\UPt^T \UM$.
The following lemma shows that the matrix $\bV$ ``approximately commutes''
with several diagonal matrices that will be of import in later computations.

\begin{lemma}\emph{\citep[Adapted from][Lemma 17]{lyzinski15_HSBM}}
\label{lem:approxcommute}
Let $\bV = \Va \Vb^T$ be as defined above. Then
\begin{equation} \label{eq:appcomm1}
  \| \bV \SM - \SPt \bV \|_F
  \le C m \log mn \text{ w.h.p. },
\end{equation}
\begin{equation} \label{eq:appcomm2}
  \| \bV \SM^{1/2} - \SPt^{1/2} \bV \|_F
  \le \frac{ C m^{1/2} \log mn }{ n^{1/2} } \text{ w.h.p. }
\end{equation}
and
\begin{equation} \label{eq:appcomm3}
\| \bV \SM^{-1/2} - \SPt^{-1/2} \bV \|_F \le C (mn)^{-3/2} \text{ w.h.p. }
\end{equation}
\end{lemma}
To prove our central limit theorem, we require somewhat more
precise control on certain residual terms,
which we establish in the following key lemma.
\begin{lemma}
	\label{lem:stringent_control_residuals}
	Define
	\begin{equation*} \begin{aligned}
	\bR_1 &= \UPt \UPt^T \UM - \UPt \bV \\
	\bR_2 &= \bV \SM^{1/2} - \SPt^{1/2} \bV\\
	\bR_3 &= \UM - \UPt \UPt^T \UM + \bR_1 = \UM - \UPt \bV.
	\end{aligned} \end{equation*}
	Then the following convergences in probability hold:
	\begin{equation} \label{eq:tozero1}
	\sqrt{n}\left[ (\bM-\Ptilde)\UPt(\bV \SM^{-1/2} - \SPt^{-1/2} \bV) \right]_h
	\inprob \zeromx,
	\end{equation}
	\begin{equation} \label{eq:tozero2}
	\sqrt{n} \left[ \UPt \UPt^T (\bM-\Ptilde) \UPt \bV \SM^{-1/2} \right]_h
	\inprob \zeromx,
	\end{equation}
	\begin{equation} \label{eq:tozero3}
	\sqrt{n} \left[ (\bI - \UPt \UPt^T)(\bM-\Ptilde) \bR_3 \SM^{-1/2} \right]_h
	\inprob \zeromx,
	\end{equation}
	and with high probability,
	$$\|\bR_1\SM^{1/2}+\UPt \bR_2\|_F \le \frac{C m^{1/2} \log mn}{n^{1/2}}. $$
\end{lemma}
\begin{proof}
	We begin by observing that
	$$ \| \bR_1 \SM^{1/2} + \UPt \bR_2 \|_F
	\le \| \bR_1 \|_F \| \SM^{1/2} \| + \| \bR_2 \|_F. $$
	Lemma \ref{lem:UclosetoW} and the trivial upper bound on the eigenvalues
	of $\bM$ ensures that
	$$ \| \bR_1 \|_F \| \SM^{1/2} \|
	\le \frac{ C m^{1/2} \log mn }{ n^{1/2} }
	\text{ w.h.p. }, $$
	Combining this with Equation~\eqref{eq:appcomm2}, we conclude that
	$$ \| \bR_1 \SM^{1/2} + \UPt \bR_2 \|_F
	\le \frac{ C m^{1/2} \log mn }{ n^{1/2} } \text{ w.h.p. } $$

	We will establish \eqref{eq:tozero1}, \eqref{eq:tozero2}
	and \eqref{eq:tozero3} order.
	To see \eqref{eq:tozero1}, observe that
	\begin{equation*}
	\sqrt{n} \| (\bM-\Ptilde)\UPt(\bV \SM^{-1/2} - \SPt^{-1/2} \bV) \|_F
	\le \sqrt{n} \| (\bM-\Ptilde) \UPt \| \| \bV \SM^{-1/2} - \SPt^{-1/2} \bV \|_F,
	\end{equation*}
	and application of Proposition \ref{prop:unitaryhoeff}
	and Lemma~\ref{lem:approxcommute} imply that with high probability
	$$ \sqrt{n} \| (\bM-\Ptilde)\UPt(\bV \SM^{-1/2} - \SPt^{-1/2} \bV) \|_F
	%\le C \sqrt{n} \left( m n^{1/2} \log^{1/2} mn \right) 
	%	\left( (m n)^{-3/2} \right)
	\le C \sqrt{ \frac{ \log mn }{ mn^3 } }, $$
	which goes to $0$ as $n \rightarrow \infty$.
	
	To show the convergence in \eqref{eq:tozero2}, we recall that 
	$\UPt \SPt^{1/2}=\bZ \bW^T$, and observe that
	since the rows of the latent position matrix $\bZ$ are necessarily
	bounded in Euclidean norm by $1$,
	and since the top $d$ eigenvalues of $\Ptilde$ are of order $mn$,
	it follows that
	\begin{equation}\label{eq:UP2toinfty}
	\|\UPt\|_{\tti} \le C (mn)^{-1/2} \text{ w.h.p. }
	\end{equation}
	Next,
	Proposition \ref{prop:hoeffUtransU} and Observation \ref{obs:deltaLB}
	imply that
	\begin{equation*} \begin{aligned}
	\| (\UPt \UPt^T (\bM-\Ptilde) \UPt \bV \SM^{-1/2})_h \|
	&\le \|\UPt\|_{\tti}\| \UPt^T (\bM-\Ptilde) \UPt \| \| \SM^{-1/2} \| \\
	%&= O\left( \frac{ d \log^{1/2} mn }{ n^{1/2} } \right).
	&\le \frac{ C \log^{1/2} mn }{ m^{1/2} n } \text{ w.h.p.},
	\end{aligned} \end{equation*}
which implies~\eqref{eq:tozero2}.
	
	Finally, to establish~\eqref{eq:tozero3},
	we must bound the Euclidean norm of the vector
	\begin{equation} \label{eq:toughguy}
	\left[ (\bI - \UPt \UPt^T)(\bM-\Ptilde) \bR_3 \SM^{-1/2} \right]_h,
	\end{equation}
	where, as defined above, $\bR_3 = \UM - \UPt \bV$.
	Let $\bB_1$ and $\bB_2$ be defined as follows:
	\begin{equation}\label{eq:def_B1_B2}
	\begin{aligned}
	\bB_1&=
	(\bI - \UPt \UPt^T)(\bM-\Ptilde)(\bI-\UPt\UPt^T) \UM \SM^{-1/2} \\
	\bB_2&=(\bI - \UPt \UPt^T)(\bM-\Ptilde)\UPt(\UPt^T \UM - \bV) \SM^{-1/2}
	\end{aligned}
	\end{equation}
	Recalling that $\bR_3=\UM-\UPt \bV$, we have
	\begin{equation*} \begin{aligned}
	(\bI - \UPt \UPt^T)(\bM-\Ptilde) \bR_3 \SM^{-1/2}
	&= (\bI - \UPt \UPt^T)(\bM-\Ptilde)(\UM-\UPt\UPt^T \UM) \SM^{-1/2} \\
	&~~~~~~+ (\bI - \UPt \UPt^T)(\bM-\Ptilde)(\UPt\UPt^T \UM -\UPt \bV) \SM^{-1/2} \\
	&= \bB_1 + \bB_2.
	\end{aligned} \end{equation*}
	We will bound the Euclidean norm of the
	$h$-th row of each of these two matrices on the right-hand side,
	from which a triangle inequality will yield our desired bound
	on the quantity in Equation~\eqref{eq:toughguy}.
	Recall that we use $C$ to denote a positive constant,
	independent of $n$ and $m$, which may change from line to line.
	
	Let us first consider
	$\bB_2 = (\bI - \UPt \UPt^T)(\bM-\Ptilde)\UPt(\UPt^T \UM - \bV) \SM^{-1/2}$.
	We have
	\begin{equation*}
	\| \bB_2 \|_F
	\le \| (\bI - \UPt \UPt^T)(\bM-\Ptilde)\UPt \|
	\| \UPt^T \UM - \bV \|_F \| \SM^{-1/2} \|.
	\end{equation*}
	By submultiplicativity of the spectral norm and Lemma \ref{lem:Mclose},
	$\| (\bI - \UPt \UPt^T)(\bM-\Ptilde)\UPt \| \le C m n^{1/2} \log^{1/2} mn$
	with high probability.
	From Lemma \ref{lem:UclosetoW} and Observation \ref{obs:deltaLB},
	respectively, we have with high probability
	\begin{equation*}
	\| \UPt^T \UM - \bV \|_F \le C n^{-1} \log mn
	\enspace \text{  and  } \enspace
	\| \SM^{-1/2} \| \le C (mn)^{-1/2} 
	\end{equation*}
	Thus, we deduce that with high probability,
	\begin{equation} \label{eq:B2bound}
	\| \bB_2 \|_F \le \frac{ C m^{1/2} \log^{3/2} mn }{ n }
	\end{equation}
	from which it follows that $\| \sqrt{n} \bB_2 \|_F \inprob 0$,
	and hence $\| \sqrt{n} (\bB_2)_h \| \inprob 0$.
	
Turning our attention to $\bB_1$, and recalling that $\UM^T \UM=I$, we note that
	\begin{equation*} \begin{aligned}
	\| (\bB_1)_h \| &=
	\left\| \left[ (\bI - \UPt \UPt^T)(\bM-\Ptilde)(\bI-\UPt\UPt^T)
	\UM \SM^{-1/2} \right]_h \right\| \\
	&= \left\| \left[ (\bI - \UPt \UPt^T)(\bM-\Ptilde)(I-\UPt\UPt^T)
	\UM \UM^T \UM \SM^{-1/2} \right]_h \right\| \\
	&\le  \left\| \UM \SM^{-1/2} \right\|
	\left\| \left[ (\bI - \UPt \UPt^T)(\bM-\Ptilde)(\bI-\UPt\UPt^T)\UM \UM^T
	\right]_h \right\|.
	\end{aligned} \end{equation*}
	Let $\epsilon > 0$ be a constant. We will show that
	\begin{equation} \label{eq:inprobwewant}
	\lim_{n \rightarrow \infty}
	\Pr\left[ \| \sqrt{n}(\bB_1)_h \| > \epsilon \right]
	= 0.
	\end{equation}
	For ease of notation, define
	$$ \bE_1 = (\bI - \UPt \UPt^T)(\bM-\Ptilde)(\bI-\UPt\UPt^T)\UM \UM^T. $$
	We will show that
	\begin{equation} \label{eq:E1inprob}
	\lim_{n \rightarrow \infty}
	\Pr\left[ \sqrt{n} \left\| \left[ \bE_1 \right]_h \right\|
	> n^{1/4} \right]= 0,
	\end{equation}
	which will imply \eqref{eq:inprobwewant}
	since, by Observation~\ref{obs:deltaLB},  
	$$\| \UM \SM^{-1/2} \| \le C(mn)^{-1/2} \text{ w.h.p. } $$
	
	%To show \eqref{eq:E1inprob}, we will need the fact that
	%Frobenius norms of the rows of $E_1$ are equidistributed.
	%Recall that we define the rows of a random matrix $A$ to be exchangeable if the rows of $QAQ^T$ have the same distribution as $A$ for any permutation matrix $Q$.
	We verify the bound in \eqref{eq:E1inprob} by showing that the Frobenius norms of the rows of $\bE_1$ are exchangeable, and thus all of these Frobenius norms have the same expectation. We can then invoke Markov's inequality to bound the probability that the Frobenius norm of any fixed row exceeds a specified threshold. 
	
	To prove exchangeability of the Frobenius norms of the rows of $\bE_1$,  note that if $\bQ \in \R^{mn \times mn}$ is any permutation matrix, right multiplication of an $mn \times mn$ matrix ${\bf G}$ by $\bQ^T$ merely permutes the columns of ${\bf G}$; hence the Frobenius norm of the $i$-th row of $\bQ {\bf G} \bQ^T$ is the same as the Frobenius norm of the $i$-th row of $\bQ {\bf G}$.
	
	For any $n\times n$ matrix real symmetric matrix ${\bf G}$, let $\mathcal{P}_{d}({\bf G})$ denote the projection onto the eigenspace defined by the top $d$ eigenvalues (in magnitude) of ${\bf G}$. Similarly, let $\mathcal{P}_{d}^{\perp}({\bf G})$ denote the projection onto the orthogonal complement of that eigenspace.
	For the matrix $\Ptilde$, for example, the columns of $\UPt$ are a basis for the eigenspace associated to the top $d$ eigenvalues, and the matrix $\UPt \UPt^T$ is the {\em unique} projection operator into this eigenspace.
	
	Observe that 
	$$ \bQ \UPt \UPt^T\bQ^T \bQ\Ptilde \bQ^T=\bQ\UPt \UPt^T\Ptilde \bQ^T, $$
	and thus $\bQ\UPt \UPt^T\bQ^T$ is the projection matrix onto the eigenspace of top $d$ eigenvalues of $\bQ\Ptilde \bQ^T$
	if and only if $\UPt \UPt^T$ is the corresponding projection matrix for $\Ptilde$. Similarly,  $\UM \UM^T$ is the unique projection operator onto the eigenspace defined by the top $d$ eigenvalues (in magnitude) of $\bM$, and $\bQ\UM \UM^T\bQ^T$ is the corresponding projection matrix for $\bQ \bM \bQ^T$. 
	
	Now, for any pair of $n \times n$ matrices $({\bf G}, {\bf H})$, let $\mathcal{L}({\bf G},{\bf H})$ represent the following operator:
	$$\mathcal{L}({\bf G}, {\bf H})=\mathcal{P}_d^{\perp}({\bf G})({\bf G}-{\bf H})\mathcal{P}_d^{\perp}({\bf G})\mathcal{P}_d({\bf H})$$
	We see that $\mathcal{L}(\bM, \Ptilde)=\bE_1$, and by uniqueness of projections, we note that
	\begin{align}\label{eq:E1_underpermutation}
	    &\mathcal{L}(\bQ \bM  \bQ^T, \bQ \Ptilde \bQ^T)\\
	    &=(\bQ (\bI-\UP\UPt^T)\bQ^T\bQ(\bM-\Ptilde)\bQ^T \bQ(\bI-\UPt \UPt^T) \bQ^T \bQ(\UM \UM^T)\bQ^T\\
	    &=\bQ \bE_1 \bQ^T
	\end{align}
	Since we assume that the latent positions for our graphs are i.i.d, the entries of the matrix pair $(\bM, \Ptilde)$ have the same joint distribution as the entries of the pair $(\bQ \bM \bQ^T, \bQ \Ptilde \bQ^T)$.
	Therefore, the entries of the matrix $\mathcal{L}(\bM, \Ptilde)$ have the same distribution as those of $\mathcal{L}(\bQ \bM \bQ^T, \bQ \Ptilde \bQ^T)$.  By Eq.~\eqref{eq:E1_underpermutation}, this implies that
	$\bE_1$ has the same distribution as $\bQ \bE_1 \bQ^T$.  
	Since the Frobenius norm of any row of $\bQ \bE_1 \bQ^T$ is exactly equal to the Frobenius norm of the corresponding to row of $\bQ \bE_1$, we conclude that the Frobenius norms of rows of $\bE_1$ have the same distribution as the Frobenius norms of the rows of $\bQ\bE_1$, thereby establishing that the Frobenius norms of the rows of $\bE_1$ are exchangeable. 
	%Combining this with the exchangeability structure of the matrix
	%$\bM - \Ptilde$, it follows that the Frobenius norms
	%of the rows of $\bE_1$ are equidistributed.
	%This guarantees that that $\UPt \UPt^T$ has exchangeable rows,
	%and the same argument guarantees that the rows of $\UM \UM^T$ are also exchangeable. As a consequence, the rows of $E_1$ are exchangeable.
	%$  (I - \UPt \UPt^T)(M-\Ptilde)(I-\UPt\UPt^T)\UM \UM^T $
	This row-exchangeability for the Frobenius norms of $\bE_1$ implies that each row has the same expectation, and hence
	$ mn \E \| (\bE_1)_h \|^2 = \E\| \bE_1 \|_F^2 $.
	Applying Markov's inequality,
	\begin{equation} \label{eq:markov} \begin{aligned}
	\Pr\left[ \left\|\sqrt{n}\left[ \bE_1 \right]_h \right\| > t \right]
	&\le \frac{ n \E \left\| \left[
		(\bI - \UPt \UPt^T)(\bM-\Ptilde)(\bI-\UPt\UPt^T)\UM \UM^T
		\right]_h \right\|^2 }{ t^2 } \\
	&= \frac{ \E \left\| (\bI - \UPt \UPt^T)(\bM-\Ptilde)(\bI-\UPt\UPt^T)\UM \UM^T
		\right\|_F^2 }{ m t^2 }.
	\end{aligned} \end{equation}
	We will proceed by showing that with high probability,
	\begin{equation} \label{eq:intermed}
	\left\| (\bI - \UPt \UPt^T)(\bM-\Ptilde)(\bI-\UPt\UPt^T)\UM \UM^T
	\right\|_F \le C m \log mn,
	\end{equation}
	whence choosing $t = n^{1/4}$ in \eqref{eq:markov} yields that
	$$ \lim_{n\rightarrow \infty}
	\Pr\left[ \left\|\sqrt{n}\left[ (\bI - \UPt \UPt^T)(\bM-\Ptilde)
	(\bI-\UPt\UPt^T)\UM \UM^T \right]_h \right\| > n^{1/4} \right]
	= 0, $$
	and \eqref{eq:inprobwewant} will follow.
	We have
	\begin{equation*}
	\left\| (\bI - \UPt \UPt^T)(\bM-\Ptilde)(\bI-\UPt\UPt^T)\UM \UM^T
	\right\|_F
	\le \| \bM-\Ptilde \| \| \UM - \UPt\UPt^T \UM \|_F \| \UM \|
	\end{equation*}
	Theorem \ref{thm:mxbern} implies that the first term in this product
	is at most $C mn^{1/2} \log^{1/2} mn$ with high probability,
	and the final term in this product is, trivially, at most $1$.
	To bound the second term, we will follow reasoning similar to that
	in Lemma \ref{lem:approxcommute}, combined with the Davis-Kahan theorem.
	% https://arxiv.org/pdf/1007.1684.pdf
	% Page 32ish.
	The Davis-Kahan Theorem \citep{DavKah1970,Bhatia1997}
	implies that for a suitable constant $C > 0$,
	\begin{equation*}
	\| \UM \UM^T - \UPt\UPt^T \|
	\le \frac{ C \| \bM - \Ptilde \| }{ \lambda_d(\Ptilde) }.
	\end{equation*}
	By Theorem 2 in \cite{DK_usefulvariant},
	there exists orthonormal $\bW \in \R^{d \times d}$ such that
	$$ \| \UM - \UPt \bW \|_F \le C \| \UM \UM^T - \UPt \UPt^T \|_F. $$
	We observe further that the multivariate linear least squares problem
	\begin{equation*}
	\min_{\bT \in \R^{d \times d}} \| \UM - \UPt \bT \|_F^2
	\end{equation*}
	is solved by $\bT = \UPt^T \UM$.
	Thus, combining all of the above,
	\begin{equation*} \begin{aligned}
	\| \UM - \UPt\UPt^T \UM \|_F^2
	&\le \| \UM - \UPt \bW \|_F^2
	\le C \| \UM \UM^T - \UPt \UPt^T \|_F^2 \\
	&\le C \| \UM \UM^T - \UPt \UPt^T \|^2 
	\le \frac{ C \| \bM - \Ptilde \| }{ \lambda_d(\Ptilde) }
	\le \frac{ C \log^{1/2} mn }{ n^{1/2} } \enspace \text{ w.h.p. }
	\end{aligned} \end{equation*}
	Thus, we have
	\begin{equation*} \begin{aligned}
	\left\| (\bI - \UPt \UPt^T)(\bM-\Ptilde)(\bI-\UPt\UPt^T)\UM \UM^T
	\right\|_F
	&\le \| \bM-\Ptilde \| \| \UM - \UPt\UPt^T \UM \| \| \UM \|_F \\
	&\le C m \log mn \text{ w.h.p. },
	\end{aligned} \end{equation*}
	which implies \eqref{eq:intermed}, as required,
	and thus the convergence in \eqref{eq:tozero3} is established,
	completing the proof.
\end{proof}

We are now ready to prove Lemma~\ref{lem:omni2toinf} on the consistency of the omnibus embedding; that is, we can now prove that
there exists
an orthogonal matrix $\Wtilde \in \R^{d \times d}$ such that
with high probability,
\begin{equation*}
\|\UM \SM^{1/2}-\UPt \SPt^{1/2} \Wtilde \|_{\tti}
\le \frac{Cm^{1/2} \log mn }{\sqrt{n}} .
\end{equation*}

\begin{proof}[Proof of Lemma~\ref{lem:omni2toinf}]
	Observe that
	\begin{equation} \label{eq:expansion} \begin{aligned}
	\UM \SM^{1/2} - \UPt \SPt^{1/2} \bV
	%&= (\bM-\Ptilde) \UPt \bV \SM^{-1/2}
	%- \UPt \UPt^T (\bM-\Ptilde) \UPt \SM^{-1/2} \\
	%&~~~~~~+ (\bI - \UPt \UPt^T)(\bM-\Ptilde) \bR_3 \SM^{-1/2}
	%+ \bR_1 \SM^{1/2} + \UPt \bR_2 \\
	&= (\bM-\Ptilde)\UPt \SPt^{-1/2} \bV
	+ (\bM-\Ptilde)\UPt(\bV \SM^{-1/2} - \SPt^{-1/2} \bV) \\
	&~~~~~~- \UPt \UPt^T (\bM-\Ptilde) \UPt \bV \SM^{-1/2} \\
	&~~~~~~+ (\bI - \UPt \UPt^T)(\bM-\Ptilde) \bR_3 \SM^{-1/2}
		+ \bR_1 \SM^{1/2} + \UPt \bR_2.
	\end{aligned} \end{equation}
	
With $\bB_1$ and $\bB_2$ defined in Equation~\eqref{eq:def_B1_B2}, the arguments in the proof of Lemma~\ref{lem:stringent_control_residuals} imply that with high probability
\begin{equation*}
\begin{aligned}
\|(\bM-\Ptilde)\UPt(\bV \SM^{-1/2} - \SPt^{-1/2} \bV)\|
	& \leq Cm^{1/2} n^{-1} \log^{1/2} mn \\
\|\UPt \UPt^T (\bM-\Ptilde) \UPt \bV \SM^{-1/2}\|_F
	& \leq C n^{-1/2} \log^{1/2}mn \\
\|(\bI - \UPt \UPt^T)(\bM-\Ptilde) \bR_3 \SM^{-1/2}\|_F
&\leq \|\bB_1\|_F+\|\bB_2\|_F \\
& \leq C n^{-1/2} m^{1/2} \log mn + C m^{1/2} n^{-1} \log^{3/2}mn
\end{aligned}
\end{equation*}
As a consequence, there exists an orthogonal matrix $\Wtilde$ such that
	$$\|\UM\SM^{1/2}-\UPt \SPt \Wtilde \|_F \leq \|(\bM-\Ptilde)\UPt \SPt^{-1/2}\|_F + \frac{C m^{1/2} \log mn }{\sqrt{n}} \text{ w.h.p.}$$
	From this, we deduce that
	\begin{equation*} %\label{eq:final_2toinf_omni}
	\max_{i}\|(\UM \SM^{1/2}-\UPt \SPt^{1/2} \Wtilde)_i\| \leq \frac{1}{\lambda_d(\Ptilde)}\max_{i}\|((\bM-\Ptilde)\UPt)_i\| + \frac{C m^{1/2} \log mn}{\sqrt{n}}
	\text{ w.h.p. }
	\end{equation*}
Standard application of Hoeffding's inequality
as in Proposition~\ref{prop:unitaryhoeff} shows that with high probability,
$$\max_{i}\|((\bM-\Ptilde)\UPt)_i\| \le C\left(m^{1/2} +\log^{1/2} mn\right). $$
The desired bound follows from Observation \ref{obs:tildeeigs}
applied to $\lambda_d(\Ptilde)$.
\end{proof}

We are now ready to consider the asymptotic distributional behavior of our estimates
of the latent positions.
%Recall that $X = \UP \SP^{1/2}$ is the matrix of latent positions
%(identifiable only up to an orthogonal rotation, since the matrix
%$P = X X^T = (X W)(X W)^T$ for any orthogonal $W$).
By the definition of the JRDPG (Definition \ref{def:JRDPG}),
the latent positions of the expected omnibus matrix
$\E \bM = \Ptilde = \UPt \SPt \UPt^T$
are given by
\begin{equation*}
\Zstar = \begin{bmatrix} \Xstar \\ \Xstar \\ \vdots \\ \Xstar \end{bmatrix}
        = \UPt \SPt^{1/2} \in \R^{mn \times d}.
\end{equation*}
Recall that we denote the matrix of these ``true'' latent positions
by $ \bZ = [ \bX^T, \bX^T, \dots, \bX^T ]^T \in \R^{mn \times d},$
so that $\bZ = \Zstar \bW$ for some suitably-chosen orthogonal matrix $\bW$.

\begin{lemma} \label{lem:inlaw}
Fix some $i \in [n]$ and some $s \in [m]$ and let $h = m(s-1) + i$.
Conditional on $\bX_i = \bx_i \in \R^d$,
there exists a sequence of $d$-by-$d$ orthogonal matrices $\{ \Wn \}$ such that
%$$ n^{1/2} \Wn \left( \bY - \bZ \Wn^T \right)_h
$$ n^{1/2} \Wn^T \left[ (\bM - \Ptilde) \UPt \SPt^{-1/2} \right]_h
        \inlaw \calN( 0, \bSigma(\bx_i) ), $$
where $\bSigma(\bx_i) \in \R^{d \times d}$ is a covariance matrix that
depends on $\bx_i$.
\end{lemma}
\begin{proof}
For each $n=1,2,\dots$, choose orthogonal $\Wn \in \R^{d \times d}$
so that $\bX = \Xstar \Wn$ (and hence $\bZ = \Zstar \Wn$, as well).
At least one such $\Wn$ exists for each value of $n$, since, as discussed
previously, the true latent positions $\bX$ are specified
only up to some rotation $\bX = \UP \SP^{1/2} \bW = \Xstar \bW$.
We have
%\begin{equation*} \begin{aligned}
%n^{1/2} \Wn ( \bY - \bZ \Wn^T )_h
%&= n^{1/2} \Wn \left(
%                (\bM \bZ \SPt^{-1})_h
%                - (\Ptilde \bZ \SPt^{-1})_h \right) \\
%&= n^{1/2} \Wn \left(
%                (\bM \Zstar \Wn \SPt^{-1})_h
%                - (\Ptilde \Zstar \Wn \SPt^{-1})_h \right) \\
%&= \frac{ n^{1/2} \Wn \SP^{-1} \Wn^T }{ m } \left(
%                (\bM \Zstar)_h
%                - (\Ptilde \Zstar)_h \right),
%\end{aligned} \end{equation*}
\begin{equation*} \begin{aligned}
n^{1/2} \Wn^T \left[ (\bM - \Ptilde) \UPt \SPt^{-1/2} \right]_h
&= n^{1/2} \Wn^T
	\left[ \bM \Zstar \SPt^{-1} - \Ptilde \Zstar \SPt^{-1} \right]_h \\
&= n^{1/2} \Wn^T \SPt^{-1} \Wn \left[ \bM \bZ - \Ptilde \bZ \right]_h \\
&= \frac{ n^{1/2} \Wn^T \SP^{-1} \Wn }{ m }
	\left[ \bM \bZ - \Ptilde \bZ \right]_h,
\end{aligned} \end{equation*}
where we have used the fact that $\SPt = m\SP$.

Recalling the structure of $\bZ = \Zstar \Wn$ (see Equation~\eqref{eq:Zstruct})
and recalling that $\bX_j = (\bX_{j \cdot})^T$,
we have
\begin{equation*} \begin{aligned}
n^{1/2} & \Wn^T \left[ (\bM - \Ptilde) \UPt \SPt^{-1/2} \right]_h
= \frac{ n^{1/2} \Wn^T \SP^{-1} \Wn }{ m } \left(
        \sum_{q=1}^m \sum_{j=1}^n
        \left( \frac{ \bA^{(q)}_{ij} + \bA^{(s)}_{ij} }{2} - \bP_{ij} \right)
        \bX_j \right) \\
&= \frac{ n^{1/2} \Wn^T \SP^{-1} \Wn }{ m } \left(
        \sum_{j \neq i} \left( \frac{m+1}{2}(\bA^{(s)}_{ij} - \bP_{ij})
                        + \sum_{q \neq s} \frac{ \bA^{(q)}_{ij} - \bP_{ij} }{2}
                        \right) \bX_j
                \right) \\
   &~~~~~~- n^{1/2} \Wn^T \SP^{-1} \Wn \bP_{ii} \bX_i \\
&= \left( n \Wn^T \SP^{-1} \Wn \right)
	\left[ n^{-1/2}
        \sum_{j \neq i} \left(
                \frac{ (m+1) }{2m}(\bA^{(s)}_{ij} - \bP_{ij})
                + \frac{1}{ m } \sum_{q \neq s}
                        \frac{ \bA^{(q)}_{ij} - \bP_{ij} }{2}
                \right) \bX_j \right] \\
   &~~~~~~- n \Wn^T \SP^{-1} \Wn \frac{ \bP_{ii} \bX_i }{ n^{1/2} }.
\end{aligned} \end{equation*}

Conditioning on $\bX_i = \bx_i \in \R^d$, we first observe that
\begin{equation} \label{eq:xi_inprob}
  \frac{ \bP_{ii} }{ n^{1/2} } \bX_i
  = \frac{ \bx_i^T \bx_i }{ n^{1/2} } x_i \rightarrow 0 \text{ a.s. }
\end{equation}
further, the scaled sum
\begin{equation*} \begin{aligned}
n^{-1/2} \sum_{j \neq i} & \left(
        \frac{ m+1 }{2m}(\bA^{(s)}_{ij} - \bP_{ij})
        + \frac{1}{ m } \sum_{q \neq s}
        \frac{ \bA^{(q)}_{ij} - \bP_{ij} }{2}
\right) \bX_j \\
&= n^{-1/2} \sum_{j \neq i} \left(
        \frac{ (m+1) }{2m}(\bA^{(s)}_{ij} - \bX_j^T \bx_i)
        + \frac{1}{ m } \sum_{q \neq s}
                        \frac{ \bA^{(q)}_{ij} - \bX_j^T \bx_i }{2}
        \right) \bX_j
\end{aligned} \end{equation*}
is a sum of $n-1$ independent $0$-mean random variables,
each with covariance matrix given by
$$  \Sigmatilde(\bx_i) =
\frac{ m+3 }{ 4m }
  \E \left[  \left( \bx_i^T \bX_j - (\bx_i^T \bX_j)^2 \right) \bX_j \bX_j^T \right].
$$
The multivariate central limit theorem thus implies that
\begin{equation} \label{eq:inlaw1}
 n^{-1/2} \sum_{j \neq i} \left(
        \frac{ (m+1) }{ 2m }(\bA^{(s)}_{ij} - \bX_j \bx_i^T)
        + \frac{1}{ m } \sum_{q \neq s}
                        \frac{ \bA^{(q)}_{ij} - \bX_j \bx_i^T }{2}
        \right) \bX_j
        \inlaw \calN( \zeromx, \Sigmatilde(\bx_i) ) .
\end{equation}

By the strong law of large numbers,
$$ \frac{1}{n} \bX^T \bX - \bDelta \rightarrow \zeromx \text{ a.s. } $$
However, we also have
$$ \frac{1}{n} (\Xstar)^T \Xstar - \Wn \bDelta \Wn^T
        = \Wn \left( \frac{1}{n} \bX^T \bX  - \bDelta \right) \Wn^T
        \rightarrow \zeromx \text{ a.s. }, $$
and $ \SP = \SP^{1/2} \UP^T \UP \SP^{1/2} = (\Xstar)^T \Xstar. $
Thus,
$$ \frac{1}{n} \SP - \Wn \bDelta \Wn^T
        \rightarrow \zeromx \text{ a.s. } $$
Since all matrices involved are order $d$, which is fixed in $n$,
the convergences in the preceding three equations
can be thought of either as element-wise or under any matrix norm.
In particular, we have
$ \| \frac{1}{n} \SP - \Wn \bDelta \Wn^T \| \rightarrow 0 $,
whence Weyl's inequality \citep{horn85:_matrix_analy} implies that
the eigenvalues of $\SP/n$ converge to those of $\bDelta$.
Since both $\SP/n$ and $\bDelta$ are diagonal, this implies
that $\SP/n \rightarrow \bDelta$.
We note that in the case where $\bDelta$ has distinct diagonal entries,
this implies that $\Wn \rightarrow I$ as in \cite{athreya2013limit},
though in the case where $\bDelta$ has repeated eigenvalues,
no such convergence is guaranteed.
Thus we have shown that
$n \Wn^T \SPt^{-1} \Wn \rightarrow \bDelta^{-1}$ almost surely.
Combining this fact with~\eqref{eq:inlaw1},
the multivariate version of Slutsky's theorem yields
%$$ n^{1/2} \Wn \left( \bY_h - (\bZ \Wn)_h \right) \inlaw \calN(\zeromx, \bSigma(\bx_i) ), $$
$$ n^{1/2} \Wn^T \left[ (\bM - \Ptilde) \UPt \SPt^{-1/2} \right]_h
	\inlaw \calN\left(\zeromx, \bSigma(\bx_i) \right)
$$
where $\bSigma(\bx_i) = \bDelta^{-1} \Sigmatilde(\bx_i) \bDelta^{-1}$.
Integrating over the possible values of $\bx_i$ with respect to distribution
$F$ completes the proof.
\end{proof}

We are now ready to prove our main result, Theorem~\ref{thm:main}.

\begin{proof}[Proof of Theorem~\ref{thm:main}]
Let $h \in [mn]$, with $h = (m-1)s + i$ for $s,i \in [n]$.
We wish to consider the (transposed) $h$-th row of the matrix
$\sqrt{n} \left( \UM \SM^{1/2} - \UPt \SPt^{1/2} \bV \right)$,
where we recall from Lemma~\ref{lem:approxcommute} that
$\bV = \Va \Vb^T$, where $\Va \bSigma \Vb^T$ is the SVD of $\UPt^T \UM$.
We follow the reasoning of Theorem 18 in~\cite{athreya2013limit},
decomposing this matrix as
$\sqrt{n} \left( \UM \SM^{1/2} - \UPt \SPt^{1/2} \bV \right)
= \sqrt{n}(\bN + \bH)$,
where $\bN,\bH \in \R^{mn \times d}$.
We will show that the (transposed)
$h$-th row of $\sqrt{n} \bH$ converges in probability to $\zeromx$
and, using Lemma~\ref{lem:inlaw},
that the (transposed) $h$-th row of $\sqrt{n} \bN$ converges
in distribution to a mixture of normals.
An application of Slutsky's Theorem yields the desired result.

Recall our earlier definitions of $\bR_1$, $\bR_2$, and $\bR_3$:
\begin{equation*} \begin{aligned}
  \bR_1 &= \UPt \UPt^T \UM - \UPt \bV \\
  \bR_2 &= \bV \SM^{1/2} - \SPt^{1/2} \bV\\
  \bR_3 &=\UM-\UPt \bV=\UM - \UPt \UPt^T\UM + \bR_1
\end{aligned} \end{equation*}
%Adding and subtracting appropriate quantities, we have
%\begin{equation*}
%\UM \SM^{1/2} - \UPt \SPt^{1/2} V
%= (\UM - \UPt\UPt^T \UM) \SM^{1/2} + R_1 \SM^{1/2} + \UPt R_2,
%\end{equation*}
%and, noting that $\UPt \UPt^T \Ptilde = \Ptilde$
%and $\UM \SM^{1/2} = M \UM \SM^{-1/2}$,
%\begin{equation*} \begin{aligned}
%\UM \SM^{1/2} - \UPt \SPt^{1/2} V
%&= (M-\Ptilde) \UM \SM^{-1/2}
%        - \UPt \UPt^T (M-\Ptilde) \UM \SM^{-1/2}
%        + R_1 \SM^{1/2} + \UPt R_2.
%\end{aligned} \end{equation*}
%%Define $ R_3 = \UM - \UPt \UPt^T\UM + R_1$.
As we noted in the proof of Lemma \ref{lem:omni2toinf}, adding and subtracting appropriate quantities, we deduce as in Eq. \eqref{eq:expansion} that
\begin{equation} \label{eq:expansion2} \begin{aligned}
\UM \SM^{1/2} - \UPt \SPt^{1/2} \bV
%&= (\bM-\Ptilde) \UPt \bV \SM^{-1/2}
%        - \UPt \UPt^T (\bM-\Ptilde) \UPt \SM^{-1/2} \\
%        &~~~~~~+ (\bI - \UPt \UPt^T)(\bM-\Ptilde) \bR_3 \SM^{-1/2}
%                        + \bR_1 \SM^{1/2} + \UPt \bR_2 \\
&= (\bM-\Ptilde)\UPt \SPt^{-1/2} \bV
        + (\bM-\Ptilde)\UPt(\bV \SM^{-1/2} - \SPt^{-1/2} \bV) \\
        &~~~~~~- \UPt \UPt^T (\bM-\Ptilde) \UPt \bV \SM^{-1/2} \\
        &~~~~~~+ (\bI - \UPt \UPt^T)(\bM-\Ptilde) \bR_3 \SM^{-1/2}
        	+ \bR_1 \SM^{1/2} + \UPt \bR_2.
\end{aligned} \end{equation}
Applying Lemma~\ref{lem:inlaw} and integrating over $\bX_i$, we have
that there exists a sequence of orthogonal matrices
$\{ \Wn \}_{n=1}^\infty$ such that
$$
\lim_{n \rightarrow \infty}
\Pr\left[ \sqrt{n} \Wn^T [(\bM-\Ptilde)\UPt \SPt^{-1/2} ]_h \le \bx \right]
= \int_{\supp F} \Phi\left(\bx, \bSigma(\by) \right) dF(\by). $$
Now consider $\UM\SM^{1/2}\bV^T - \UPt \SPt^{1/2}$.
From Equation~\eqref{eq:expansion2}, we have
$$ \left( \UM \SM^{1/2} \bV^T - \UPt \SPt^{1/2} \right) \Wn
= (\bM - \Ptilde)\UPt\SPt^{-1/2} \Wn + \bH \bV^T \Wn, $$
where
\begin{equation} \label{eq:Hdef}
\begin{aligned}
\bH &= (\bM-\Ptilde)\UPt(\bV \SM^{-1/2} - \SPt^{-1/2} \bV)
        - \UPt \UPt^T (\bM-\Ptilde) \UPt \bV \SM^{-1/2} \\
        &~~~~~~+ (\bI - \UPt \UPt^T)(\bM-\Ptilde) \bR_3 \SM^{-1/2}
        + \bR_1 \SM^{1/2} + \UPt \bR_2.
\end{aligned} \end{equation}
Since $\bV^T \Wn$ is unitary, it suffices to show that
the $h$-th row of $\bH$, as defined in \eqref{eq:Hdef},
when multiplied by $\sqrt{n}$, goes to $\zeromx$ in probability,
from which Slutsky's Theorem will yield our desired result.
This is precisely the content of Lemma \ref{lem:stringent_control_residuals}, except that we need to establish the following convergence in probability:
%Thus, it will suffice for us to show each of the
%following convergences in probability:
%\begin{equation} \label{eq:tozero1}
%\sqrt{n}\left[ (M-\Ptilde)\UPt(V \SM^{-1/2} - \SPt^{-1/2} V) \right]_h
%        \inprob 0,
%\end{equation}
%\begin{equation} \label{eq:tozero2}
% \sqrt{n} \left[ \UPt \UPt^T (M-\Ptilde) \UPt \SM^{-1/2} \right]_h
%        \inprob 0,
%\end{equation}
%\begin{equation} \label{eq:tozero3}
% \sqrt{n} \left[ (I - \UPt \UPt^T)(M-\Ptilde) R_3 \SM^{-1/2} \right]_h
%        \inprob 0,
%\end{equation}
\begin{equation} \label{eq:tozero4}
\sqrt{n} \left[ \left( \bR_1 \SM^{1/2} + \UPt \bR_2 \right) \right]_h
        \inprob \zeromx.
\end{equation}

We recall that by Lemma \ref{lem:UclosetoW} and Equation~\eqref{eq:UP2toinfty},
\begin{equation*} %\label{eq:better_R1_bound}
\|\bR_1\|_{\tti} \leq \|\UPt\|_{\tti} \|\UPt^T \UM -\bV\|
\le \frac{C \log mn}{n^{3/2}} \text{ w.h.p. }
\end{equation*}
%
%From Lemma \ref{lem:omni2toinf}, we note that with high probability,
%\begin{equation}\label{eq:2toinfbound}
%\|\UM \SM^{1/2}-\UP \SP^{1/2}\bV\|_{\tti}
%\le \frac{Cm^{1/2} \log mn}{\sqrt{n}}.
%\end{equation}
%From Equations \eqref{eq:2toinfbound} and \eqref{eq:UP2toinfty}, we find that 
%$$\| \UM \|_{\tti} \le C(mn)^{-1/2} \text{ w.h.p. } $$
Combining this with Observation~\ref{obs:deltaLB} and
Lemma~\ref{lem:approxcommute} along with Equation~\eqref{eq:UP2toinfty} again,
\begin{equation*} \begin{aligned}
\| (\bR_1 \SM^{1/2} + \UPt \bR_2)_h\| & \le \|\bR_1\|_{\tti} \|\SM^{1/2}\| + \|\UPt\|_{\tti} \|\bR_2\|\\
&\le \frac{C \log mn}{n} \text{ w.h.p. },
\end{aligned}
\end{equation*}
from which the convergence in \eqref{eq:tozero4} follows,
completing the proof.
\end{proof}

\end{document}